\documentclass{article}
\usepackage{arxiv}

\usepackage[utf8]{inputenc} 
\usepackage[T1]{fontenc}    
\usepackage{hyperref}       
\usepackage{url}            
\usepackage{booktabs}       
\usepackage{amsfonts}       
\usepackage{nicefrac}       
\usepackage{microtype}      
\usepackage{graphicx}
\usepackage{enumerate}
\usepackage{animate}
\usepackage{scalerel}
\usepackage{natbib}
\usepackage[title]{appendix}
\usepackage{epstopdf}
\usepackage[flushleft]{threeparttable}
\usepackage{multirow}

\usepackage{amssymb}
\usepackage{booktabs}
\usepackage{mathtools}
\usepackage{dsfont}
\usepackage{makecell}
\usepackage[inline]{enumitem}
\usepackage[amsmath,thmmarks]{ntheorem}
\usepackage{float}
\usepackage{multicol}
\usepackage{tabularray}
\UseTblrLibrary{booktabs}
\usepackage{tikz} 
\usetikzlibrary{calc,backgrounds}
\usetikzlibrary{shapes.geometric, arrows, arrows.meta}
\usetikzlibrary{positioning}
\usetikzlibrary{decorations.pathreplacing,calligraphy}
\usetikzlibrary{fadings}
\usepackage{bbm}
\usepackage{bm}
\usepackage{xcolor}

\usepackage[caption=false,subrefformat=parens,labelformat=parens]{subfig}
\usepackage{algorithmic}
\usepackage{algorithm}
\ifpdf
  \DeclareGraphicsExtensions{.eps,.pdf,.png,.jpg}
\else
  \DeclareGraphicsExtensions{.eps}
\fi

\definecolor{vir1}{HTML}{440154}
\definecolor{vir2}{HTML}{443983}
\definecolor{vir3}{HTML}{31688E}
\definecolor{vir4}{HTML}{21918C}
\definecolor{vir5}{HTML}{35B779}
\definecolor{vir6}{HTML}{90D743}
\definecolor{vir7}{HTML}{FDE725}

\DeclareMathOperator*{\argmax}{\arg\!\max}
\DeclareMathOperator{\diag}{diag}

\newtheorem{theorem}{Theorem}[section]

\newtheorem{proposition}[theorem]{Proposition}

\theoremstyle{plain}

\theoremstyle{nonumberplain}
\theorembodyfont{\normalfont}
\theoremsymbol{\ensuremath{\square}}
\newtheorem{proof}{Proof}

\newcommand{\shortdgpdiagram}[1]{%
\begin{tikzpicture}[
    shorten >=1pt,
    ->,
    draw=black!50,
    every pin edge/.style={<-,shorten <=1pt},
    neuron/.style={circle,minimum size=35pt,inner sep=0pt},
    layerone/.style={neuron,fill=vir1!50},
    layertwo/.style={neuron,fill=vir2!50},
    layerlast/.style={neuron,fill=vir4!50},
    likelihood/.style={neuron,fill=vir6!50},
    innerlayer/.style={neuron,fill=white},
    ghost/.style={circle,minimum size=35pt,inner sep=0pt}
]
    \node[layerone,pin distance=6mm,pin=left:$\mathbf{x}$]
        (l1-0) at (0,0.2)
        {$\mathcal{GP}^{(1)}_{1}$};
    \node[layerone,pin distance=6mm,pin=left:$\mathbf{x}$]
        (l1-1) at (0,-1.25)
        {$\mathcal{GP}^{(2)}_{1}$};
    \node[layerone,pin distance=6mm,pin=left:$\mathbf{x}$]
        (l1-2) at (0,-3.15)
        {$\mathcal{GP}^{(P_1)}_{1}$};

    \node[layertwo] (l2-0) at (2.4,0.2)
        {$\mathcal{GP}^{(1)}_{2}$};
    \node[layertwo] (l2-1) at (2.4,-1.25)
        {$\mathcal{GP}^{(2)}_{2}$};
    \node[layertwo] (l2-2) at (2.4,-3.15)
        {$\mathcal{GP}^{(P_2)}_{2}$};

    \node[innerlayer] (I-0) at (4.8,0) {$\ldots$};
    \node[innerlayer] (I-1) at (4.8,-1.25) {$\ldots$};
    \node[innerlayer] (I-2) at (4.8,-2.95) {$\ldots$};

    \ifnum#1=0\relax
        \node[layerlast,pin distance=8mm,
            pin={[pin edge={->}]right:$\mathbf{F}^{(1)}$}]
            (f-0) at (7.2,0.2) {$\mathcal{GP}^{(1)}_{L}$};
        \node[layerlast,pin distance=8mm,
            pin={[pin edge={->}]right:$\mathbf{F}^{(2)}$}]
            (f-1) at (7.2,-1.25) {$\mathcal{GP}^{(2)}_{L}$};
        \node[layerlast,pin distance=8mm,
            pin={[pin edge={->}]right:$\mathbf{F}^{(Q)}$}]
            (f-2) at (7.2,-3.15) {$\mathcal{GP}^{(Q)}_{L}$};
    \else
        \node[layerlast] (f-0) at (7.2,0.2)
            {$\mathcal{GP}^{(1)}_{L}$};
        \node[layerlast] (f-1) at (7.2,-1.25)
            {$\mathcal{GP}^{(2)}_{L}$};
        \node[layerlast] (f-2) at (7.2,-3.15)
            {$\mathcal{GP}^{(Q)}_{L}$};
    \fi

    \node[ghost,pin distance=6mm,
        pin={[pin edge={draw=none}]left:$\vdots$}]
        at (0,-2.20) {};

    \node[black] at (0,-2.20) {$\vdots$};
    \node[black] at (2.4,-2.20) {$\vdots$};
    \node[black] at (4.8,-2.20) {$\vdots$};
    \node[black] at (7.2,-2.20) {$\vdots$};

    \foreach \i in {0,1,2}
        \foreach \j in {0,1,2}
            \draw[->] (l1-\i) -- (l2-\j);

    \foreach \i in {0,1,2}
        \foreach \j in {0,1,2}
            \draw[->,path fading=east] (l2-\i) -- (I-\j);
    \foreach \i in {0,1,2}
        \foreach \j in {0,1,2}
            \draw[->,path fading=west] (I-\i) -- (f-\j);

    \ifnum#1=0\relax
        \path (f-1) -- (f-2)
            node[black,midway,sloped,transform canvas={xshift=13.5mm}]
            {$\dots$};
    \else
        \node[likelihood,pin={[pin edge={->}]right:$\mathbf{Y}$}]
            (like) at (9.7,-1.25) {$\mathcal{L}$};
        \draw[->] (f-0) --
            node[black,above,sloped,pos=0.42] {$\mathbf{F}^{(1)}$}
            (like);
        \draw[->] (f-1) --
            node[black,above,pos=0.42] {$\mathbf{F}^{(2)}$}
            (like);
        \draw[->] (f-2) --
            node[black,below,sloped,pos=0.42] {$\mathbf{F}^{(Q)}$}
            (like);
    \fi
\end{tikzpicture}%
}

\makeatother

\title{Generalized and Scalable Deep Gaussian Process Emulation}

\author{
  Deyu Ming\thanks{Corresponding author: \texttt{deyu.ming@ucl.ac.uk}.} \\
  School of Management\\
  University College London, UK \\
  \AND
  Daniel Williamson \\
  Land Environment Economics and Policy Institute \\
  University of Exeter, UK \\
}

\begin{document}
\maketitle

\begin{abstract}
Gaussian process (GP) emulators have become essential tools for approximating complex simulators, significantly reducing computational demands in optimization, sensitivity analysis, and model calibration. While traditional GP emulators effectively model continuous and Gaussian-distributed simulator outputs with homogeneous variability, they typically struggle with discrete, heteroskedastic Gaussian, or non-Gaussian data, limiting their applicability to increasingly common stochastic simulators. {A scalable Generalized Deep Gaussian Process (GDGP) emulation framework is introduced} to accommodate simulators with heteroskedastic Gaussian outputs and a wide range of non-Gaussian response distributions, including Poisson, negative binomial, and categorical distributions. The GDGP framework leverages the expressiveness of DGPs and extends them to latent GP structures, enabling it to capture the complex, non-stationary behavior inherent in many simulators while also modeling non-Gaussian simulator outputs. {GDGP is made scalable} by incorporating the Vecchia approximation for settings with a large number of input locations, while also developing efficient inference procedures for handling large numbers of replicates. In particular, {methodological developments are presented} that further enhance the computation of the approach for heteroskedastic Gaussian responses. {A series of synthetic and empirical examples demonstrates} that these extensions deliver the practical application of GDGP emulators and a unified methodology capable of addressing diverse modeling challenges. The proposed GDGP framework is implemented in the open-source \texttt{R} package \texttt{dgpsi}.
\end{abstract}


\section{Introduction}
Gaussian process (GP) emulators are widely used statistical surrogates for computationally expensive computer simulators that enable tasks such as optimization, sensitivity analysis and calibration to be performed without embedding the simulator directly. Although many simulators are deterministic, GPs express uncertainty in their outputs by assigning any finite collection of outputs a multivariate Gaussian distribution, with correlation controlled by a kernel function. However, many simulators have discrete outputs or other forms where Gaussian models are not appropriate, and the approach is generally to use a latent GP surface with a link function to transform to the outcomes in the likelihood. Examples from the literature have included binary responses~\citep{chang2016ice}, count modeling~\citep{salter2025poisson} and strictly-non negative functions~\citep{spiller2023zero}. Given the increasing use of stochastic simulators, such as agent-based models, \citet{baker2022analyzing} reviewed state-of-the-art GP-based methods for analyzing such simulators. However, most existing methods~\citep{goldberg1997regression,binois2018practical}, together with recent developments and applications~\citep{cole2022large,murph2024sensitivity,yi2024stochastic, patil2025vecchia}, focus on heteroskedastic GP emulation, in which simulator outputs are treated as continuous observations with heteroskedastic Gaussian noise.

Non-stationarity in GP emulation has been addressed separately in the literature through more flexible constructions such as deep Gaussian processes (DGPs; \citealp{salimbeni2017doubly, sauer2023active, ming2021deep}), treed GPs~\citep{gramacy2008bayesian}, and related extensions. However, these developments have largely been pursued independently of the ``non-Gaussian'' simulator literature and have mostly focused on deterministic or continuous-response settings. To the best of our knowledge, only very recent work~\citep{cooper2026modernizing} has started to examine the combination of DGPs with binary outputs in a fully Bayesian framework. More generally, the existing literature has tended to produce model-specific methods designed for particular response types, rather than a unified approach applicable across a broad range of output distributions. This leaves an important gap in the emulation literature. There is currently no general framework that combines the flexibility of DGP emulators with the ability to handle diverse non-Gaussian and heteroskedastic Gaussian outputs. In this work, we propose a Generalized Deep Gaussian Process (GDGP) emulation framework for stochastic simulators whose outputs may follow a wide class of non-Gaussian distributions, including Poisson, negative binomial, and categorical distributions, among others. By exploiting the flexibility of DGP emulators, the proposed framework not only accommodates different response distributions, but also captures non-stationary behaviors when present.

From a computational perspective, our work further develops GDGP by enhancing its scalability. {Two broad families of scalable GP inference particularly relevant to the present work are sparse inducing-point approximations~\citep{titsias2009variational} and sparse nearest-neighbor conditional approximations, including the Vecchia approximation~\citep{vecchia1988estimation,katzfuss2020vecchia,katzfuss2021general} and nearest-neighbor GPs (NNGPs)~\citep{datta2016hierarchical}. Inducing-point methods were extended to DGPs in the variational framework of \citet{salimbeni2017doubly}, while Vecchia-based inference was shown by \citet{sauer2023vecchia} to achieve strong performance for DGP emulation in a fully Bayesian setting.} Motivated by the computational advantages of Stochastic Imputation (SI; \citealp{ming2021deep}) as an approximation to fully Bayesian DGP inference, we derive Vecchia-based approximations for SI to enable efficient inference of GDGP for problems with large numbers of input locations, and further show that SI can efficiently accommodate large numbers of replicates. In addition, for the heteroskedastic Gaussian case, we derive a suite of closed-form expressions that provide further computational savings.

{In summary, this work extends the capabilities of the SI-based DGP emulation framework of \citet{ming2021deep} in three main ways:
\begin{itemize}
\item We propose the GDGP hierarchy, which introduces a likelihood layer into the DGP structure and thereby extends DGP emulation beyond deterministic and homogeneous Gaussian settings to stochastic simulators with heteroskedastic Gaussian and non-Gaussian outputs.
\item We generalize the SI inference framework from DGPs to GDGPs by adapting the training and prediction procedures to accommodate the likelihood layer.
\item We enhance the computational scalability of SI by incorporating the Vecchia approximation and developing an efficient treatment of replicated outputs, enabling GDGP to accommodate both large numbers of unique input locations and repeated simulator runs.
\end{itemize}}

The remainder of the manuscript is organized as follows. {Section~\ref{sec:gdgp} introduces GDGPs, beginning with a brief review of DGPs and followed by the GDGP formulation and SI-based inference for prediction and training.} Section~\ref{sec:scalability} presents scalability extensions of GDGP for settings with large numbers of input locations and replicates, together with further methodological developments for the heteroskedastic Gaussian case. Section~\ref{sec:experiments} reports a series of experiments covering heteroskedastic Gaussian, categorical, and count distributions. Finally, Section~\ref{sec:conclusion} concludes the manuscript.

\section{Generalized Deep Gaussian Processes}
\label{sec:gdgp}

\subsection{{Review of Deep Gaussian Processes}}
\label{sec:dgp}
In this {subsection}, we review the DGP formulation of \citet{ming2021deep}, since its inference extends naturally to the GDGP framework presented in the remainder of the section. A DGP, whose hierarchical architecture is shown in Figure~\subref*{fig:figure_1a}, is defined as an $L$-layer feed-forward network of stationary GPs. The model takes $\mathbf{x}\in\mathbb{R}^{N\times D}$, consisting of $N$ input points in $D$ dimensions, and produces $\mathbf{F}\in\mathbb{R}^{N\times Q}$ as the output, where the $p$-th column $\mathbf{F}^{(p)}$ denotes the $p$-th output of the network for $p=1,\dots,Q$.

\begin{figure}
\centering 
\subfloat[DGP hierarchy]{\label{fig:figure_1a}
    \scalebox{0.63}{\shortdgpdiagram{0}}
}
\hspace{0.05\textwidth}%
\subfloat[GDGP hierarchy]{\label{fig:figure_1b}
    \scalebox{0.63}{\shortdgpdiagram{1}}
}
\caption{{\emph{(a)} The hierarchy of the DGP model, which produces $\mathbf{F}=(\mathbf{F}^{(1)},\dots,\mathbf{F}^{(Q)})$ from the input $\mathbf{x}$. \emph{(b)} The hierarchy of the GDGP model, which augments the DGP with a likelihood layer $\mathcal{L}$, through which $\mathbf{F}$ parameterizes the conditional distribution of the observed output $\mathbf{Y}$.}}
\end{figure}

Let $\{\mathcal{GP}^{(p)}_{l}\}_{p=1,\dots,P_l,l=1,\dots,L}$ with $P_L=Q$ be the stationary GPs that form the DGP and $\mathbf{W}^{(p)}_l\in\mathbb{R}^{N\times 1}$ be the output of $\mathcal{GP}^{(p)}_{l}$ (with $\mathbf{W}^{(p)}_L=\mathbf{F}^{(p)}$ for $p=1,\dots,Q$), then $\{\mathbf{W}^{(p)}_l\}_{p=1,\dots,P_l}$ at the given layer $l$ are assumed independent and multivariate normally distributed:
\begin{equation}
\label{eq:gp}
\mathbf{W}^{(p)}_l|\mathbf{W}_{l-1}\overset{\text{ind}}{\sim}\mathcal{N}\left(\mathbf{0},\,\boldsymbol{\Sigma}^{(p)}_l(\mathbf{W}_{l-1})\right),
\end{equation}
for all $p=1,\dots,P_l$, where $\mathbf{W}_{l-1}=(\mathbf{W}^{(1)}_{l-1},\dots,\mathbf{W}^{(P_{l-1})}_{l-1})\in\mathbb{R}^{N\times P_{l-1}}$ with $\mathbf{W}_{0}=\mathbf{x}$, and $\boldsymbol{\Sigma}^{(p)}_l(\mathbf{W}_{l-1})=(\sigma^{(p)}_l)^2\mathbf{R}^{(p)}_l(\mathbf{W}_{l-1})\in\mathbb{R}^{N\times N}$ is the covariance matrix with $\mathbf{R}^{(p)}_l(\mathbf{W}_{l-1})$ being the correlation matrix. The $ij$-th element of $\mathbf{R}^{(p)}_l(\mathbf{W}_{l-1})$ is specified by $
k_l^{(p)}(\mathbf{W}_{l-1,i*},\,\mathbf{W}_{l-1,j*})+\eta\mathds{1}_{\{i=j\}}$, where $k_l^{(p)}(\cdot,\cdot)$ is a known kernel function with $\eta$ being the nugget term and $\mathds{1}_{\{\cdot\}}$ being the indicator function. The kernel function $k_l^{(p)}(\cdot,\cdot)$ is specified in the following multiplicative form:
\begin{equation*}
k_l^{(p)}(\mathbf{W}_{l-1,i*},\,\mathbf{W}_{l-1,j*})=\prod_{d=1}^{P_{l-1}} k_{l,d}^{(p)}(W_{l-1,id},\,W_{l-1,jd})\,,
\end{equation*}
where $k_{l,d}^{(p)}(\cdot,\cdot)$ is a one-dimensional kernel function for the $d$-th dimension of input $\mathbf{W}_{l-1}$ to $\mathcal{GP}^{(p)}_{l}$. Two common choices of $k_{l,d}^{(p)}(\cdot,\cdot)$ include squared exponential and Matérn kernels. For notational convenience, in the remainder of the manuscript we write $\{\mathcal{GP}^{(p)}_{l}\}$ and $\{\mathbf{W}^{(p)}_l\}$ as shorthand for $\{\mathcal{GP}^{(p)}_{l}\}_{p=1,\dots,P_l,l=1,\dots,L}$ and $\{\mathbf{W}^{(p)}_l\}_{p=1,\dots,P_l,l=1,\dots,L-1}$, respectively.

\subsection{Generalized Deep Gaussian Process Formulation}
The DGP formulation reviewed in Section~\ref{sec:dgp} is capable of emulating non-stationary computer models with deterministic responses, as well as stochastic responses with homogeneous Gaussian behavior. To generalize this framework to stochastic computer models with heteroskedastic or non-Gaussian responses, we draw on the idea underlying generalized linear models (GLMs) by introducing a likelihood layer at the top of the DGP hierarchy (see Figure~\subref*{fig:figure_1b}), thereby explicitly modeling the stochasticity and distributional form of the outputs.

Let $\mathbf{Y}=(Y_1,\dots,Y_N)^\top$ denote a vector of $N$ outputs with corresponding inputs $\mathbf{x}$. We assume $Y_i$, for $i=1,\dots,N$, are conditionally independent with probability density function (PDF) $p(y_i \mid \boldsymbol{\phi}_i)$, given a vector of $Q$ distributional parameters $\boldsymbol{\phi}_i=(\phi_{i1},\dots,\phi_{iQ})$. {Let $\mathbf{g}^{-1}$ denote a known inverse link function that maps the DGP output $\mathbf{F}_i$ at input $\mathbf{x}_i$ to the corresponding distributional parameters $\boldsymbol{\phi}_i$ through
\begin{equation*}
\boldsymbol{\phi}_i=\mathbf{g}^{-1}(\mathbf{F}_i)\,,\quad i=1,\dots,N\,,
\end{equation*}
where $\mathbf{F}_i\in\mathbb{R}^Q$ denotes the vector formed by the $i$-th row of $\mathbf{F}$, the collection of outputs from the DGP network in the GDGP hierarchy. This construction allows the DGP network to determine how the parameters of the output distribution vary across the input space, while the likelihood specifies the conditional distribution of the simulator output. For example, for a heteroskedastic Gaussian response, two DGP outputs control the mean and log-variance, respectively, and for a categorical response, a vector of DGP outputs is transformed into class probabilities.}

Since GDGP may be viewed as a DGP with a non-GP node in the final layer, inference can be implemented naturally within the SI framework developed for DGPs. {In particular, SI treats the latent-layer outputs, which also serve as inputs to subsequent layers, as missing data. During training, it follows the logic of stochastic expectation-maximization (SEM), estimating the GP hyperparameters by alternating between latent-value imputation and maximum-likelihood updates of the hyperparameters conditional on the imputed latent variables. For prediction, SI draws on the linked Gaussian process (LGP;~\citealp{kyzyurova2018coupling,ming2021linked}), which was originally developed to emulate networks of computer models. In such networks, separate GP emulators are fitted to the observed inputs and outputs of the individual component models and then connected according to the network structure to form an emulator of the overall network. For a DGP, each imputation of its latent layers exposes the input-output pairs of the individual GP nodes, thereby converting the DGP prediction problem into that of LGPs. Consequently, at a new input, SI leverages the LGP to propagate the predictive distribution of the individual GP outputs through the hierarchy and obtain analytically tractable approximations to the posterior predictive mean and variance of the DGP output. Two other approaches, fully Bayesian (FB; \citealp{sauer2023active}) inference and variational inference (VI; \citealp{salimbeni2017doubly}), have also been developed for DGP inference. For training, FB fits the model by jointly sampling the latent layers and GP hyperparameters using MCMC, whereas VI fits the model by maximizing an evidence lower bound based on a sparse variational approximation to the posterior. For prediction, instead of utilizing the LGP approximation, FB approximates posterior predictive quantities by propagating and averaging over posterior samples, while VI propagates Monte Carlo samples through the fitted variational posterior. Although FB more fully accounts for posterior uncertainty, VI is generally computationally faster. SI was therefore introduced as a natural and computationally efficient approximation to FB inference for DGPs, analogous to the common practice in GP modeling of first estimating the correlation lengthscales and then conditioning on those estimates for posterior inference and prediction. A detailed comparison of SI with FB and VI for DGPs is provided in~\citet{ming2021deep}.}

In the remainder of this section, we present the details of the three components of SI inference for GDGP, namely prediction (Section~\ref{sec:prediction}), imputation (Section~\ref{sec:imputation}), and training (Section~\ref{sec:training}).

\subsection{Prediction}
\label{sec:prediction}
Let $\boldsymbol{\theta}^{(p)}_l=\{(\sigma_l^{(p)})^2,\boldsymbol{\gamma}_l^{(p)}\}$ with $\boldsymbol{\gamma}_l^{(p)}=(\gamma_{l,1}^{(p)},\dots,\gamma_{l,P_{l-1}}^{(p)})^\top$ be the model parameters in $\mathcal{GP}^{(p)}_l$ and assume that $\boldsymbol{\theta}^{(p)}_l$ are known and distinct for all $p=1,\dots,P_l$ and $l=1,\dots,L$. Then given a realization $\mathbf{y}$ of output $\mathbf{Y}$, the posterior predictive distribution of the GDGP output $Y_0(\mathbf{x}_0)$ at a new input $\mathbf{x}_0$ can be approximated as follows:
\begin{align*}
p(y_0|\mathbf{x}_0;\mathbf{y},\mathbf{x})=&\int p(y_0|\mathbf{x}_0;\mathbf{y},\mathbf{f},\{\mathbf{w}^{(p)}_{l}\},\mathbf{x})\,p(\mathbf{f},\{\mathbf{w}^{(p)}_{l}\}|\mathbf{y},\mathbf{x})\,\mathrm{d}\mathbf{f}\,\mathrm{d}\{\mathbf{w}^{(p)}_{l}\}\nonumber\\
=&\int\left(\int p(y_0|\mathbf{f}_0)\,p(\mathbf{f}_0|\mathbf{x}_0;\mathbf{f},\{\mathbf{w}^{(p)}_{l}\},\mathbf{x})\,\mathrm{d}\mathbf{f}_0\right)\,p(\mathbf{f},\{\mathbf{w}^{(p)}_{l}\}|\mathbf{y},\mathbf{x})\,\mathrm{d}\mathbf{f}\,\mathrm{d}\{\mathbf{w}^{(p)}_{l}\}\nonumber\\
=&\mathbb{E}_{\mathbf{F},\{\mathbf{W}^{(p)}_{l}\}|\mathbf{y},\mathbf{x}}\left[\int p(y_0|\mathbf{f}_0)\,p(\mathbf{f}_0|\mathbf{x}_0;\mathbf{F},\{\mathbf{W}^{(p)}_{l}\},\mathbf{x})\,\mathrm{d}\mathbf{f}_0\right]\nonumber\\
\mathrel{\dot=}&\frac{1}{K}\sum_{k=1}^K \int p(y_0|\mathbf{f}_0)\,p(\mathbf{f}_0|\mathbf{x}_0;\mathbf{f}_k,\{\mathbf{w}^{(p)}_{l}\}_k,\mathbf{x})\,\mathrm{d}\mathbf{f}_0\\
\mathrel{\dot=}&\frac{1}{K}\sum_{k=1}^K \int p(y_0|\mathbf{f}_0)\,\prod_{q=1}^{Q}\widehat{p}(f^{(q)}_0|\mathbf{x}_0;\mathbf{f}_k,\{\mathbf{w}^{(p)}_{l}\}_k,\mathbf{x})\,\mathrm{d}\mathbf{f}_0\\
=&\frac{1}{K}\sum_{k=1}^K \widehat{p}(y_0|\mathbf{x}_0;\mathbf{f}_k,\{\mathbf{w}^{(p)}_{l}\}_k,\mathbf{x})\,,
\end{align*}
where {a Monte Carlo approximation is used in the fourth line to the expectation in the third line, whose exact evaluation is analytically intractable because of the nonlinear DGP structure. In the fourth line, $p(\mathbf{f}_0\mid\mathbf{x}_0;\mathbf{f}_k,\{\mathbf{w}^{(p)}_l\}_k,\mathbf{x})$ denotes the conditional predictive distribution of the DGP output $\mathbf{F}_0$ at the new input $\mathbf{x}_0$, given the $k$-th realization $\mathbf{f}_k$ and $\{\mathbf{w}^{(p)}_{l}\}_k$ of the latent layers. Although this distribution remains analytically intractable, it is induced by a feed-forward network of connected GPs with realized latent layers and is therefore the type of predictive distribution for which the LGP was developed to provide an analytical approximation. Consequently, the fifth line replaces each conditional DGP predictive distribution with its LGP approximation $\prod_{q=1}^{Q}\widehat{p}(f^{(q)}_0\mid\mathbf{x}_0;\mathbf{f}_k,\{\mathbf{w}^{(p)}_l\}_k,\mathbf{x})$,} where $\widehat{p}(f^{(q)}_0|\mathbf{x}_0;\mathbf{f}_k,\{\mathbf{w}^{(p)}_{l}\}_k,\mathbf{x})$ for $q=1,\dots,Q$ is a univariate normal distribution with analytically tractable mean $\mu^{(q)}_{1\rightarrow L,k}(\mathbf{x}_0)$ and variance ${\sigma^2}^{(q)}_{1\rightarrow L,k}(\mathbf{x}_0)$ that can be obtained by iterating the following formulae:
\begin{align}
\label{eq:linkgp_mean}
\mu^{(q)}_{1\rightarrow l,k}(\mathbf{x}_0)=&\mathbf{I}_{l,k}^{(q)}(\mathbf{x}_0)^\top\left(\mathbf{R}^{(q)}_{l,k}\right)^{-1}\mathbf{w}^{(q)}_{l,k}\,,\\
\label{eq:linkgp_var}
{\sigma^2}^{(q)}_{1\rightarrow l,k}(\mathbf{x}_0)=&\left(\mathbf{w}^{(q)}_{l,k}\right)^\top\left(\mathbf{R}^{(q)}_{l,k}\right)^{-1}\mathbf{J}^{(q)}_{l,k}(\mathbf{x}_0)\left(\mathbf{R}^{(q)}_{l,k}\right)^{-1}\mathbf{w}^{(q)}_{l,k}-\left(\mathbf{I}_{l,k}^{(q)}(\mathbf{x}_0)^\top\left(\mathbf{R}^{(q)}_{l,k}\right)^{-1}\mathbf{w}^{(q)}_{l,k}\right)^2\nonumber\\
&\quad+\left(\sigma^{(q)}_{l}\right)^2\left(1+\eta^{(q)}_{l}-\mathrm{tr}\left\{\left(\mathbf{R}^{(q)}_{l,k}\right)^{-1}\mathbf{J}^{(q)}_{l,k}(\mathbf{x}_0)\right\}\right)
\end{align}
for $l=2,\dots,L$ and $q=1,\dots,P_l$, where the $i$-th element of $\mathbf{I}^{(q)}_{l,k}(\mathbf{x}_0)\in\mathbb{R}^{N\times 1}$ is given by 
    $$
    \prod_{d=1}^{P_{l-1}}\xi^{(q)}_{l}\left({\mu}^{(d)}_{1\rightarrow(l-1),k}(\mathbf{x}_0),\,{{\sigma}^{2}}^{(d)}_{1\rightarrow(l-1),k}(\mathbf{x}_0),\,(\mathbf{w}^{(d)}_{l-1,k})_{i}\right)\,;
    $$
the $ij$-th element of $\mathbf{J}^{(q)}_{l,k}(\mathbf{x}_0)\in\mathbb{R}^{N\times N}$ is given by 
    $$
   \prod_{d=1}^{P_{l-1}}\zeta^{(q)}_{l}\left({\mu}^{(d)}_{1\rightarrow(l-1),k}(\mathbf{x}_0),\,{{\sigma}^{2}}^{(d)}_{1\rightarrow(l-1),k}(\mathbf{x}_0),\,(\mathbf{w}^{(d)}_{l-1,k})_{i},\,(\mathbf{w}^{(d)}_{l-1,k})_{j}\right)\,;
    $$
and ${\mu}^{(q)}_{1\rightarrow 1,k}(\mathbf{x}_0)$ and ${{\sigma}^{2}}^{(q)}_{1\rightarrow 1,k}(\mathbf{x}_0)$ are given by
\begin{align}
\label{eq:gp_mean}
{\mu}^{(q)}_{1\rightarrow 1,k}(\mathbf{x}_0)&=\mathbf{r}^{(q)}(\mathbf{x}_0)^\top\left(\mathbf{R}^{(q)}_{1}\right)^{-1}\mathbf{w}^{(q)}_{1,k}\,,\\
\label{eq:gp_var}
{{\sigma}^{2}}^{(q)}_{1\rightarrow 1,k}(\mathbf{x}_0)&=\left(\sigma^{(q)}_1\right)^2\left(1+\eta^{(q)}_1-\mathbf{r}^{(q)}(\mathbf{x}_0)^\top\left(\mathbf{R}^{(q)}_{1}\right)^{-1}\mathbf{r}^{(q)}(\mathbf{x}_0)\right)\,,   
\end{align}
respectively, for $q=1,\dots,P_1$, where $\mathbf{r}^{(q)}(\mathbf{x}_0)=[k^{(q)}_1(\mathbf{x}_0,\mathbf{x}_{1*}),\dots,k^{(q)}_1(\mathbf{x}_0,\mathbf{x}_{N*})]^\top$, and $\xi^{(q)}_{l}(\cdot,\cdot,\cdot)$ and $\zeta^{(q)}_{l}(\cdot,\cdot,\cdot,\cdot)$ are analytically tractable functions given in~\citet[Appendix A]{ming2021linked} for $\mathcal{GP}^{(q)}_{l}$ with squared exponential or Matérn kernels.

Let $\tilde{\mu}^{Y}_{0,k}$ and $(\tilde{\sigma}^{Y}_{0,k})^2$ denote the mean and variance of $\widehat{p}(y_0|\mathbf{x}_0;\mathbf{f}_k,\{\mathbf{w}^{(p)}_{l}\}_k,\mathbf{x})$, respectively. For many parametric distributions for the likelihood layer, $p(y_0|\boldsymbol{\phi}_0=\mathbf{g}^{-1}(\mathbf{f}_0))$, including the heteroskedastic Gaussian, Poisson, negative binomial, and zero-inflated Poisson and negative binomial distributions considered in the work, the LGP approximation yields closed-form expressions for $\tilde{\mu}^{Y}_{0,k}$ and $(\tilde{\sigma}^{Y}_{0,k})^2$; see Appendix~\ref{app:mean_variance} for a non-exhaustive list of distributions admitting such expressions. The mean and variance of $Y_0(\mathbf{x}_0)$ can therefore be approximated by

\begin{equation}
\label{eq:mean_variance}
    \tilde{\mu}^{Y}_0\mathrel{\dot=}\frac{1}{K}\sum_{k=1}^K\tilde{\mu}^{Y}_{0,k}\quad\mathrm{and}\quad
    (\tilde{\sigma}^{Y}_0)^2\mathrel{\dot=}\frac{1}{K}\sum_{k=1}^K((\tilde{\mu}^{Y}_{0,k})^2+(\tilde{\sigma}^{Y}_{0,k})^2)-(\tilde{\mu}^{Y}_0)^2\,,
\end{equation}
respectively. When closed-form expressions for the relevant summaries are unavailable, or when other quantities of interest (QoIs) are needed to characterize $Y_0(\mathbf{x}_0)$, such as category probabilities in the categorical case, fast sample-based approximations can be obtained via the method of composition~\citep{tanner1993tools}. Specifically, noting that
\begin{align*}
p(y_0,\mathbf{f}_0,\mathbf{f},\{\mathbf{w}^{(p)}_{l}\}|\mathbf{x}_0;\mathbf{y},\mathbf{x})=& p(y_0|\mathbf{f}_0)\,p(\mathbf{f}_0|\mathbf{x}_0;\mathbf{f},\{\mathbf{w}^{(p)}_{l}\},\mathbf{x})\,p(\mathbf{f},\{\mathbf{w}^{(p)}_{l}\}|\mathbf{y},\mathbf{x})\\
\mathrel{\dot=}&p(y_0|\mathbf{f}_0)\,\prod_{q=1}^{Q}\widehat{p}(f^{(q)}_0|\mathbf{x}_0;\mathbf{f},\{\mathbf{w}^{(p)}_{l}\},\mathbf{x})\,p(\mathbf{f},\{\mathbf{w}^{(p)}_{l}\}|\mathbf{y},\mathbf{x})\,,
\end{align*}
we can generate $K$ approximate samples from $p(y_0|\mathbf{x}_0;\mathbf{y},\mathbf{x})$ by repeatedly applying Algorithm~\ref{alg:composition}. These samples can then be used to approximate the mean, variance, or other QoIs of $Y_0(\mathbf{x}_0)$. The overall prediction procedure for GDGP is summarized in Algorithm~\ref{alg:prediction}.

\begin{algorithm}[htbp]
\caption{Method of Composition}
\label{alg:composition}
\begin{algorithmic}[1]
    \STATE{Draw $\mathbf{f}_k$ and $\{\mathbf{w}^{(p)}_{l}\}_k$ from $p(\mathbf{f},\{\mathbf{w}^{(p)}_{l}\}|\mathbf{y},\mathbf{x})$;}
    \STATE{Draw $\mathbf{f}_{0,k}$ from the LGPs $\widehat{p}(f^{(q)}_0|\mathbf{x}_0;\mathbf{f}_k,\{\mathbf{w}^{(p)}_{l}\}_k,\mathbf{x})$ for all $q=1,\dots,Q$;}
    \label{step2}
    \STATE{ Draw $y_{0,k}$ from $p(y_0|\mathbf{f}_{0,k})$.}
    \label{step3}
\end{algorithmic} 
\end{algorithm}

\begin{algorithm}[htbp]
\caption{Prediction from the GDGP}
\label{alg:prediction}
\begin{algorithmic}[1]
\REQUIRE{\begin{enumerate*}[label=(\roman*)]
  \item Observations $\mathbf{x}$ and $\mathbf{y}$;
  \item Trained $\{\mathcal{GP}^{(p)}_l\}_{p=1,\dots,P_l,l=1,\dots,L}$ in the GDGP hierarchy;
  \item A new input position $\mathbf{x}_0$.
\end{enumerate*}}
\ENSURE{Mean, $\tilde{\mu}^{Y}_0$, and variance, $(\tilde{\sigma}^{Y}_0)^2$, of $Y_0(\mathbf{x}_0)$.}
\STATE{\label{alg:mi}Draw $K$ realizations $\mathbf{f}_1,\dots,\mathbf{f}_K$ and $\{\mathbf{w}^{(p)}_l\}_1,\dots,\{\mathbf{w}^{(p)}_l\}_K$ from $p(\mathbf{f},\{\mathbf{w}^{(p)}_l\}|\mathbf{y},\mathbf{x})$;}
\STATE{\label{alg:lgp}Construct $K$ LGPs $\widehat{p}(f^{(q)}_0|\mathbf{x}_0;\mathbf{f}_1,\{\mathbf{w}^{(p)}_{l}\}_1,\mathbf{x}),\dots,\widehat{p}(f^{(q)}_0|\mathbf{x}_0;\mathbf{f}_K,\{\mathbf{w}^{(p)}_{l}\}_K,\mathbf{x})$ for all $q=1,\dots,Q$;}
\STATE{Compute the mean, variance, or other QoIs of $Y_0(\mathbf{x}_0)$ using
\begin{enumerate*}[label=(\roman*)]
\item Equations~\eqref{eq:mean_variance}, when the relevant closed-form expressions are available; or
\item the method of composition by executing Steps~\ref{step2}--\ref{step3} of Algorithm~\ref{alg:composition} for all $K$ LGPs.
\end{enumerate*}}
\end{algorithmic} 
\end{algorithm}

\subsection{Imputation}
\label{sec:imputation}
Simulations or imputations of $\mathbf{F}$ and $\{\mathbf{W}^{(p)}_{l}\}$ from $p(\mathbf{f},\{\mathbf{w}^{(p)}_{l}\}|\mathbf{y},\mathbf{x})$ in Step~\ref{alg:mi} of Algorithm~\ref{alg:prediction} cannot be achieved exactly due to the complexity introduced by the GDGP hierarchy. However, the conditional posteriors $p(\mathbf{w}^{(p)}_{l}|\{\mathbf{w}^{(p)}_l\}\setminus\mathbf{w}^{(p)}_{l},\mathbf{f},\mathbf{y},\mathbf{x})$ can be expressed as
\begin{equation}
\label{eq:dgp_pos}
p(\mathbf{w}^{(p)}_{l}|\{\mathbf{w}^{(p)}_l\}\setminus\mathbf{w}^{(p)}_{l},\mathbf{f},\mathbf{y},\mathbf{x})\propto \prod_{q=1}^{P_{l+1}}p(\mathbf{w}^{(q)}_{l+1}|\mathbf{w}^{(1)}_{l},\dots,\mathbf{w}^{(p)}_{l},\dots,\mathbf{w}^{(P_{l})}_{l})\,p(\mathbf{w}^{(p)}_{l}|\mathbf{w}^{(1)}_{l-1},\dots,\mathbf{w}^{(P_{l-1})}_{l-1})
\end{equation}
for $p=1,\dots,P_l$ and $l=1,\dots,L-1$, where $\mathbf{w}^{(q)}_{L}=\mathbf{f}^{(q)}$ and $P_L=Q$. Analogously,
\begin{align}
\label{eq:lik_pos}
p(\mathbf{f}^{(q)}|\mathbf{f}\setminus\mathbf{f}^{(q)},\{\mathbf{w}^{(p)}_l\},\mathbf{y},\mathbf{x})&\propto p(\mathbf{y}|\mathbf{f}^{(1)},\dots,\mathbf{f}^{(q)},\dots,\mathbf{f}^{(Q)})\,p(\mathbf{f}^{(q)}|\mathbf{w}^{(1)}_{L-1},\dots,\mathbf{w}^{(P_{L-1})}_{L-1})\nonumber\\
&=\prod_{i=1}^Np(y_i|g^{-1}_1(f_i^{(1)}),\dots,g^{-1}_q(f_i^{(q)}),\dots,g^{-1}_Q(f_i^{(Q)}))\,p(\mathbf{f}^{(q)}|\mathbf{w}^{(1)}_{L-1},\dots,\mathbf{w}^{(P_{L-1})}_{L-1})
\end{align}
for $q=1,\dots,Q$. Since $p(\mathbf{w}^{(p)}_{l}|\mathbf{w}^{(1)}_{l-1},\dots,\mathbf{w}^{(P_{l-1})}_{l-1})$ in Equation~\eqref{eq:dgp_pos} and $p(\mathbf{f}^{(q)}|\mathbf{w}^{(1)}_{L-1},\dots,\mathbf{w}^{(P_{L-1})}_{L-1})$ in Equation~\eqref{eq:lik_pos} are multivariate normal, one can impute latent layers $\mathbf{F}$ and $\{\mathbf{W}^{(p)}_{l}\}$ by utilizing the Elliptical Slice Sampling~\citep{murray2010elliptical} within a Gibbs (ESS-within-Gibbs) sampler. Algorithm~\ref{alg:ess} illustrates a single imputation of $\mathbf{F}$ and $\{\mathbf{W}^{(p)}_{l}\}$ from the ESS-within-Gibbs sampler. 

\begin{algorithm}[!ht]
\caption{One-step ESS-within-Gibbs sampler to impute $\mathbf{F}$ and $\{\mathbf{W}^{(p)}_{l}\}$}
\label{alg:ess}
\begin{algorithmic}[1]
\REQUIRE{A current imputation $\mathbf{f}_i$ and $\{\mathbf{w}^{(p)}_l\}_i$.}
\ENSURE{A new imputation $\mathbf{f}_{i+1}$ and $\{\mathbf{w}^{(p)}_l\}_{i+1}$.}
\FOR{$l=1,\dots,L$}
\IF{$l=L$}
\FOR{$q=1,\dots,Q$}
\STATE{\label{alg:lik_ess}Impute $\mathbf{F}^{(q)}$ by drawing $\mathbf{f}^{(q)}_{i+1}$ from $p(\mathbf{f}^{(q)}|\mathbf{f}\setminus\mathbf{f}^{(q)},\{\mathbf{w}^{(p)}_l\},\mathbf{y},\mathbf{x})$ in the form of~\eqref{eq:lik_pos} via an ESS update;}
\ENDFOR
\ELSE
\FOR{$p=1,\dots,P_l$}
\STATE{Impute $\mathbf{W}^{(p)}_l$ by drawing $\mathbf{w}^{(p)}_{l,i+1}$ from $p(\mathbf{w}^{(p)}_{l}|\{\mathbf{w}^{(p)}_l\}\setminus\mathbf{w}^{(p)}_{l},\mathbf{f},\mathbf{y},\mathbf{x})$ in the form of~\eqref{eq:dgp_pos} via an ESS update;}
\ENDFOR
\ENDIF
\ENDFOR
\end{algorithmic} 
\end{algorithm}

\subsection{Training}
\label{sec:training}
Following the SI framework introduced by \citet{ming2021deep}, the unknown GP parameters $\boldsymbol{\theta}_l^{(p)}$ in the GDGP structure are estimated using the Stochastic Expectation-Maximization (SEM) algorithm~\citep{celeux1985sem}, as summarized in Algorithm~\ref{alg:sem}. Once the GDGP hyperparameters have been estimated, the imputation module (see Section~\ref{sec:imputation}) generates a set of GDGP realizations, thereby enabling fast posterior prediction via the LGPs constructed according to Algorithm~\ref{alg:prediction}.

\begin{algorithm}[!ht]
\caption{Training for the GDGP via SEM}
\label{alg:sem}
\begin{algorithmic}[1]
\REQUIRE{\begin{enumerate*}[label=(\roman*)]
  \item Observations $\mathbf{x}$ and $\mathbf{y}$;
  \item initial values of model parameters $\{\widehat{\boldsymbol{\theta}}^{(p,1)}_l\}$;
  \item total number of iterations $T$ and burn-in period $B$ for SEM; 
  \item burn-in periods $C$ for ESS. 
\end{enumerate*}}
\ENSURE{Point estimates of model parameters.}
\FOR{$t=1,\dots,T$}
\STATE{\textbf{Imputation-step}: draw an imputation of $\mathbf{F}$ and $\{\mathbf{W}^{(p)}_l\}$ from $p(\mathbf{f},\{\mathbf{w}^{(p)}_{l}\}|\mathbf{y},\mathbf{x};\{\widehat{\boldsymbol{\theta}}^{(p,t)}_l\})$ by evaluating $C$ steps of ESS-within-Gibbs sampler in Algorithm~\ref{alg:ess};}
\STATE{\textbf{Maximization-step}: update model parameters by maximizing log-likelihoods of individual GPs:
\begin{align*}
    {\widehat{\boldsymbol{\theta}}^{(p,t+1)}_l}=&\argmax\log p(\mathbf{w}^{(p)}_{l}|\mathbf{w}^{(1)}_{l-1},\dots,\mathbf{w}^{(P_{l-1})}_{l-1};\boldsymbol{\theta}^{(p)}_l)\,,
\end{align*}
for all $p=1,\dots,P_l$ and $l=1,\dots,L$, where $\mathbf{w}^{(p)}_{L}=\mathbf{f}^{(p)}$ and $P_L=Q$.}
\ENDFOR
\STATE{Compute point estimates ${\widehat{\boldsymbol{\theta}}^{(p)}_l}$ of model parameters by:
\begin{equation*}
\label{eq:estimates}
  \widehat{\boldsymbol{\theta}}^{(p)}_l=\frac{1}{T-B}\sum_{t=B+1}^{T}\widehat{\boldsymbol{\theta}}^{(p,t)}_l\quad\forall\,p,l.
\end{equation*}}
\end{algorithmic} 
\end{algorithm}

\section{Scalability}
\label{sec:scalability}
The computational complexity of performing inference for GDGP emulators using SI can become prohibitive as the training data size increases, primarily due to the well-known cubic complexity of covariance matrix inversion. {To address this issue, in the remainder of this section, we first show how to incorporate the Vecchia approximation into the SI framework, enabling scalable GDGP emulation with large numbers of unique input locations. We then show how large numbers of replicates can be accommodated efficiently within SI. Finally, we demonstrate further scalability improvements for the heteroskedastic Gaussian case in settings with large numbers of unique training locations and replicates.}

\subsection{Vecchia Approximation}
The main computational bottleneck of SI arises from the evaluation of likelihood functions and sampling from $\{\mathcal{GP}^{(p)}_{l}\}$, which require repeated inversion of covariance matrices during the ESS-within-Gibbs sampling in the Imputation step and the log-likelihood optimization in the Maximization step of Algorithm~\ref{alg:sem}. To address this with the Vecchia approximation, consider an elementary GP, $\mathcal{GP}^{(p)}_{l}$, in the GDGP hierarchy. The Vecchia's log-likelihood function~\citep{guinness2021gaussian} of $\mathcal{GP}^{(p)}_{l}$ can be written as:
\begin{align}
\label{eq:vecch_loglik}
\log\widehat{\mathcal{L}}(\boldsymbol{\theta}^{(p)}_l)=&\log \widehat{p}(\mathbf{w}^{(p)}_{l}|\mathbf{w}_{l-1})\nonumber\\
=&\sum_{i=1}^N\log p(w^{(p)}_{l,i}|\mathbf{w}^{(p)}_{l,g(-i)},\mathbf{w}_{l-1})\nonumber\\
=&\sum_{i=1}^N\left(\log p(\mathbf{w}^{(p)}_{l,g(i)}|\mathbf{w}_{l-1,g(i)*})-\log p(\mathbf{w}^{(p)}_{l,g(-i)}|\mathbf{w}_{l-1,g(-i)*})\right)\nonumber\\
=&-\frac{N}{2}\log(2\pi)-\frac{1}{2}\sum_{i=1}^N\left(\log\det\boldsymbol{\Sigma}^{(p)}_{l,g(i)}-\log\det\boldsymbol{\Sigma}^{(p)}_{l,g(-i)}\right)\nonumber\\
&-\frac{1}{2}\sum_{i=1}^N\left[\left(\mathbf{w}^{(p)}_{l,g(i)}\right)^{\top}\left(\boldsymbol{\Sigma}^{(p)}_{l,g(i)}\right)^{-1}\mathbf{w}^{(p)}_{l,g(i)} - \left(\mathbf{w}^{(p)}_{l,g(-i)}\right)^{\top}\left(\boldsymbol{\Sigma}^{(p)}_{l,g(-i)}\right)^{-1}\mathbf{w}^{(p)}_{l,g(-i)}
\right]\,,
\end{align}
where $g(-i)\subseteq\{1,2,\dots,i-1\}$ is a conditioning index set of size $|g(-i)|=\min(M, i-1)$ for all $i=2,\dots,N$ with $g(-1)=\varnothing$ and $g(i)=g(-i)\cup\{i\}$; $\boldsymbol{\Sigma}^{(p)}_{l,g(i)}\in\mathbb{R}^{(|g(-i)|+1)\times (|g(-i)|+1)}$ and $\boldsymbol{\Sigma}^{(p)}_{l,g(-i)}\in\mathbb{R}^{|g(-i)|\times |g(-i)|}$ are covariance matrices of $\mathbf{W}^{(p)}_{l,g(i)}$ and $\mathbf{W}^{(p)}_{l,g(-i)}$ respectively. The complexity of evaluating Vecchia's log-likelihood function~\eqref{eq:vecch_loglik} is $\mathcal{O}(NM^3)$, which is significantly lower than the $\mathcal{O}(N^3)$ complexity of evaluating the original log-likelihood function of $\mathcal{GP}^{(p)}_{l}$ when $M \ll N$. Notably, the summations in~\eqref{eq:vecch_loglik} are independent, allowing for parallel processing across different summands, which can effectively reduce the $N$ factor in $\mathcal{O}(NM^3)$. 

The Vecchia approximation implies that the probability distribution of $\mathbf{W}^{(p)}_{l}$, defined by the probability density function $\widehat{p}(\mathbf{w}^{(p)}_{l}|\mathbf{w}_{l-1})=\prod_{i=1}^N p(w^{(p)}_{l,i}|\mathbf{w}^{(p)}_{l,g(-i)},\mathbf{w}_{l-1})$, is itself a multivariate normal distribution:
\begin{equation}
\label{eq:vecch_gp}
\mathbf{W}^{(p)}_l|\mathbf{W}_{l-1}\sim\mathcal{N}\left(\mathbf{0},\,\left({\mathbf{P}^{(p)}_l}\right)^{-1}\right)\,,
\end{equation}
where $\mathbf{P}^{(p)}_l=\mathbf{U}^{(p)}_l{\mathbf{U}^{(p)}_l}^{\top}$ is the precision matrix of $\mathbf{W}^{(p)}_l|\mathbf{W}_{l-1}$ with $\mathbf{U}^{(p)}_l$ being a sparse upper triangular matrix whose $ji$-th element is given by~\citep{katzfuss2021general}:
\begin{equation}
\label{eq:u_matrix}
\mathbf{U}^{(p)}_{l,ji}=\begin{dcases}
\left(d^{(p)}_{l,i}\right)^{-\frac{1}{2}}, & j=i\,,\\
-b^{(p)}_{l,ik}\left(d^{(p)}_{l,i}\right)^{-\frac{1}{2}}, & j\in g(-i)\,,\\
0\,, & \mathrm{otherwise}\,,
\end{dcases}
\end{equation}
where 
\begin{itemize}
    \item $d^{(p)}_{l,i}=(\sigma_l^{(p)})^2\left(1+\eta-\left(\mathbf{r}^{(p)}_{l,i}\right)^{\top}\left(\mathbf{R}^{(p)}_{l,g(-i)}\right)^{-1}\mathbf{r}^{(p)}_{l,i}\right)$\,;
    \item $\mathbf{b}^{(p)}_{l,i}=\left(\mathbf{r}^{(p)}_{l,i}\right)^{\top}\left(\mathbf{R}^{(p)}_{l,g(-i)}\right)^{-1}$
\end{itemize} 
with $\mathbf{r}^{(p)}_{l,i}=[k^{(p)}_{l}(\mathbf{w}_{l-1,i*},\mathbf{w}_{l-1,j*})]^{\top}_{j\in g(-i)}\in\mathbb{R}^{|g(-i)|\times 1}$ and $\mathbf{R}^{(p)}_{l,g(-i)}$ being the correlation matrix of $\mathbf{W}^{(p)}_{l,g(-i)}$; and $b^{(p)}_{l,ik}$ is the $k$-th element of $\mathbf{b}^{(p)}_{l,i}$ with $k=\mathrm{pos}_{g(-i)}(j)$ denoting the position of $j$ in the conditioning set $g(-i)$. Note that $\mathbf{U}^{(p)}_{l}$ is highly sparse, and its construction has a complexity at most of $\mathcal{O}(NM^3)$, as the construction of each column is independent (allowing for parallel computation), each column contains at most $M$ off-diagonal non-zero elements, and computing each column involves inverting the correlation matrix $\mathbf{R}^{(p)}_{l,g(-i)}\in\mathbb{R}^{|g(-i)|\times |g(-i)|}$. Given the fact that 
\begin{equation*}
\mathbf{W}^{(p)}_l \mid \mathbf{W}_{l-1} = \left({\mathbf{U}^{(p)}_{l}}^{\top}\right)^{-1}\mathbf{Z}\,,
\end{equation*}
where $\mathbf{Z} \in \mathbb{R}^{N \times 1}$ is a vector of i.i.d. univariate standard normal random variables, a realization $\mathbf{w}^{(p)}_l$ of $\mathbf{W}^{(p)}_l \mid \mathbf{W}_{l-1}$ can then be generated from the multivariate normal distribution~\eqref{eq:vecch_gp} by performing:
\begin{equation}
\label{eq:vecch_sampling}
\mathbf{w}^{(p)}_l = \left({\mathbf{U}^{(p)}_{l}}^{\top}\right)^{-1}\mathbf{z}\,,
\end{equation}
where $\mathbf{z}$ is a realization from $\mathbf{Z}$. Since $\mathbf{U}^{(p)}_{l}$ is a sparse upper triangular matrix, the computation of~\eqref{eq:vecch_sampling} can be performed via forward substitution with a complexity of $\mathcal{O}(NM)$ instead of $\mathcal{O}(N^2)$, leveraging the sparsity of $\mathbf{U}^{(p)}_{l}$.

Replacing the log-likelihood function in the ESS update in Algorithm~\ref{alg:ess} and the Maximization step in Algorithm~\ref{alg:sem} with Vecchia's form~\eqref{eq:vecch_loglik}, and replacing the sampling from $\{\mathcal{GP}^{(p)}_{l}\}$ in the ESS update in Algorithm~\ref{alg:ess} with~\eqref{eq:vecch_sampling}, then enables scalable training of GDGP via SEM.

It is worth noting that Vecchia's log-likelihood function~\eqref{eq:vecch_loglik} can equivalently be expressed in terms of $\mathbf{U}^{(p)}_{l}$ via~\eqref{eq:vecch_gp} as:
\begin{equation}
\label{eq:vecch_loglik_alter}
\log\widehat{\mathcal{L}}(\boldsymbol{\theta}^{(p)}_l) = -\frac{N}{2}\log(2\pi) - \frac{1}{2}\log\det\left(\mathbf{U}^{(p)}_l{\mathbf{U}^{(p)}_l}^{\top}\right)^{-1} - \frac{1}{2}\left(\mathbf{w}^{(p)}_{l}\right)^{\top}\left(\mathbf{U}^{(p)}_l{\mathbf{U}^{(p)}_l}^{\top}\right)\mathbf{w}^{(p)}_{l}.
\end{equation}

However, there are two main computational reasons why~\eqref{eq:vecch_loglik} is preferred over~\eqref{eq:vecch_loglik_alter} for SI. Firstly, the computation of~\eqref{eq:vecch_loglik_alter} requires storage of $\mathbf{U}^{(p)}_{l}$. Although $\mathbf{U}^{(p)}_{l}$ is sparse, it can still be memory-intensive when $N$ is large. In contrast,~\eqref{eq:vecch_loglik} is more memory-efficient, as each summand only requires storing a small covariance matrix $\boldsymbol{\Sigma}^{(p)}_{l,g(i)}$, and the evaluation of $\log\widehat{\mathcal{L}}(\boldsymbol{\theta}^{(p)}_l)$ can be performed by incrementally accumulating contributions of individual or a batch of summands with minimal memory overhead. Secondly, the form in~\eqref{eq:vecch_loglik} enables analytically {tractable} gradients with respect to $\boldsymbol{\theta}^{(p)}_l$, facilitating faster optimization of $\{\mathcal{GP}^{(p)}_{l}\}$ in the Maximization step of SEM in Algorithm~\ref{alg:sem}.

The ordering of $w^{(p)}_{l,1},\dots,w^{(p)}_{l,N}$ and the corresponding conditioning index set $g(-i)$ for $i=2,\dots,N$ affect the accuracy of Vecchia's log-likelihood function~\eqref{eq:vecch_loglik} and the multivariate normal distribution~\eqref{eq:vecch_sampling}. Different ordering approaches have been discussed in~\citet{katzfuss2021general}. In this study, we adopt a simple random ordering following~\citet{sauer2023vecchia}. Regarding the selection of the conditioning index set $g(-i)$, we employ nearest-neighbor (NN) conditioning, {a standard construction in NNGP and Vecchia methods~\citep{datta2016hierarchical,katzfuss2021general}.} Specifically, $g(-i)$ consists of the indices of up to $M$ nearest variables that precede $w^{(p)}_{l,i}$ based on the Euclidean distance between their respective inputs and the input $\mathbf{w}_{l-1,i*}$ of $w^{(p)}_{l,i}$, with each input dimension scaled by the corresponding lengthscale $\gamma^{(p)}_{l,d}$ for $d\in\{1,\dots,P_{l-1}\}$~\citep{katzfuss2022scaled}.

To achieve scalable prediction of GDGP, we can replace the log-likelihood function and sampling of $\{\mathcal{GP}^{(p)}_{l}\}$ in the imputation step on Line~\ref{alg:mi} of Algorithm~\ref{alg:prediction} by~\eqref{eq:vecch_loglik} and~\eqref{eq:vecch_sampling} respectively, and the following proposition gives the Vecchia's expressions of~\eqref{eq:linkgp_mean} and~\eqref{eq:linkgp_var} for scalable LGP construction on Line~\ref{alg:lgp} of Algorithm~\ref{alg:prediction}:
\begin{proposition}
\label{prop:vecchia_lgp}
Under the Vecchia approximation, the mean $\mu^{(q)}_{1\rightarrow L,k}(\mathbf{x}_0)$ and variance ${\sigma^2}^{(q)}_{1\rightarrow L,k}(\mathbf{x}_0)$ of the univariate normal distribution defined by $\widehat{p}(f^{(q)}_0|\mathbf{x}_0;\mathbf{f}_k,\{\mathbf{w}^{(p)}_{l}\}_k,\mathbf{x})$ can be obtained by iterating the following formulae:
\begin{align}
\label{eq:vecchia_linkgp_mean}
\mu^{(q)}_{1\rightarrow l,k}(\mathbf{x}_0)=&\mathbf{I}_{l,k,\mathcal{C}}^{(q)}(\mathbf{x}_0)^\top\left(\mathbf{R}^{(q)}_{l,k,\mathcal{C}}\right)^{-1}\mathbf{w}^{(q)}_{l,k,\mathcal{C}}\,,\\
\label{eq:vecchia_linkgp_var}
{\sigma^2}^{(q)}_{1\rightarrow l,k}(\mathbf{x}_0)=&\left(\mathbf{w}^{(q)}_{l,k,\mathcal{C}}\right)^\top\left(\mathbf{R}^{(q)}_{l,k,\mathcal{C}}\right)^{-1}\mathbf{J}^{(q)}_{l,k,\mathcal{C}}(\mathbf{x}_0)\left(\mathbf{R}^{(q)}_{l,k,\mathcal{C}}\right)^{-1}\mathbf{w}^{(q)}_{l,k,\mathcal{C}}-\left(\mathbf{I}_{l,k,\mathcal{C}}^{(q)}(\mathbf{x}_0)^\top\left(\mathbf{R}^{(q)}_{l,k,\mathcal{C}}\right)^{-1}\mathbf{w}^{(q)}_{l,k,\mathcal{C}}\right)^2\nonumber\\
&\quad+\left(\sigma^{(q)}_{l}\right)^2\left(1+\eta^{(q)}_{l}-\mathrm{tr}\left\{\left(\mathbf{R}^{(q)}_{l,k,\mathcal{C}}\right)^{-1}\mathbf{J}^{(q)}_{l,k,\mathcal{C}}(\mathbf{x}_0)\right\}\right)
\end{align}
for $l=2,\dots,L$ and $q=1,\dots,P_l$, with ${\mu}^{(q)}_{1\rightarrow 1,k}(\mathbf{x}_0)$ and ${{\sigma}^{2}}^{(q)}_{1\rightarrow 1,k}(\mathbf{x}_0)$ given by
\begin{align}
\label{eq:vecchia_gp_mean}
{\mu}^{(q)}_{1\rightarrow 1,k}(\mathbf{x}_0)&=\mathbf{r}^{(q)}_{\mathcal{C}}(\mathbf{x}_0)^\top\left(\mathbf{R}^{(q)}_{1,\mathcal{C}}\right)^{-1}\mathbf{w}^{(q)}_{1,k,\mathcal{C}}\,,\\
\label{eq:vecchia_gp_var}
{{\sigma}^{2}}^{(q)}_{1\rightarrow 1,k}(\mathbf{x}_0)&=\left(\sigma^{(q)}_1\right)^2\left(1+\eta^{(q)}_1-\mathbf{r}^{(q)}_{\mathcal{C}}(\mathbf{x}_0)^\top\left(\mathbf{R}^{(q)}_{1,\mathcal{C}}\right)^{-1}\mathbf{r}^{(q)}_{\mathcal{C}}(\mathbf{x}_0)\right)\,,   
\end{align}
respectively for $q=1,\dots,P_1$ where
\begin{itemize}
    \item $\mathcal{C}\subseteq\{1,2,\dots,N\}$ is a conditioning index set of size $|\mathcal{C}|\,;$
    \item $\mathbf{R}^{(q)}_{l,k,\mathcal{C}}\in\mathbb{R}^{|\mathcal{C}|\times |\mathcal{C}|}$ is the correlation matrix of $\mathbf{W}^{(q)}_{l,k,\mathcal{C}}=[(\mathbf{W}^{(q)}_{l,k})_i]^{\top}_{i\in \mathcal{C}}\,;$
    \item $\mathbf{I}_{l,k,\mathcal{C}}^{(q)}(\mathbf{x}_0) \in \mathbb{R}^{|\mathcal{C}| \times 1}$ is the sub-vector of $\mathbf{I}_{l,k}^{(q)}(\mathbf{x}_0)$ with its $i$-th element defined as
    $$
    \prod_{d=1}^{P_{l-1}}\xi^{(q)}_{l}\left({\mu}^{(d)}_{1\rightarrow(l-1),k}(\mathbf{x}_0),\,{{\sigma}^{2}}^{(d)}_{1\rightarrow(l-1),k}(\mathbf{x}_0),\,(\mathbf{w}^{(d)}_{l-1,k})_{\mathcal{C}_i}\right)\,;
    $$
    \item $\mathbf{J}^{(q)}_{l,k,\mathcal{C}}(\mathbf{x}_0)\in \mathbb{R}^{|\mathcal{C}| \times |\mathcal{C}|}$ is the sub-matrix of $\mathbf{J}^{(q)}_{l,k}(\mathbf{x}_0)$ with its $ij$-th element defined as 
    $$
   \prod_{d=1}^{P_{l-1}}\zeta^{(q)}_{l}\left({\mu}^{(d)}_{1\rightarrow(l-1),k}(\mathbf{x}_0),\,{{\sigma}^{2}}^{(d)}_{1\rightarrow(l-1),k}(\mathbf{x}_0),\,(\mathbf{w}^{(d)}_{l-1,k})_{\mathcal{C}_i},\,(\mathbf{w}^{(d)}_{l-1,k})_{\mathcal{C}_j}\right)\,;
    $$
    \item $\mathbf{R}^{(q)}_{1,\mathcal{C}}\in\mathbb{R}^{|\mathcal{C}|\times |\mathcal{C}|}$ is the correlation matrix of $\mathbf{W}^{(q)}_{1,\mathcal{C}}=[(\mathbf{W}^{(q)}_{1})_i]^{\top}_{i\in \mathcal{C}}\,;$
    \item $\mathbf{r}^{(q)}_{\mathcal{C}}(\mathbf{x}_0)=[k^{(q)}_1(\mathbf{x}_0,\mathbf{x}_{i*})]^\top_{i\in\mathcal{C}}\,.$
\end{itemize}
\end{proposition}
\begin{proof}
A sketch of the proof is given in Appendix~\ref{app:vecchia_lgp}.
\end{proof}

Note that Proposition~\ref{prop:vecchia_lgp} yields a substantially cheaper construction of LGPs, with complexity of $\mathcal{O}(|\mathcal{C}|^3)$ when $|\mathcal{C}|\ll N$. In this construction, the NN conditioning is applied for selecting the conditioning index set $\mathcal{C}$. Specifically, $\mathcal{C}$ contains the row indices of the closest input positions in  $\left[\frac{\mathbf{x}_{*1}}{\gamma^{(q)}_{1,1}},\dots,\frac{\mathbf{x}_{*D}}{\gamma^{(q)}_{1,D}}\right]$ to $\left[\frac{x_{0,1}}{\gamma^{(q)}_{1,1}},\frac{\dots,x_{0,D}}{\gamma^{(q)}_{1,D}}\right]$ when $l=1$, and the row indices of the closest input positions in $\left[\frac{\mathbf{w}^{(1)}_{l-1,k}}{\gamma^{(q)}_{l,1}},\dots,\frac{\mathbf{w}^{(P_{l-1})}_{l-1,k}}{\gamma^{(q)}_{l,P_{l-1}}}\right]$ to $\left[\frac{\mu^{(1)}_{1\rightarrow(l-1),k}}{\gamma^{(q)}_{l,1}},\dots,\frac{\mu^{(P_{l-1})}_{1\rightarrow(l-1),k}}{\gamma^{(q)}_{l,P_{l-1}}}\right]$ for $l=2,\dots,L$.

\subsection{Replicates}
\label{sec:replicates}

A distinctive feature of stochastic simulators is that multiple outputs may be generated at the same input setting. Such replicates are often used to characterize the intrinsic stochasticity of the simulator more accurately, but they can also impose substantial additional computational cost on GDGP inference when the total number of observations becomes large relative to the number of unique input locations. In this subsection, we show that replicates can be incorporated into GDGP in a way that preserves the original computational complexity.

Let $\mathbf{Y}=\{\mathbf{Y}_1,\dots,\mathbf{Y}_N\}$ be the collection of outputs, where $\mathbf{Y}_i\in\mathbb{R}^{S_i}$ contains $S_i$ observations repeatedly generated at the $i$-th input $\mathbf{x}_{i*}$. Assume that $Y_{i,1},\dots,Y_{i,S_i}$ are conditionally independent given $\boldsymbol{\phi}_i$. In a direct formulation, increasing the number of replicates $S_i$ enlarges the row dimension of the imputed latent variables $\mathbf{w}^{(p)}_l$ to $\sum_{i=1}^N S_i$, which in turn yields correlation matrices $\mathbf{R}^{(p)}_l$ of size $(\sum_{i=1}^N S_i)\times (\sum_{i=1}^N S_i)$ for all $p=1,\dots,P_l$ and $l=1,\dots,L$. As a result, the likelihood evaluations and sampling of multivariate normal distributions specified in Equations~\eqref{eq:dgp_pos} and~\eqref{eq:lik_pos} make Algorithm~\ref{alg:ess}, and consequently both training and prediction in GDGP, computationally burdensome when the numbers of replicates are large.

Note, however, that these replicates arise from repeated observations at the same input locations rather than from distinct inputs. Consequently, Equation~\eqref{eq:lik_pos} can be rewritten as
\begin{align}
\label{eq:lik_pos_rep}
p(\mathbf{f}^{(q)}|\mathbf{f}\setminus\mathbf{f}^{(q)},\{\mathbf{w}^{(p)}_l\},\mathbf{y},\mathbf{x})\propto \prod_{i=1}^N\prod_{s=1}^{S_i}p(y_{i,s}|g^{-1}_1(f_i^{(1)}),\dots,g^{-1}_q(f_i^{(q)}),\dots,g^{-1}_Q(f_i^{(Q)}))\,p(\mathbf{f}^{(q)}|\mathbf{w}^{(1)}_{L-1},\dots,\mathbf{w}^{(P_{L-1})}_{L-1})\,,
\end{align}
thereby avoiding the computational inefficiency that would otherwise be induced by treating replicates as independent observations over an expanded set of inputs. In particular, Equation~\eqref{eq:lik_pos_rep} implies that the correlation matrices arising in the DGP layers remain of size $N\times N$, defined only over the $N$ unique input locations. Consequently, the computational order for multivariate normal likelihood evaluation and sampling in the imputation step is unchanged from the no-replicate case: $N^3$ in the standard setting and $NM^3$ under the Vecchia approximation. At the likelihood layer, the presence of replicates only introduces additional product terms in Equation~\eqref{eq:lik_pos_rep}, which increases the cost linearly in the total number of replicated observations.

\subsection{The Heteroskedastic Gaussian Case}
\label{sec:hetero}

In addition to the general scalability improvements introduced above through the Vecchia approximation and the explicit treatment of replicates, the heteroskedastic Gaussian case permits further computational simplifications. This setting is of particular interest because heteroskedastic Gaussian emulation remains one of the most widely used approaches for stochastic simulators, while also illustrating how distribution-specific structure can be exploited within GDGP. In this subsection, we show that, under the heteroskedastic Gaussian likelihood, the conditional posterior distribution $p(\boldsymbol{\mu}|\log\boldsymbol{\sigma}^2, \{\mathbf{w}^{(p)}_l\}, \mathbf{y}, \mathbf{x})$ of the mean parameter admits closed-form expressions in the standard setting, as well as under the Vecchia approximation and in the presence of replicates. As a result, $\boldsymbol{\mu}$ can be sampled directly in the imputation step described in Section~\ref{sec:imputation}, avoiding ESS and providing additional computational gains.

\begin{proposition}
\label{prop:hetero_simple}
Given $\mathbf{Y}\in\mathbb{R}^N$, where $Y_i$ for $i=1,\dots,N$ are conditionally independent and distributed as $\mathcal{N}(\mu_i, \sigma_i^2)$ under the GDGP, the posterior distribution $p(\boldsymbol{\mu}|\log\boldsymbol{\sigma}^2, \{\mathbf{w}^{(p)}_l\}, \mathbf{y}, \mathbf{x})$ of $\boldsymbol{\mu} = \mathbf{F}^{(1)}$, given $\boldsymbol{\sigma}^2=\exp\{\mathbf{f}^{(2)}\}$, the imputed latent variables $\{\mathbf{w}^{(p)}_l\}$, and the observed inputs $\mathbf{x}$ and outputs $\mathbf{y}$, is a multivariate normal distribution given by:
\begin{equation}
\label{eq:hetero_simple}
\mathcal{N}\left(\boldsymbol{\Sigma}_L^{(1)}(\mathbf{w}_{L-1})\left(\boldsymbol{\Sigma}_L^{(1)}(\mathbf{w}_{L-1})+\boldsymbol{\Gamma}\right)^{-1}\mathbf{y},\, \boldsymbol{\Sigma}_L^{(1)}(\mathbf{w}_{L-1})\left(\boldsymbol{\Sigma}_L^{(1)}(\mathbf{w}_{L-1})+\boldsymbol{\Gamma}\right)^{-1}\boldsymbol{\Gamma}\right)\,,
\end{equation}
where $\boldsymbol{\Gamma}={\diag}(\sigma^2_1,\dots,\sigma^2_N)={\diag}(e^{f^{(2)}_1},\dots,e^{f^{(2)}_N})\,.$
\end{proposition}
\begin{proof}
The proof is straightforward by noting that 
\begin{equation*}
p(\boldsymbol{\mu}|\log\boldsymbol{\sigma}^2, \{\mathbf{w}^{(p)}_l\}, \mathbf{y}, \mathbf{x}) \propto\prod_{i=1}^N p(y_i|\mu_i=f_i^{(1)},\sigma^2_i=\exp\{f_i^{(2)}\})\,p(\mathbf{f}^{(1)}|\mathbf{w}^{(1)}_{L-1},\dots,\mathbf{w}^{(P_{L-1})}_{L-1})\,,
\end{equation*}
where $p(y_i|\mu_i=f_i^{(1)}, \sigma_i^2=\exp\{f_i^{(2)}\})$ is a normal density, and $p(\mathbf{f}^{(1)}|\mathbf{w}^{(1)}_{L-1}, \dots, \mathbf{w}^{(P_{L-1})}_{L-1})$ is a multivariate normal density with zero mean and covariance matrix $\boldsymbol{\Sigma}_L^{(1)}(\mathbf{w}_{L-1})\,$.
\end{proof}

Applying Proposition~\ref{prop:hetero_simple} works well when $N$ is relatively small. However, for large $N$, the inversion of $\boldsymbol{\Sigma}_L^{(1)}(\mathbf{w}_{L-1}) + \boldsymbol{\Gamma}$ involved in~\eqref{eq:hetero_simple} can become computationally prohibitive. There are three scenarios in which $N$ may be large: (i) a large number of replicates, leading to $\sum_{i=1}^N S_i \gg N$ even when $N$ is moderate; (ii) a naturally large number of observations, i.e., large $N$ without replicates; and (iii) a large number of observations with varying degrees of replications, from small to large. To address this scalability challenge while retaining the desired sampling efficiency during imputation, Propositions~\ref{prop:hetero_case1}, \ref{prop:hetero_case2} and~\ref{prop:hetero_case3} present analytical forms of the posterior distribution $p(\boldsymbol{\mu} \mid \log \boldsymbol{\sigma}^2, \{\mathbf{w}^{(p)}_l\}, \mathbf{y}, \mathbf{x})$ for each of these three cases, leveraging linear algebra techniques and the Vecchia approximation.

\begin{proposition}[Small $N$, large $S_i$]
\label{prop:hetero_case1}
Let $\mathbf{Y}$ be a permutation of $[\mathbf{Y}_1,\dots,\mathbf{Y}_N]$, in which $\mathbf{Y}_i\in\mathbb{R}^{S_i}$ has $S_i$ observations that are repeatedly generated at the $i$-th input $\mathbf{x}_{i*}$, and that $Y_{i,1},\dots,Y_{i,S_i}$ are conditionally independent and distributed as $\mathcal{N}(\mu_i, \sigma_i^2)$ under the GDGP. The posterior distribution $p(\boldsymbol{\mu}|\log\boldsymbol{\sigma}^2, \{\mathbf{w}^{(p)}_l\}, \mathbf{y}, \mathbf{x})$ of $\boldsymbol{\mu} = \mathbf{F}^{(1)}$, given $\boldsymbol{\sigma}^2=\exp\{\mathbf{f}^{(2)}\}$, the imputed latent variables $\{\mathbf{w}^{(p)}_l\}$, and the observed inputs $\mathbf{x}$ and outputs $\mathbf{y}$, is a multivariate normal distribution given by:
{\begin{equation}
\label{eq:hetero_case1}
\mathcal{N}\left(\boldsymbol{\Sigma}_L^{(1)}(\mathbf{w}_{L-1})\left(\mathbf{I}_N+\mathbf{M}^\top\boldsymbol{\Lambda}^{-1}\mathbf{M}\boldsymbol{\Sigma}_L^{(1)}(\mathbf{w}_{L-1})\right)^{-1}\mathbf{M}^\top\boldsymbol{\Lambda}^{-1}\mathbf{y},\, \boldsymbol{\Sigma}_L^{(1)}(\mathbf{w}_{L-1})\left(\mathbf{I}_N+\mathbf{M}^\top\boldsymbol{\Lambda}^{-1}\mathbf{M}\boldsymbol{\Sigma}_L^{(1)}(\mathbf{w}_{L-1})\right)^{-1}\right)\,,
\end{equation}}
where $\mathbf{I}_N\in\mathbb{R}^{N\times N}$ is the identity matrix; $\mathbf{M} \in\mathbb{R}^{(\sum_{i=1}^N S_i)\times N} $ is a replication matrix whose rows are standard basis vectors (i.e. each row contains a single 1 and zeros elsewhere), so that $\mathbf{Mx}$ reorders and replicates the entries of $\mathbf{x}$ to align with $\mathbf{Y}$; and {$\boldsymbol{\Lambda}={\diag}(\mathbf{M}\boldsymbol{\sigma}^2)$} with $\boldsymbol{\sigma}^2=(\sigma^2_1,\dots,\sigma^2_N)^\top=(e^{f^{(2)}_1},\dots,e^{f^{(2)}_N})^\top\,$.
\end{proposition}

\begin{proposition}[Large $N$, $S_i=1$]
\label{prop:hetero_case2}
Given $\mathbf{Y} \in \mathbb{R}^N$, where $Y_i$ for $i = 1, \dots, N$ are conditionally independent and distributed as $\mathcal{N}(\mu_i, \sigma_i^2)$, the posterior distribution $p(\boldsymbol{\mu} \mid \log \boldsymbol{\sigma}^2, \{\mathbf{w}^{(p)}_l\}, \mathbf{y}, \mathbf{x})$ of $\boldsymbol{\mu} = \mathbf{F}^{(1)}$ under the Vecchia approximation, given $\boldsymbol{\sigma}^2 = \exp\{\mathbf{f}^{(2)}\}$, the imputed latent variables $\{\mathbf{w}^{(p)}_l\}$, and the observed inputs $\mathbf{x}$ and outputs $\mathbf{y}$, is a multivariate normal distribution given by:
{\begin{equation}
\label{eq:hetero_case2}
\mathcal{N}\left(-\left(\mathbf{U}_{\mathbf{F}^{(1)}\mathbf{F}^{(1)}}^\top\right)^{-1}\mathbf{U}_{\mathbf{Y}\mathbf{F}^{(1)}}^\top\mathbf{y},\,  \left(\mathbf{U}_{\mathbf{F}^{(1)}\mathbf{F}^{(1)}}\mathbf{U}_{\mathbf{F}^{(1)}\mathbf{F}^{(1)}}^\top\right)^{-1}\right)\,,
\end{equation}}
where $\mathbf{U}_{\mathbf{F}^{(1)}\mathbf{F}^{(1)}}$ and $\mathbf{U}_{\mathbf{Y}\mathbf{F}^{(1)}}$ are block components of a sparse upper-triangular matrix
\begin{equation*}
  \mathbf{U}=\begin{bmatrix}
  \mathbf{U}_{\mathbf{Y}\mathbf{Y}} & \mathbf{U}_{\mathbf{Y}\mathbf{F}^{(1)}} \\
  \mathbf{0} & \mathbf{U}_{\mathbf{F}^{(1)}\mathbf{F}^{(1)}}
  \end{bmatrix}\,,
\end{equation*}
which is the upper-lower decomposition of the precision matrix $\mathbf{P}$ for the Vecchia-approximated joint distribution of $\mathbf{Z}=\begin{pmatrix}\mathbf{Y}|\mathbf{W}_{L-1},\log\boldsymbol{\sigma}^2\\ \mathbf{F}^{(1)}|\mathbf{W}_{L-1}\end{pmatrix}$, such that $\mathbf{P} = \mathbf{U}\mathbf{U}^\top$. Specifically, the $ji$-th element of $\mathbf{U}_{*\mathbf{F}^{(1)}}=[\mathbf{U}_{\mathbf{Y}\mathbf{F}^{(1)}}^\top,\mathbf{U}_{\mathbf{F}^{(1)}\mathbf{F}^{(1)}}^\top]^\top\in\mathbb{R}^{2N\times N}$ is given by:
\begin{equation}
\label{eq:joint_u_matrix}
\mathbf{U}_{*\mathbf{F}^{(1)},ji}=\begin{dcases}
\left(d_i\right)^{-\frac{1}{2}}, & j=i+N\,,\\
-b_{ik}\left(d_{i}\right)^{-\frac{1}{2}}, & j\in g(-(i+N))\,,\\
0\,, & \mathrm{otherwise}\,,
\end{dcases}
\end{equation}
with
\begin{itemize}
    \item $d_{i}=(\sigma_L^{(1)})^2(1+\eta)-\mathbf{r}_{i}^{\top}\boldsymbol{\Sigma}_{g(-(i+N))}^{-1}\mathbf{r}_{i}$\,;
    \item $\mathbf{b}_{i}=\mathbf{r}_{i}^{\top}\boldsymbol{\Sigma}_{g(-(i+N))}^{-1}$\,,
\end{itemize} 
where $g(-(i+N))\subseteq\{1,2,\dots,i+N-1\}$ is a conditioning index set of size $
\min(M, i+N-1)$ for all $i=1,\dots,N$; $\mathbf{r}_{i}=(\sigma_L^{(1)})^2[k^{(1)}_{L}(\mathbf{w}_{L-1,i*},[\mathbf{w}_{L-1}^\top,\mathbf{w}_{L-1}^\top]^\top_{j*})]^{\top}_{j\in g(-(i+N))}\in\mathbb{R}^{|g(-(i+N))|\times 1}$; $\boldsymbol{\Sigma}_{g(-(i+N))}$ is the covariance matrix of $\mathbf{Z}_{g(-(i+N))}$; and $b_{ik}$ is the $k$-th element of $\mathbf{b}_{i}$ with $k=\mathrm{pos}_{g(-(i+N))}(j)$ denoting the position of $j$ in the conditioning set $g(-(i+N))\,$.
\end{proposition}

\begin{proposition}[Large $N$, $S_i>1$]
\label{prop:hetero_case3}
Let $\mathbf{Y}$ be a permutation of $[\mathbf{Y}_1,\dots,\mathbf{Y}_N]$, in which $\mathbf{Y}_i\in\mathbb{R}^{S_i}$ has $S_i$ observations that are repeatedly generated at the $i$-th input $\mathbf{x}_{i*}$, and that $Y_{i,1},\dots,Y_{i,S_i}$ are conditionally independent and distributed as $\mathcal{N}(\mu_i, \sigma_i^2)$ under the GDGP. Define $\tilde{\mathbf{Y}}=(\mathbf{M}^T\boldsymbol{\Lambda}^{-1}\mathbf{M})^{-1}\mathbf{M}^T\boldsymbol{\Lambda}^{-1}\mathbf{Y}$, where $\mathbf{M} \in\mathbb{R}^{(\sum_{i=1}^N S_i)\times N} $ is a replication matrix whose rows are standard basis vectors, so that $\mathbf{Mx}$ reorders and replicates the entries of $\mathbf{x}$ to align with $\mathbf{Y}$; and {$\boldsymbol{\Lambda}={\diag}(\mathbf{M}\boldsymbol{\sigma}^2)$} with $\boldsymbol{\sigma}^2=(\sigma^2_1,\dots,\sigma^2_N)^\top=(e^{f^{(2)}_1},\dots,e^{f^{(2)}_N})^\top\,$. Then, the posterior distribution $p(\boldsymbol{\mu} \mid \log \boldsymbol{\sigma}^2, \{\mathbf{w}^{(p)}_l\}, \mathbf{y}, \mathbf{x})$ of $\boldsymbol{\mu} = \mathbf{F}^{(1)}$ under the Vecchia approximation, given $\boldsymbol{\sigma}^2 = \exp\{\mathbf{f}^{(2)}\}$, the imputed latent variables $\{\mathbf{w}^{(p)}_l\}$, and the observed inputs $\mathbf{x}$ and outputs $\mathbf{y}$, is a multivariate normal distribution given by:
{\begin{equation}
\label{eq:hetero_case3}
\mathcal{N}\left(-\left(\mathbf{U}_{\mathbf{F}^{(1)}\mathbf{F}^{(1)}}^\top\right)^{-1}\mathbf{U}_{\tilde{\mathbf{Y}}\mathbf{F}^{(1)}}^\top\tilde{\mathbf{y}},\,  \left(\mathbf{U}_{\mathbf{F}^{(1)}\mathbf{F}^{(1)}}\mathbf{U}_{\mathbf{F}^{(1)}\mathbf{F}^{(1)}}^\top\right)^{-1}\right)\,,
\end{equation}}
where $\tilde{\mathbf{y}} = (\mathbf{M}^T\boldsymbol{\Lambda}^{-1}\mathbf{M})^{-1}\mathbf{M}^T\boldsymbol{\Lambda}^{-1}\mathbf{y}$; $\mathbf{U}_{\mathbf{F}^{(1)}\mathbf{F}^{(1)}}$ and $\mathbf{U}_{\tilde{\mathbf{Y}}\mathbf{F}^{(1)}}$ are block components of a sparse upper-triangular matrix
\begin{equation*}
  \mathbf{U}=\begin{bmatrix}
  \mathbf{U}_{\tilde{\mathbf{Y}}\tilde{\mathbf{Y}}} & \mathbf{U}_{\tilde{\mathbf{Y}}\mathbf{F}^{(1)}} \\
  \mathbf{0} & \mathbf{U}_{\mathbf{F}^{(1)}\mathbf{F}^{(1)}}
  \end{bmatrix}\,,
\end{equation*}
which is the upper-lower decomposition of the precision matrix $\mathbf{P}$ for the Vecchia-approximated joint distribution of $\mathbf{Z}=\begin{pmatrix}\tilde{\mathbf{Y}}|\mathbf{W}_{L-1},\log\boldsymbol{\sigma}^2\\ \mathbf{F}^{(1)}|\mathbf{W}_{L-1}\end{pmatrix}\,$, such that $\mathbf{P} = \mathbf{U}\mathbf{U}^\top$. The $ji$-th element of $\mathbf{U}_{*\mathbf{F}^{(1)}}=[\mathbf{U}_{\tilde{\mathbf{Y}}\mathbf{F}^{(1)}}^\top,\mathbf{U}_{\mathbf{F}^{(1)}\mathbf{F}^{(1)}}^\top]^\top\in\mathbb{R}^{2N\times N}$ is obtained following the Equation~\eqref{eq:joint_u_matrix}.
\end{proposition}

\begin{proof}
The proofs of Propositions~\ref{prop:hetero_case1}, ~\ref{prop:hetero_case2} and~\ref{prop:hetero_case3} are given in Appendices~\ref{app:hetero_case1}, \ref{app:hetero_case2}, and \ref{app:hetero_case3}, respectively.
\end{proof}

Note that Proposition~\ref{prop:hetero_case1} enables analytical sampling from $p(\boldsymbol{\mu} \mid \log \boldsymbol{\sigma}^2, \{\mathbf{w}^{(p)}_l\}, \mathbf{y}, \mathbf{x})$ with computational complexity $\mathcal{O}(N^3)$. This is computationally feasible when $N$ is small, even in the presence of a large number of replicates. When $N$ is large, in contrast, Propositions~\ref{prop:hetero_case2} and~\ref{prop:hetero_case3} admit analytical sampling with complexity $\mathcal{O}(NM^3+NM)$, regardless of the number of replicates. If Proposition~\ref{prop:hetero_simple} were applied naively in these settings, the sampling complexity would scale as $\mathcal{O}((\sum_{i=1}^N S_i)^3)$, which is computationally prohibitive, particularly when both 
$N$ and the numbers of replicates $S_i$ are large. 

The Vecchia approximations underlying Propositions~\ref{prop:hetero_case2} and~\ref{prop:hetero_case3} are constructed using the \emph{response-first} ordering of~\citet{katzfuss2020vecchia}, in which $\mathbf{Y}$ and $\tilde{\mathbf{Y}}$ are always ordered before $\mathbf{F}^{(1)}$, respectively. Following~\citet{katzfuss2020vecchia}, in implementation the conditioning index set $g(-(i+N))$ is formed by selecting up to $M$ variables in $\mathbf{Z}$ that precede ${Z}_{i+N}$ and whose corresponding inputs in $\mathbf{w}_{L-1}$ are nearest neighbors of $\mathbf{w}_{L-1,i*}$. When a nearest neighbor corresponds to two candidate conditioning variables in $\mathbf{Z}$, which occurs because, conditional on $\mathbf{W}_{L-1}$, $\mathbf{Y}$ and $\tilde{\mathbf{Y}}$ have the same input locations as $\mathbf{F}^{(1)}$, preference is given to the conditioning variable in $\mathbf{F}^{(1)}$.

\section{Experiments}
\label{sec:experiments}
In this section, we compare the emulation performance of GDGP under a selected set of likelihoods with several competing approaches across a series of experiments. Specifically, for GDGP we employ a two-layer DGP architecture (i.e., $L=2$), in which the first layer contains a number of GP nodes equal to the input dimension, and the second layer contains a number of GP nodes equal to the number of parameters of the associated likelihood. Each GP node uses an isotropic Matérn-2.5 kernel. This two-layer DGP architecture has been shown~\citep{sauer2023active, ming2021deep} to provide a good balance between computational efficiency and the modeling flexibility of DGPs, and is used in all subsequent experiments. For consistency, all GP-based competitors are also specified using the Matérn-2.5 kernel. 

{For GDGP inference, following \citet{ming2021deep}, we set $T=500$ and $B=0.75T$, thereby discarding the first $75\%$ of the SEM iterations and averaging over the remaining post-burn-in iterations to obtain the hyperparameter estimates. For the number of ESS burn-in steps $C$ at each imputation step of SEM, a large value is generally unnecessary in practice~\citep{ming2021deep}, as each imputation step of SEM requires only one representative realization from the conditional posterior distribution of the latent variables rather than a full characterization of the distribution; we therefore set $C=10$. Finally, we set the number of imputations $K$ to $10$ to balance the representation of imputation uncertainty against computational cost. These settings remain unchanged throughout the experiments unless otherwise stated, while the NN conditioning set size $M$ is reported separately for each experiment in which the Vecchia approximation is used. As a practical guideline for users, however, if the inference is too computationally demanding or speed is a priority, one can decrease $T$ and/or $K$ to reduce computation. In our experience, for some problems, settings such as $T=200$ and $K=5$ could still yield satisfactory predictive performance while enabling faster inference.}

All GDGP emulators are implemented using the \texttt{R} package \texttt{dgpsi} ({version 2.6.0-9000}), available at \url{https://github.com/mingdeyu/dgpsi-R}, and all experiments are conducted on a Mac Studio with a 24-core Apple M2 Ultra processor and 128\,GB of RAM.

\subsection{Heteroskedastic Gaussian Likelihood}
In this section, we demonstrate the capabilities of GDGP under a heteroskedastic Gaussian likelihood. For benchmarking, we compare GDGP with specialized heteroskedastic Gaussian process models, including the maximum-likelihood-based approach
of \citet{binois2018practical}, hereinafter denoted as \emph{hetGP} and implemented in the \texttt{R} package \texttt{hetGP} ({version~1.1.9}; \citealp{binois2021hetgp}), and the fully Bayesian framework of \citet{patil2025vecchia}, hereinafter denoted as \emph{bhetGP} and implemented in the \texttt{R} package \texttt{bhetGP} ({version~1.0.2}). We have chosen functions that exhibit non-stationarity to demonstrate the efficacy of the GDGP framework, which is designed for these settings.

\subsubsection{Heteroskedastic step function}
\label{sec:1d-het}
Consider a synthetic simulator whose output $y$ at a given location $x$ follows a Gaussian distribution $\mathcal{N}(\mu(x), \sigma^2(x))$, where
\begin{equation*}
\mu(x) =
\begin{cases}
-1, & x < 0.5\,,\\
\;\,1, & x \ge 0.5\,,
\end{cases}
\end{equation*}
and
\begin{equation*}
\sigma^2(x) = 
\frac{1}{600}
\left[
\sin(4x-2) +10\exp\{-1200(2x-1)^2\} +1
\right]\,.
\end{equation*}

In this experiment, we evaluate the performance of different models by assessing their ability to reconstruct the true mean function $\mu(x)$ and log-variance function $\log \sigma^2(x)$ when $x\in[0,1]$. Performance is quantified using the normalized root mean squared error (NRMSE; \citealp{ming2021deep}, Section~4) and the negative continuous ranked probability score (NCRPS; \citealp{gneiting2007strictly}, Section~4.2). The former measures the deterministic accuracy of model predictions, while the latter evaluates the quality of uncertainty quantification. 

Specifically, we select $100$ input locations uniformly spaced over $[0,1]$ as the training inputs. {Using this training design, we first present an illustrative comparison based on $R=60$ simulator replicates at each location. Figure~\ref{fig:figure_2} displays the predicted mean and log-variance functions, together with their predictive intervals, for the three models. For the mean function, GDGP provides more appropriate uncertainty quantification across the domain. By contrast, hetGP and bhetGP produce ``sausage-shaped'' predictive intervals, consistent with their stationary GP structures, which is also reflected in their higher NCRPS values relative to GDGP. For the log-variance function, GDGP and hetGP better capture the sharp localized dip around $x=0.1$ than bhetGP. However, both hetGP and bhetGP appear to produce overly narrow predictive intervals, whereas those from GDGP provide more adequate coverage of the true function.}

\begin{figure}
\centering 
\subfloat[hetGP]{\label{fig:figure_2a}\includegraphics[width=0.33\linewidth]{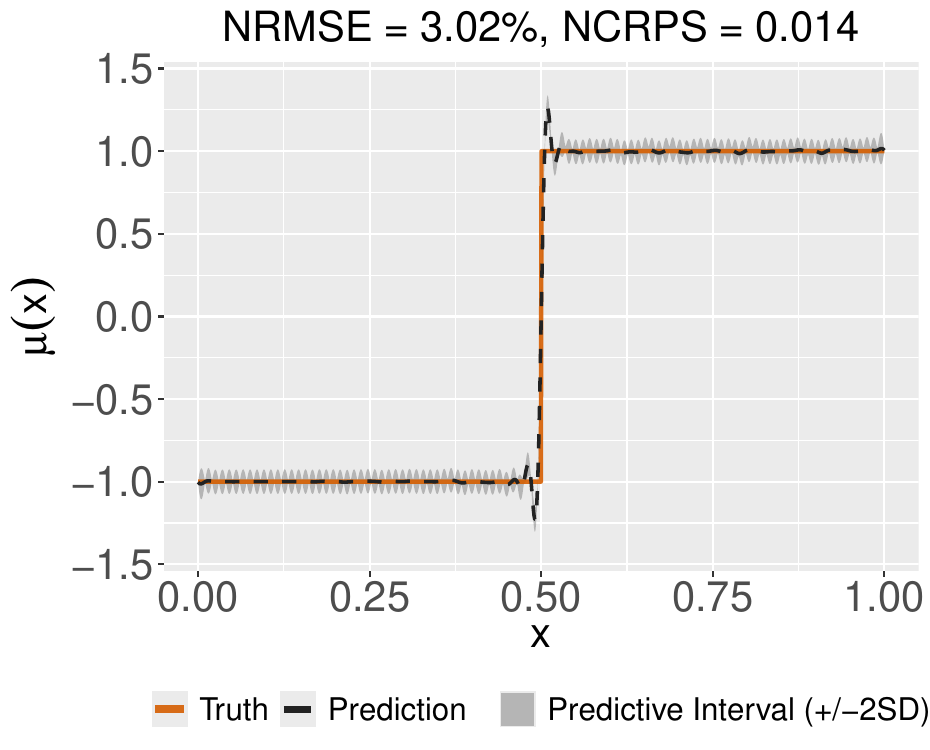}}
\subfloat[bhetGP]{\label{fig:figure_2b}\includegraphics[width=0.33\linewidth]{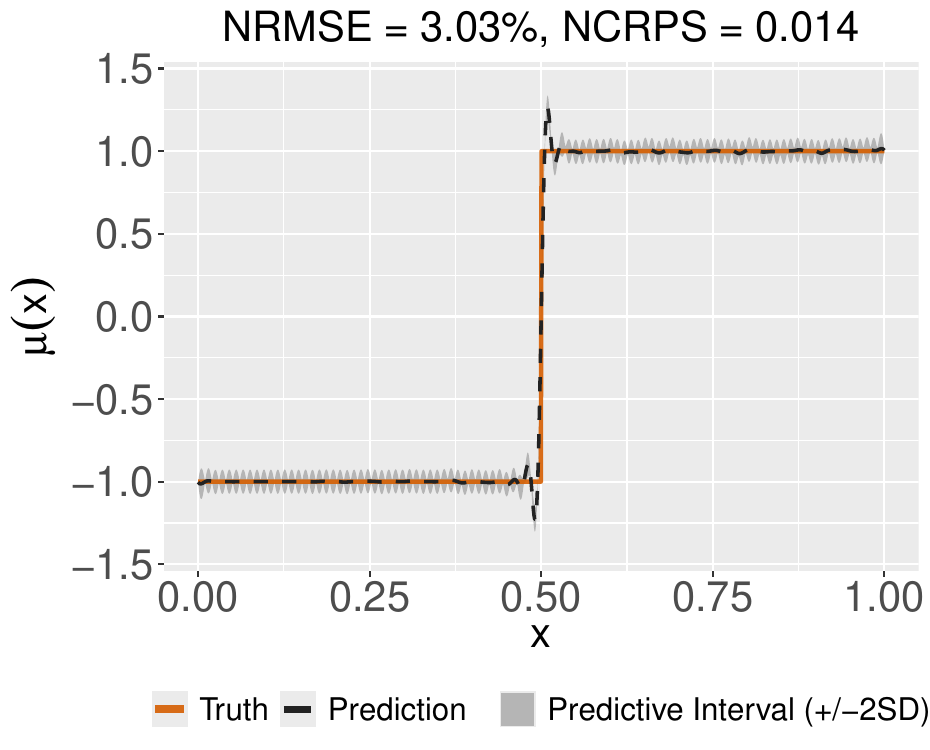}}
\subfloat[GDGP]{\label{fig:figure_2c}\includegraphics[width=0.33\linewidth]{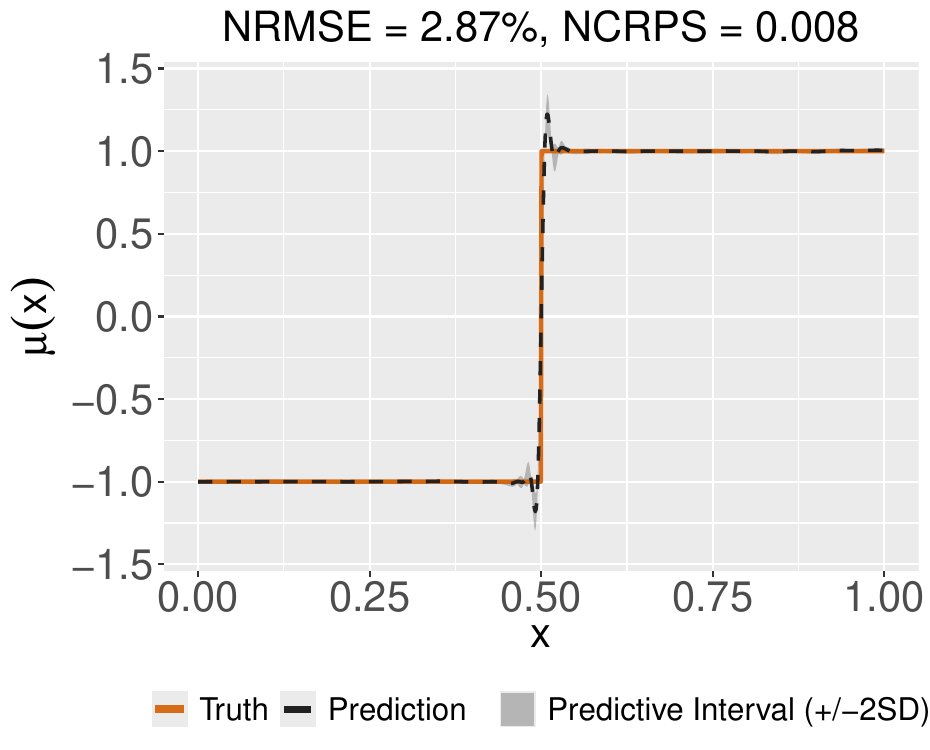}}\\
\subfloat[hetGP]{\label{fig:figure_2d}\includegraphics[width=0.33\linewidth]{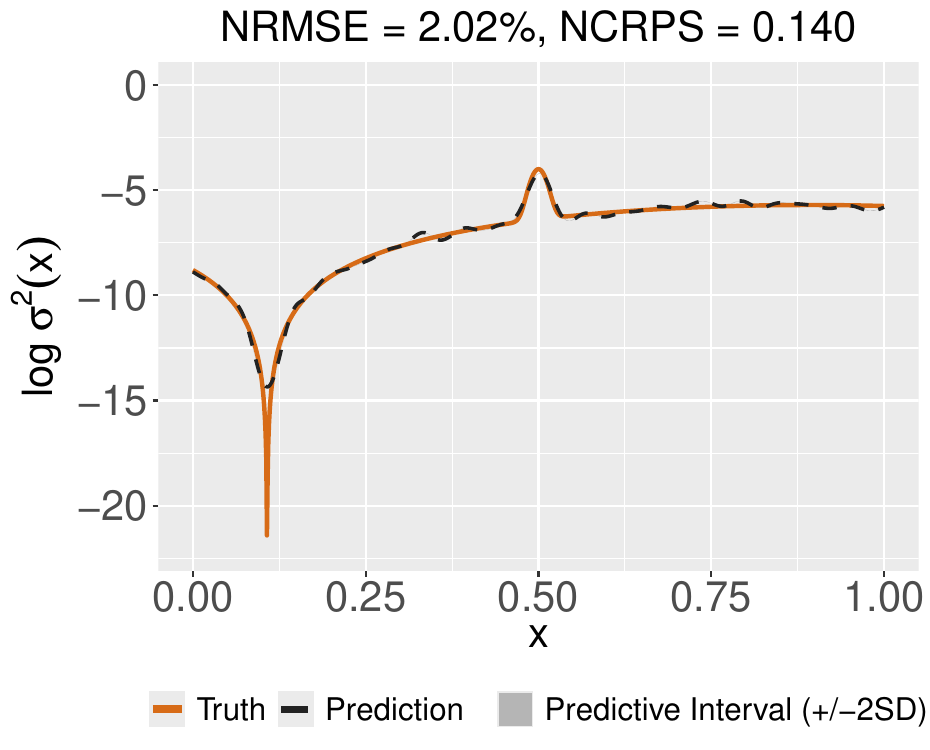}}
\subfloat[bhetGP]{\label{fig:figure_2e}\includegraphics[width=0.33\linewidth]{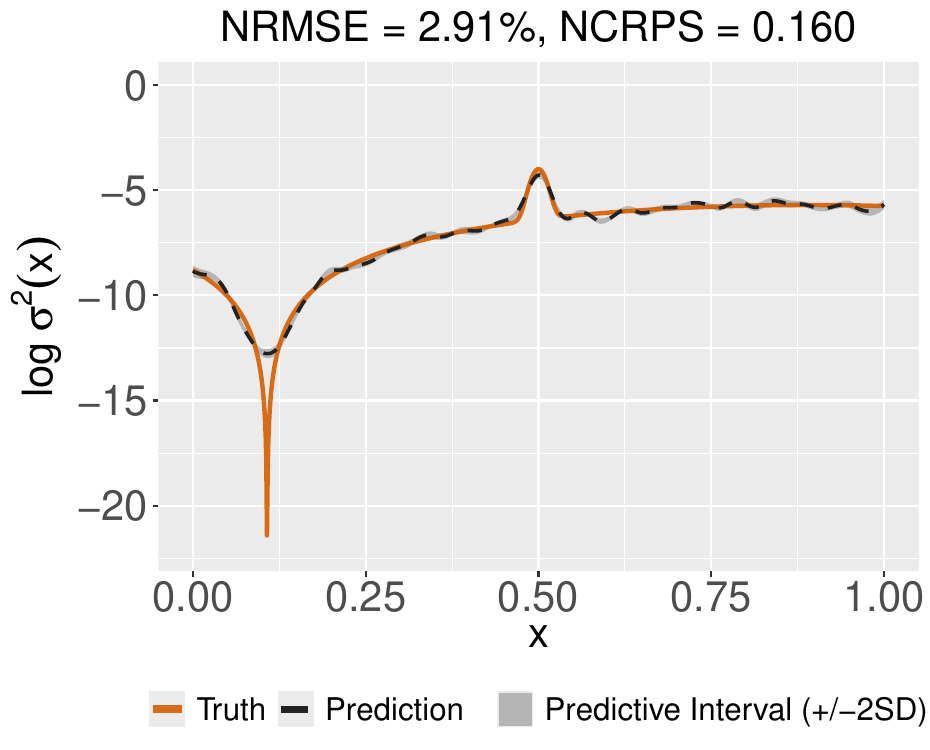}}
\subfloat[GDGP]{\label{fig:figure_2f}\includegraphics[width=0.33\linewidth]{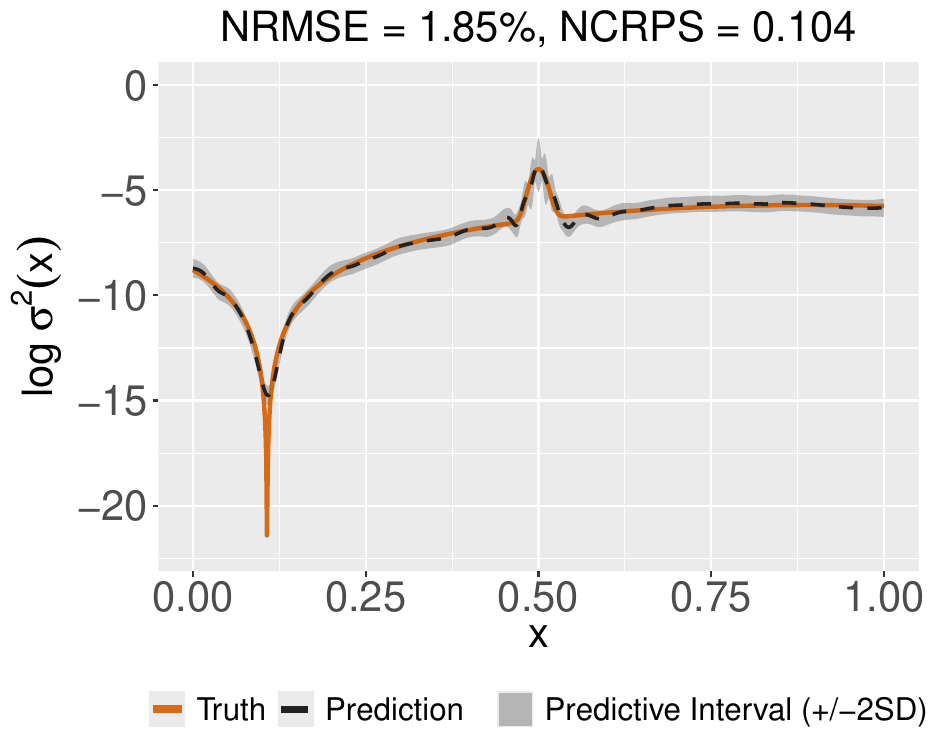}}
\caption{{Emulations of the mean function (top row) and log-variance function (bottom row) of the synthetic simulator described in Section~\ref{sec:1d-het}, obtained using hetGP, bhetGP, and GDGP. All models are trained on $100$ evenly spaced input locations over $[0,1]$, with $60$ replicates per location. NRMSE denotes the normalized root mean squared error, and NCRPS denotes the negative continuous ranked probability score; lower values indicate better performance for both metrics. Both metrics are evaluated using $1{,}000$ validation points with input locations evenly spaced over $[0,1]$.}}
\label{fig:figure_2}
\end{figure}

For a more systematic comparison, we then vary the number of simulator replicates per training location over $R\in\{20,40,60,80,100\}$. For each value of $R$, we independently regenerate the simulator outputs, train all three models, and repeat this procedure $100$ times. NRMSE and NCRPS for $\mu(x)$ and $\log\sigma^2(x)$ are evaluated using $1{,}000$ evenly spaced test locations over $[0,1]$. {The distributions of these metrics and the mean computation times across the $100$ trials for each method and each value of $R$ are summarized in Figure~\ref{fig:figure_3} and Table~\ref{tab:tab_1}, respectively. Naturally, the computational cost of a deep GP method is greater than that of a shallow GP method. Thanks to SI, however, despite its deeper hierarchy, GDGP has a computational cost comparable to that of the bhetGP, and its additional cost over hetGP remains manageable at approximately $13$ seconds per fit on average. Notably, the mean computation time of GDGP increases only slightly as the number of replicates rises from $20$ to $100$, consistent with the treatment of replicates introduced in Section~\ref{sec:replicates} and, in particular, with the result of Proposition~\ref{prop:hetero_case1} in Section~\ref{sec:hetero}. The results in Figure~\ref{fig:figure_3} show that GDGP achieves the best overall performance in terms of both NRMSE and NCRPS. Together with Table~\ref{tab:tab_1}, these results demonstrate the gains offered by GDGP in predictive accuracy and uncertainty quantification at only a modest additional computational cost.} In addition, for GDGP, both NRMSE and NCRPS decrease steadily as the number of replicates increases. In particular, for the mean function $\mu(x)$, GDGP substantially outperforms both hetGP and bhetGP, owing to its ability to capture non-stationary structures, such as the step function in this example. For the log-variance function $\log\sigma^2(x)$, both hetGP and GDGP achieve lower NRMSE and NCRPS than bhetGP across all replicate settings. Although hetGP attains lower NRMSE than GDGP when the number of replicates is small (i.e., $R=20$ and $40$, likely due to SI finding it hard to identify weaker non-stationarity in the log-variance process when the number of replicates is small), GDGP yields lower NRMSE on average as the number of replicates increases, and consistently achieves lower NCRPS across different numbers of replicates. 

\begin{figure}
\centering 
\subfloat[NRMSE of $\mu(x)$]{\label{fig:figure_3a}\includegraphics[width=0.5\linewidth]{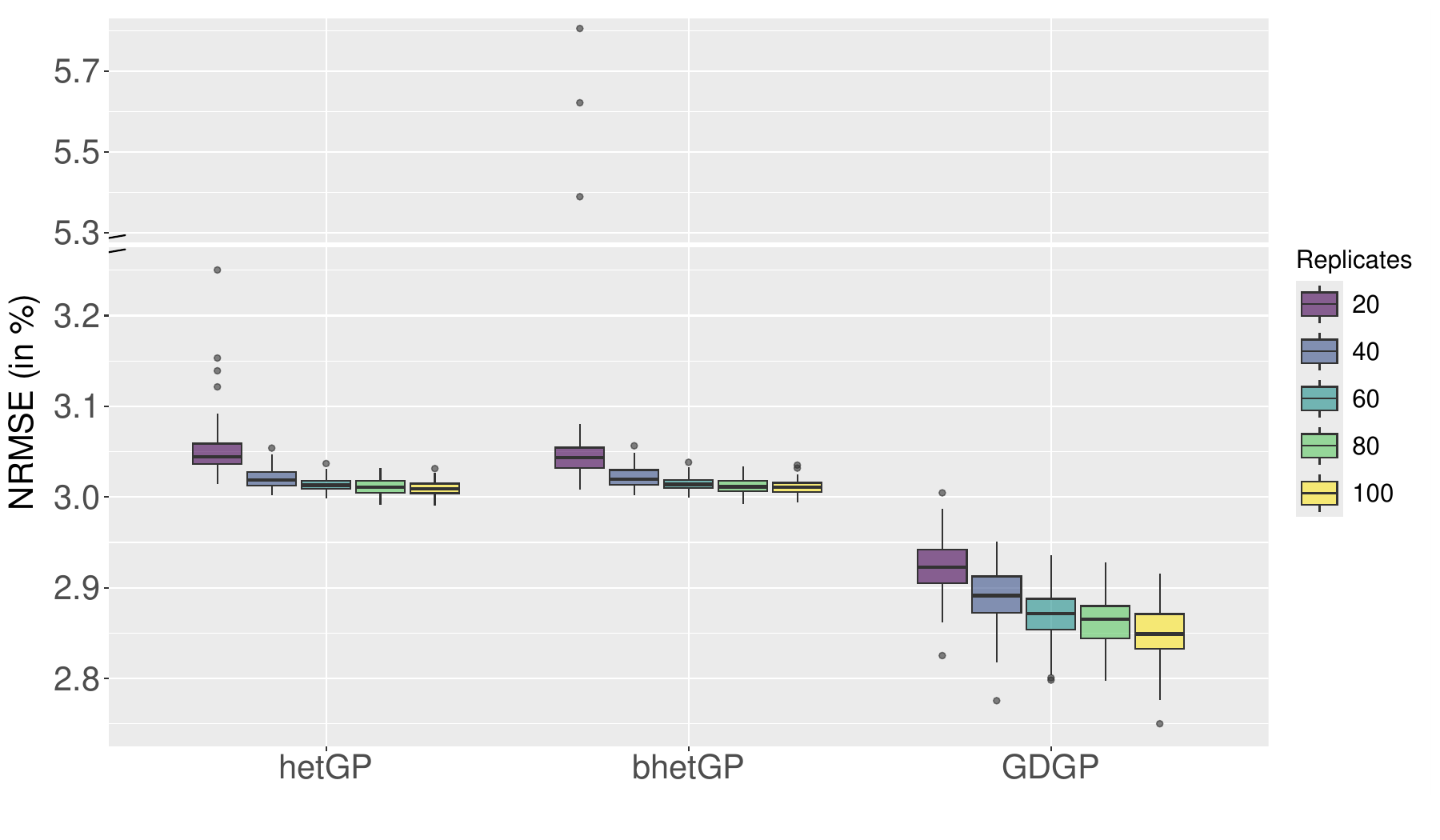}}
\subfloat[NCRPS of $\mu(x)$]{\label{fig:figure_3b}\includegraphics[width=0.5\linewidth]{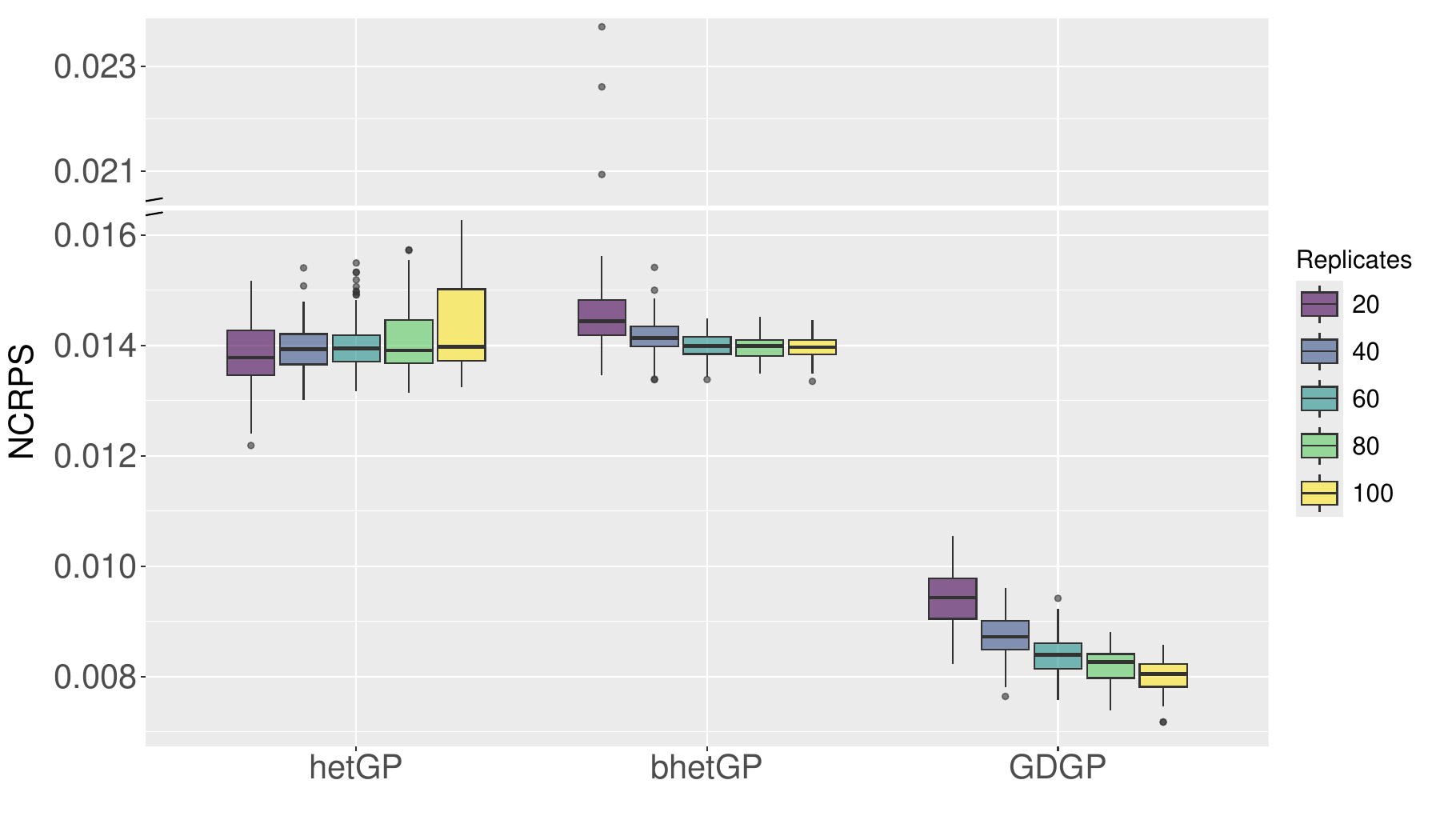}}\\
\subfloat[NRMSE of $\log\sigma^2(x)$]{\label{fig:figure_3c}\includegraphics[width=0.5\linewidth]{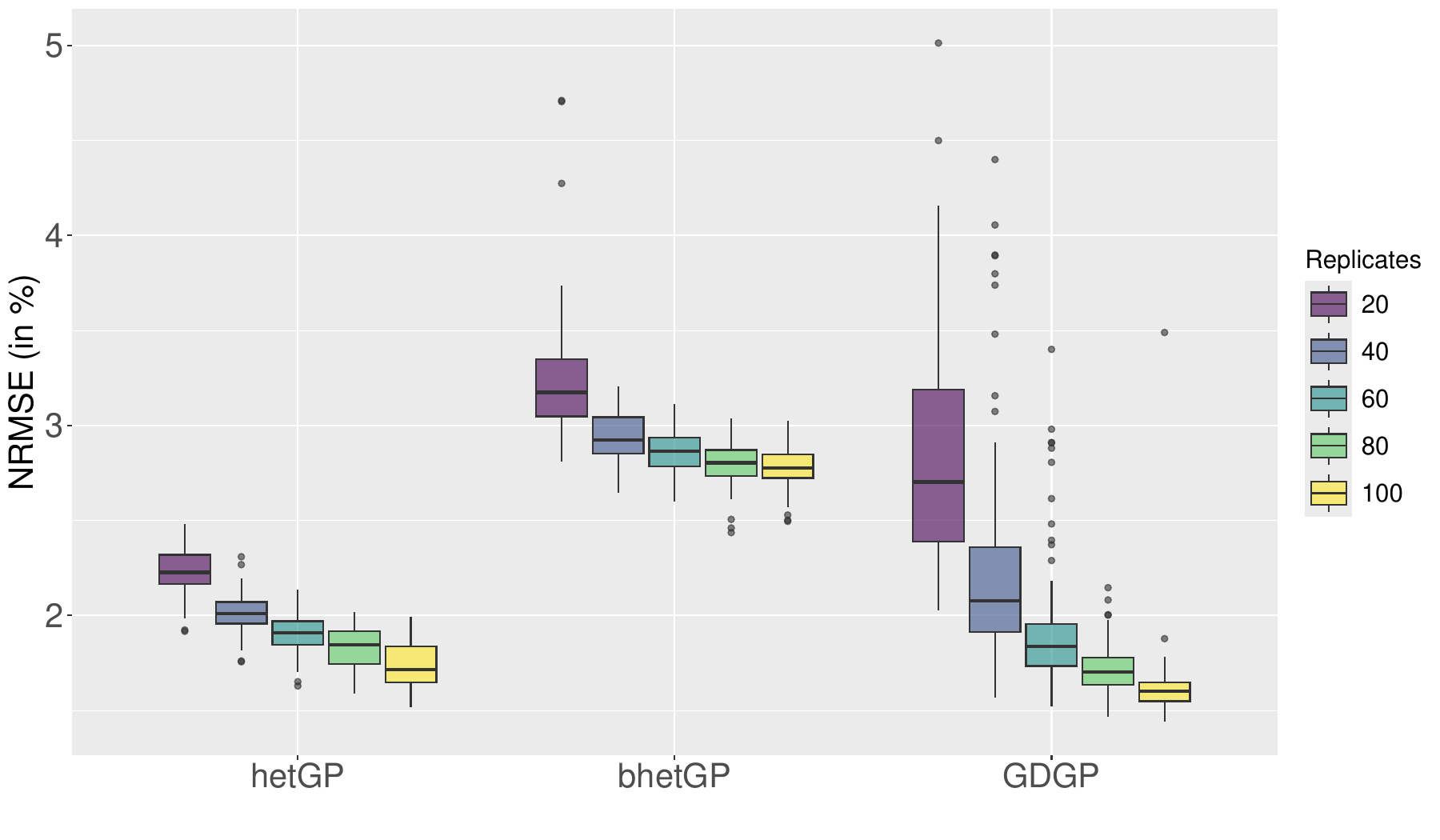}}
\subfloat[NCRPS of $\log\sigma^2(x)$]{\label{fig:figure_3d}\includegraphics[width=0.5\linewidth]{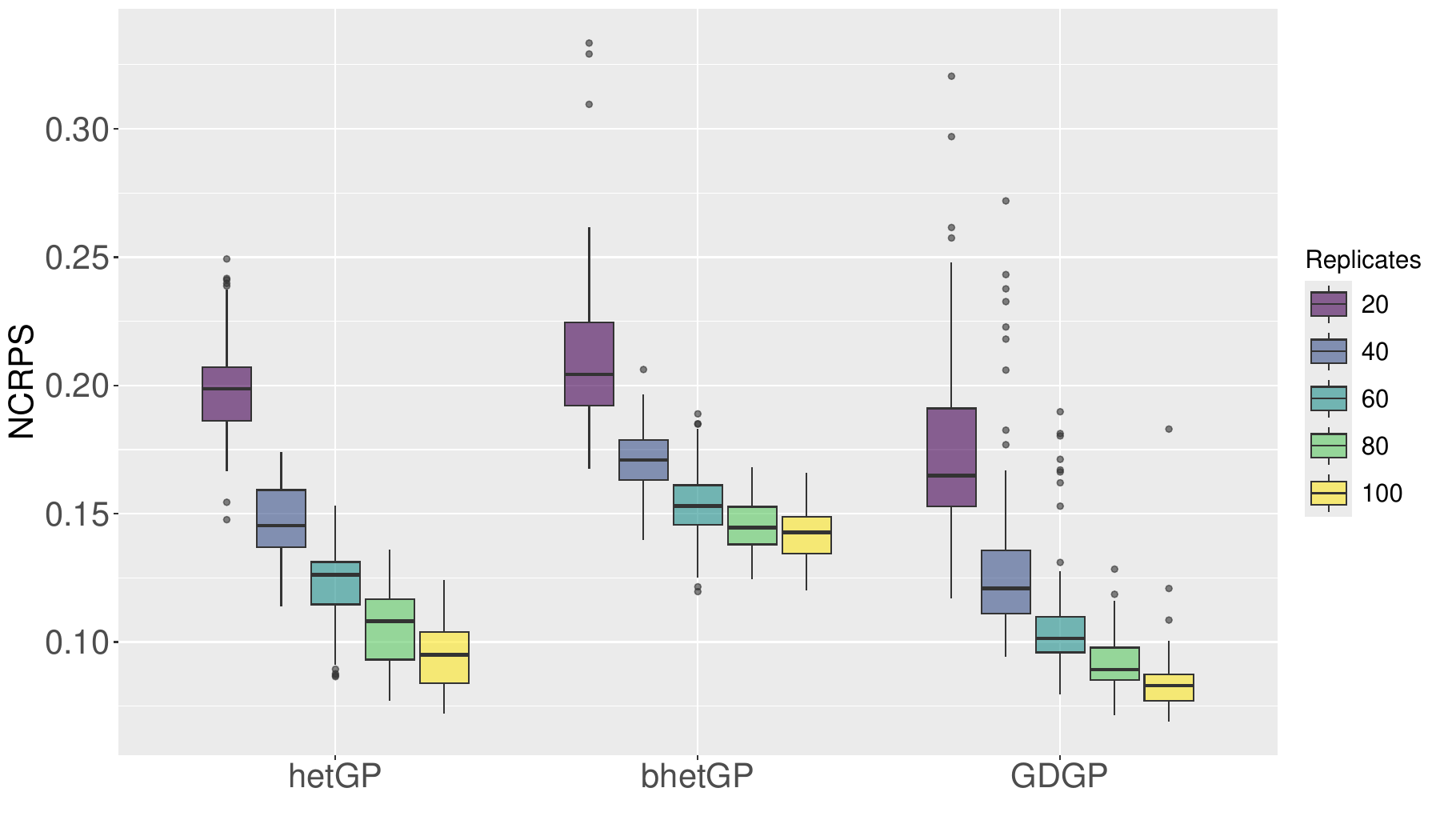}}
\caption{NRMSEs (lower is better) and NCRPSs (lower is better) for hetGP, bhetGP, and GDGP, trained on $100$ unique input locations with $R\in\{20,40,60,80,100\}$ replicates per location, for emulating the mean function $\mu(x)$ and log-variance function $\log \sigma^2(x)$ of the synthetic simulator described in Section~\ref{sec:1d-het}. Results are evaluated using $1{,}000$ validation points and summarized over $100$ independent training trials.}
\label{fig:figure_3}
\end{figure}

\begin{table}
\caption{{Mean computation times (in seconds), averaged over $100$ independent training trials, for hetGP, bhetGP, and GDGP trained on $100$ unique input locations with $R\in\{20,40,60,80,100\}$ replicates per location from the synthetic simulator described in Section~\ref{sec:1d-het}.}}
\label{tab:tab_1}
\centering
\begin{booktabs}{
  colspec = {Q[c,m] Q[r,m] Q[r,m] Q[r,m] Q[r,m] Q[r,m]},
  cell{1}{1} = {r=2}{c,m,font=\bfseries},
  cell{1}{2} = {c=5}{c,m,font=\bfseries},
  cell{2}{2-6} = {c,m}
}
\toprule
Model & Number of Replicates & & & & \\
\cmidrule[lr]{2-6}
 & 20 & 40 & 60 & 80 & 100 \\
\midrule
hetGP & 0.30 & 0.29 & 0.29 & 0.31 & 0.26 \\
bhetGP & 11.42 & 11.71 & 12.11 & 12.59 & 12.94 \\
GDGP & 12.78 & 13.43 & 13.81 & 13.56 & 14.96 \\
\bottomrule
\end{booktabs}
\end{table}

\subsubsection{Susceptible, Infective, and Recovered (SIR) model}
\label{sec:sir}
In this experiment, we consider a stochastic Susceptible-Infective-Recovered (SIR) simulator obtained by extending the benchmark model in~\citet{binois2021hetgp} distributed with the \texttt{hetGP} package. Relative to the original two-input formulation, which varies only the initial susceptible and infected populations, the modified simulator is defined on a five-dimensional unit hypercube, with inputs determining the initial susceptible population ($S_0$), the initial infected population ($I_0$), the transmission rate ($\beta$), the recovery rate ($\gamma$), and the external connectivity index ($c_{\mathrm{ext}}$), which controls the external infectious pressure. When $c_{\mathrm{ext}} < 0.5$, the simulator yields an effectively isolated population and hence a closed-population epidemic regime with no importation, whereas when $c_{\mathrm{ext}} \ge 0.5$, it yields an externally seeded epidemic regime in which higher values of $c_{\mathrm{ext}}$ correspond to greater exogenous infection pressure. Consequently, the simulator combines smooth variation in epidemic parameters with a regime change associated with the onset of external seeding, and outputs the cumulative attack proportion ($I_{\mathrm{pop}}$) at the time horizon of $100$, thereby providing a useful test problem for evaluating emulator performance under non-stationarity.

We use increasing numbers of Latin hypercube design points, $n\in\{100,200,300,400,500\}$, as the unique training input locations. For each given $n$, the SIR simulator is evaluated with a randomly selected number of replicates between $1$ and $100$ at each input location. The resulting emulators are then assessed on a test set of size $N=60{,}000$, consisting of $2{,}000$ unique space-filling input locations with $30$ replicates per location. Performance is measured using the average scoring rule (Score; \citealp{gneiting2007strictly}, Equation~(27)):
\begin{equation*}
\mathrm{Score} =
-\frac{1}{N}\sum_{i=1}^N
\left[
\frac{(y_i-\tilde{\mu}(\mathbf{x}_i))^2}{\tilde{\sigma}^2(\mathbf{x}_i)}
+\log \tilde{\sigma}^2(\mathbf{x}_i)
\right]\,,
\end{equation*}
where $y_i$, for $i=1,\dots,N$, denote the test outputs, and $\tilde{\mu}(\mathbf{x}_i)$ and $\tilde{\sigma}^2(\mathbf{x}_i)$ are the predictive mean and variance produced by an emulator at the test input $\mathbf{x}_i$. For each $n$, we repeat the above procedure $20$ times.

{Figure~\ref{fig:figure_4} compares the predictive scores and empirical computation times of the three candidate models across different numbers of unique training locations. As the training size increases, the predictive performance of all three models improves. Notably, GDGP substantially outperforms both hetGP and bhetGP: with only $200$ training locations, GDGP achieves scores comparable to those obtained by hetGP and bhetGP using $500$ locations, and does so with a computation time comparable to that of hetGP and substantially lower than that of bhetGP (although hetGP naturally is still fastest when the number of locations is fixed across all three methods). Moreover, GDGP exhibits relatively less variability in scores across repeated trials.}

\begin{figure}
\centering 
\subfloat[]{\label{fig:figure_4a}\includegraphics[width=0.5\linewidth]{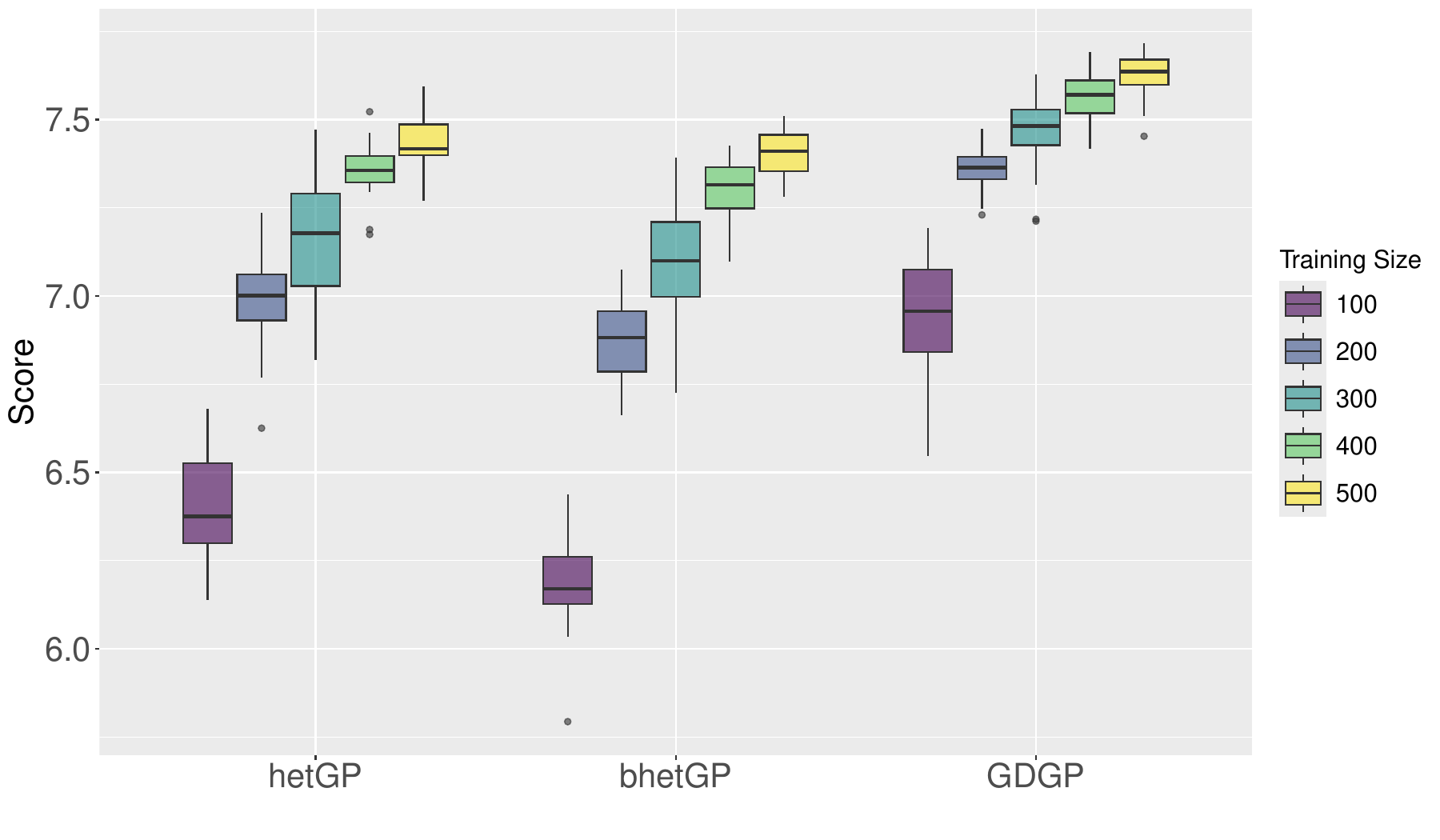}}
\subfloat[]{\label{fig:figure_4b}\includegraphics[width=0.5\linewidth]{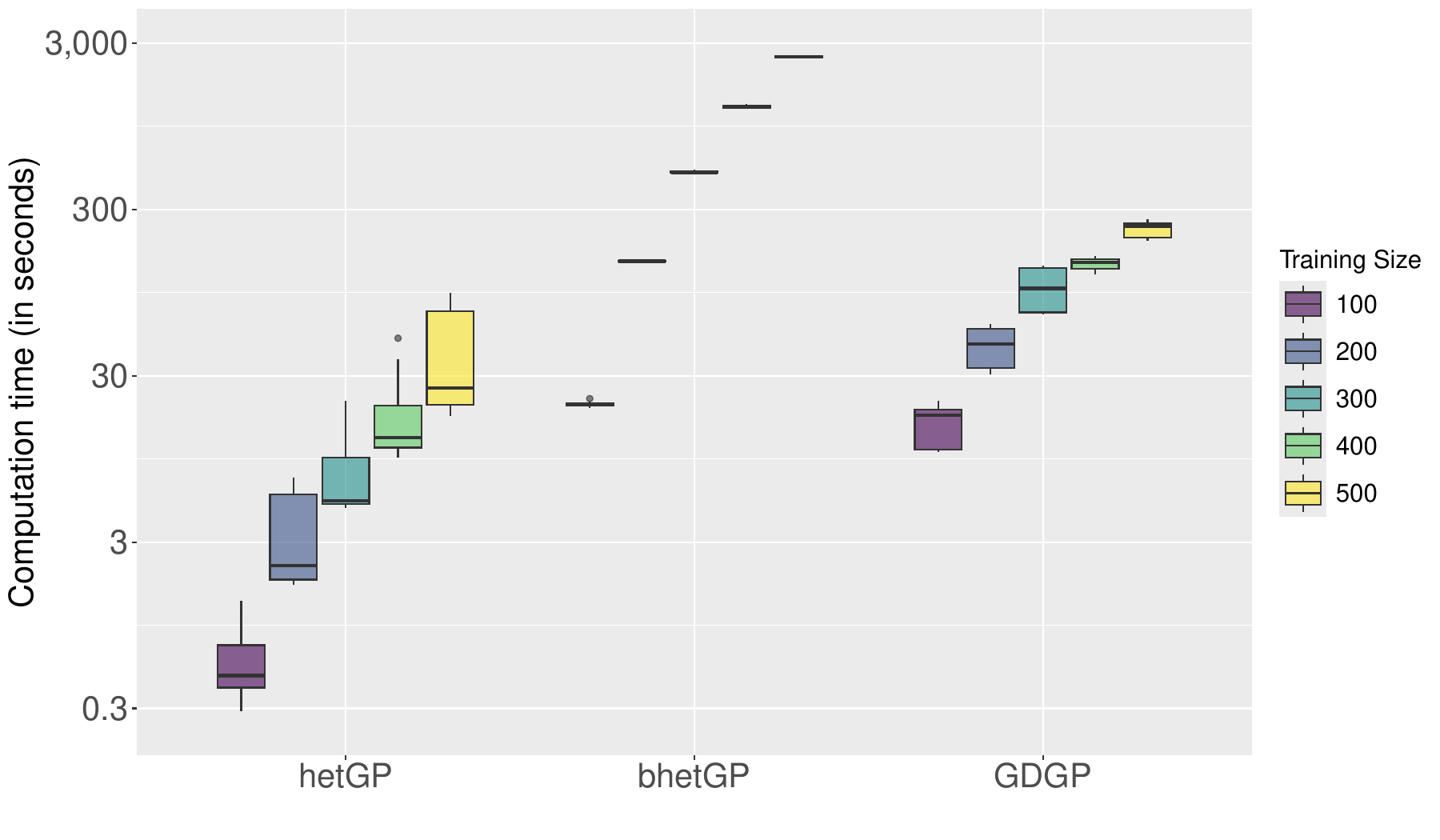}}
\caption{\emph{(a)} Scores (higher is better) and {\emph{(b)} computation times (in log-scale)} for hetGP, bhetGP, and GDGP trained on $n\in\{100,200,300,400,500\}$ unique input locations, with a random number of replicates between $1$ and $100$ at each location, for emulating the cumulative attack proportion of the SIR simulator in Section~\ref{sec:sir} at a time horizon of $100$. Scores are evaluated on a test set of size $N=60{,}000$, consisting of $2{,}000$ space-filling input locations with $30$ replicates at each location. Both scores and computation times are summarized over $20$ independent training trials for each $n$.}
\label{fig:figure_4}
\end{figure}

{To illustrate the computational benefits and predictive performance of our Vecchia-based scalability development for GDGP, which utilizes, in particular, Proposition~\ref{prop:hetero_case3} in Section~\ref{sec:hetero}, we compare the full GDGP with Vecchia-approximated GDGPs across five training conditioning set sizes in $\{5,15,25,50,75\}$ and ten training data sizes in $\{100,200,\dots,1000\}$, with each configuration evaluated over $20$ independent training trials. This results in $1{,}200$ model fits in total. To keep the overall training time manageable in this computationally extensive experiment, we therefore set the number of SEM iterations to $T=200$ for all GDGP and Vecchia-approximated GDGP emulators.} Figure~\ref{fig:figure_5} summarizes the results of this experiment. As shown in Figure~\subref*{fig:figure_5a}, increasing the conditioning set size for training (while keeping the conditioning set size for prediction sufficiently large at $200$) improves predictive performance, and the Vecchia-approximated GDGP approaches the full GDGP (i.e., the non-Vecchia version, in which all unique training inputs are included in the conditioning set and inference is implemented using Proposition~\ref{prop:hetero_case1} in Section~\ref{sec:hetero}). {Notably, conditioning set sizes of $25$ and $50$ already yield satisfactory performance close to that of the full GDGP across different numbers of training sizes.} Together with Figure~\subref*{fig:figure_5b}, we observe that the Vecchia-approximated GDGP with small conditioning set sizes can substantially reduce computation while achieving accuracy comparable to that of the full GDGP as the training size increases. 

\begin{figure}
\centering 
\subfloat[]{\label{fig:figure_5a}\includegraphics[width=0.5\linewidth]{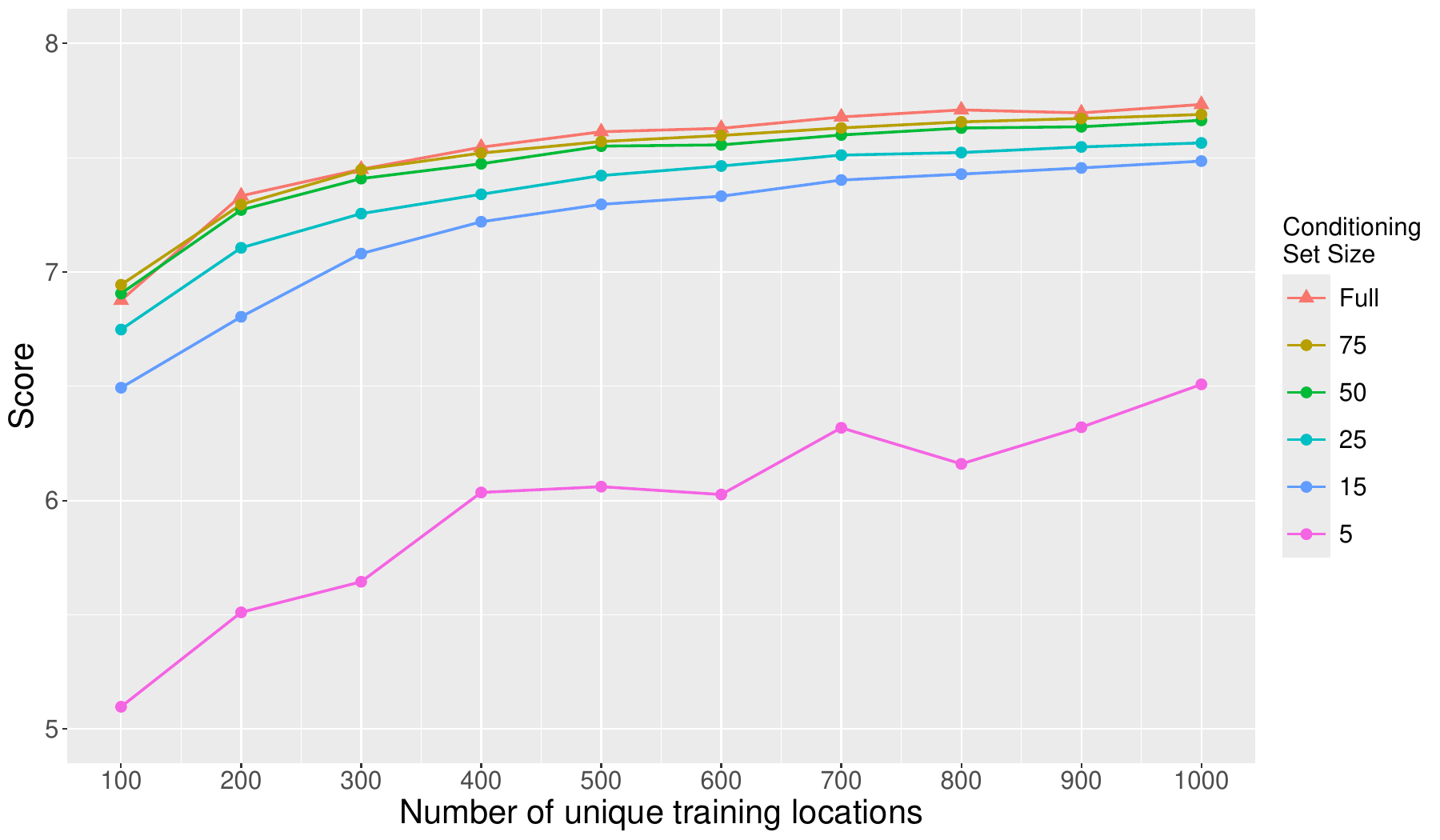}}
\subfloat[]{\label{fig:figure_5b}\includegraphics[width=0.5\linewidth]{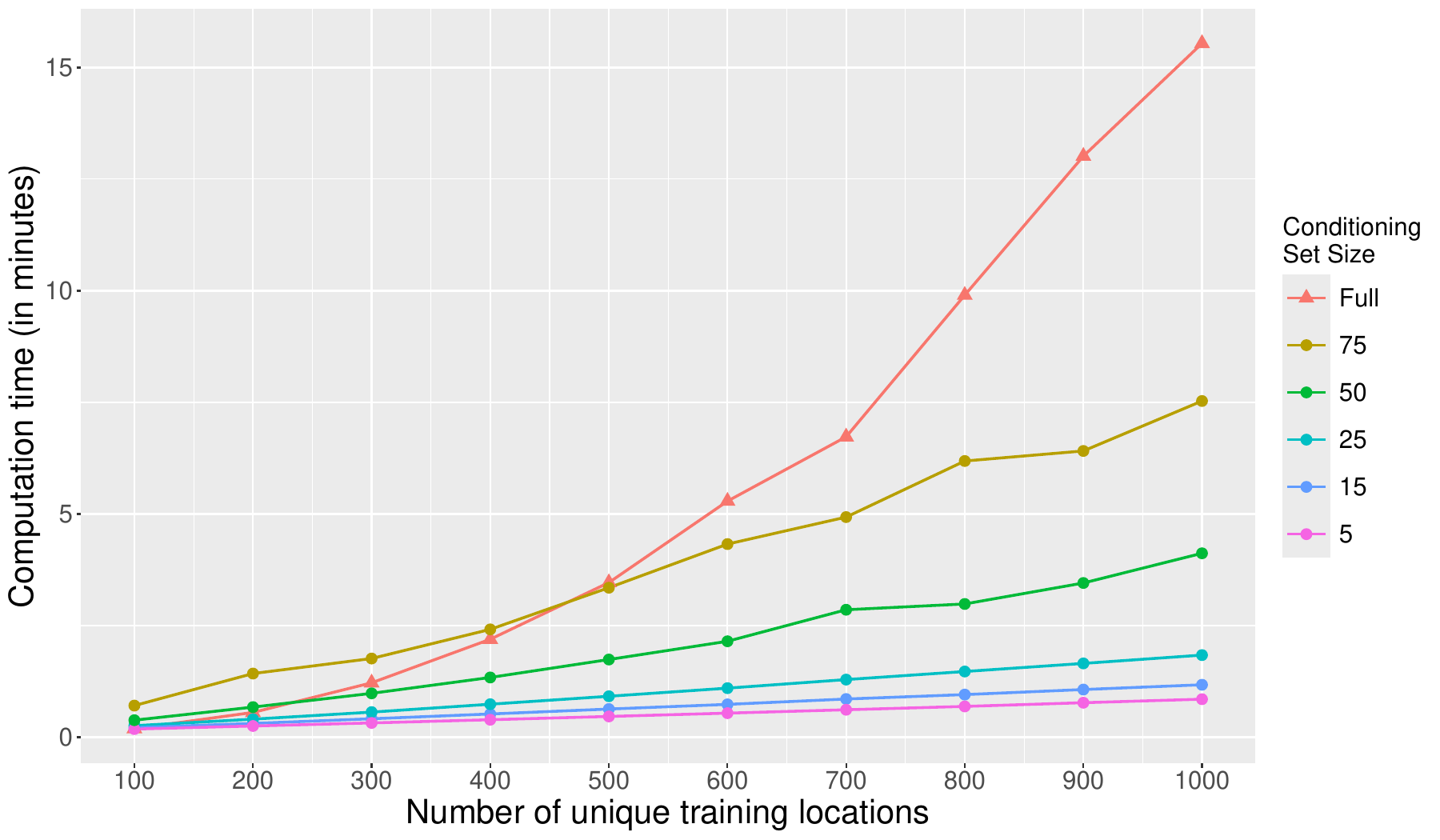}}
\caption{\emph{(a)} Scores and \emph{(b)} computation times, averaged over $20$ repeated training trials, for the full (i.e., non-Vecchia-approximated) GDGP and Vecchia-approximated GDGPs with different training conditioning set sizes in $\{5, 15, 25, 50, 75\}$ and prediction conditioning set size fixed at $200$. Models are trained on datasets with $n=100,200,\dots,1000$ unique input locations, where the outputs at each location are generated a random number of times (between $1$ and $100$) by the SIR simulator (in Section~\ref{sec:sir}). Scores are evaluated on a test set of size $N=60{,}000$, consisting of $2{,}000$ space-filling input locations with $30$ replicates at each location.}
\label{fig:figure_5}
\end{figure}

\subsection{Categorical Likelihood}
\label{sec:classification}

In this section, we assess GDGP emulation in classification settings using a categorical likelihood and a softmax inverse link. We compare several approaches, including GDGP, the Vecchia-approximated GDGP (vGDGP) with conditioning set sizes of $50$ for training and $200$ for prediction, the Bayesian GP classifier with Laplace-approximated inference (bGPC; \citealp{williams1998bayesian}) implemented in the \texttt{R} package \texttt{kernlab} ({version~0.9-33}; \citealp{karatzoglou2004kernlab}), and the sparse variational GP (SVGP; \citealp{hensman2013gaussian}) implemented in the \texttt{Python} package \texttt{gpflow} ({version~2.9.2}; \citealp{matthews2017gpflow}). To provide a broader reference point beyond GP-based methods, we also report results for a single-hidden-layer neural network (NNet) classifier, implemented in the \texttt{R} package \texttt{nnet} ({version~7.3-20}; \citealp{Venables2002nnet}), and a random forest (RF) classifier, implemented in the \texttt{R} package \texttt{randomForest} ({version~4.7-1.2}; \citealp{liaw2002classification}), both of which are widely used in practical classification tasks.

For comparison, we use a selection of widely adopted open benchmark datasets (summarized in Table~\ref{tab:tab_2}), which we view as observations generated by black-box simulators. To evaluate the performance of different models, we remove duplicated data points from the raw data extracted from the corresponding \texttt{R} packages and create $20$ random $90\%/10\%$ training/testing partitions, with class distributions approximately balanced within the splits. For each partition, we train the models on the training set and evaluate performance on the test set using the accuracy and logloss metrics, defined as follows:
\begin{align}
\label{eq:accuracy}
\mathrm{Accuracy} &= \frac{1}{|\mathcal{C}|}
\sum_{c\in\mathcal{C}}
\frac{1}{n_c}
\sum_{i:\,y_i=c}
\mathds{1}_{\left\{\hat{y}_i\,=\,c\right\}}\,,\\
\label{eq:logloss}
\mathrm{Logloss} &= -\frac{1}{|\mathcal{C}|}\sum_{c\in\mathcal{C}}
\frac{1}{n_c}
\sum_{i:\,y_i=c}
\log\tilde{\mu}^{p_{c}}(\mathbf{x}_i)\,,
\end{align}
where $\mathcal{C}$ denotes the set of all classes; $y_i$ is the true class of the $i$-th test observation; $n_c = \lvert\{i : y_i = c\}\rvert$ is the number of test observations whose true class is $c\,$; $\hat{y}_i = \argmax_{k\in\mathcal{C}} \tilde{\mu}^{p_{k}}(\mathbf{x}_i)$ is the predicted class for the $i$-th test observation; and $\tilde{\mu}^{p_{k}}(\mathbf{x}_i)$ is the predicted mean probability of class $k$ at input location $\mathbf{x}_i$.

\begin{table}
\caption{Characteristics of the datasets used in the experiments of Section~\ref{sec:classification}. The \emph{Source} column indicates the \texttt{R} packages from which the datasets are obtained.}
\label{tab:tab_2}	
\centering
\begin{tabular}{lcccc}
\toprule
\multicolumn{1}{c}{\textbf{Dataset}} & \textbf{\# Instances} & \textbf{\# Dimensions} & \textbf{\# Classes} & \textbf{Source}\\
\midrule
iris     & 149     & 4  & 3 & \texttt{datasets}\\
pima     & 768     & 8  & 2 & \texttt{mlbench}\\
thyroid  & 215     & 5  & 3 & \texttt{mclust}\\
vehicle  & 846     & 18 & 4 & \texttt{mlbench}\\
\bottomrule
\end{tabular}
\end{table}

{Table~\ref{tab:tab_3} reports the means and standard deviations, calculated across the $20$ random partitions, of accuracy, logloss, and computation time for all models. In terms of computation time, bGPC, SVGP, NNet, and RF are substantially faster than GDGP across all four datasets. This is unsurprising given the greater structural complexity of GDGP: the two competing GP methods and NNet are shallow models, while GDGP contains multiple GP nodes whose hyperparameters and latent outputs must be inferred. RF likewise avoids the latent hierarchical inference required by GDGP. {Nevertheless, GDGP achieves the best overall performance across \emph{iris}, \emph{pima}, \emph{thyroid}, and \emph{vehicle}, with lower logloss on average than the competing methods}. {In particular, although NNet attains comparable accuracy on some datasets, its logloss values are substantially higher on average than those of GDGP, indicating that GDGP provides more reliable probabilistic predictions.} The experiment therefore demonstrates that the additional computational cost (which is still manageable) associated with the flexibility in the GDGP buys tangible improvements in predictive performance and uncertainty quantification. 

Within the GDGP framework, as expected, vGDGP exhibits some degradation in predictive performance due to the Vecchia approximation, but remains competitive with the other models while reducing the mean computation time of GDGP from $280.48$ to $127.27$ seconds for \emph{pima} and from $1182.20$ to $798.62$ seconds for \emph{vehicle}. For the smaller \emph{iris} and \emph{thyroid} datasets, vGDGP is slower than full GDGP because the computational overhead (e.g., NN search, setup of the sparsity structures of the precision matrices, and parallelization) incurred by the Vecchia implementation is not offset by its savings in inference at these smaller training sizes.}

\begin{table}
\caption{Means and standard deviations (in parentheses) of accuracy (higher is better; see Equation~\eqref{eq:accuracy}), logloss (lower is better; see Equation~\eqref{eq:logloss}), and {computation time (in seconds)}, computed across the $20$ random partitions of the datasets listed in Table~\ref{tab:tab_2}, for GDGP, Vecchia-approximated GDGP (vGDGP), Bayesian GP classifier (bGPC), sparse variational GP (SVGP), neural network (NNet), and random forest (RF). The highest mean accuracy and lowest mean logloss for each dataset are highlighted in bold.}
\label{tab:tab_3}	
\centering
\begin{booktabs}{
  colspec = {Q[l,m] Q[l,m] Q[c,m] Q[c,m] Q[c,m] Q[c,m]},
  cell{1}{1} = {r=2}{c,m,font=\bfseries},
  cell{1}{2} = {r=2}{c,m,font=\bfseries},
  cell{1}{3} = {c=4}{c,m,font=\bfseries},
  cell{2}{3-6} = {c,m},
  cell{3}{1} = {r=3}{l,m},
  cell{6}{1} = {r=3}{l,m},
  cell{9}{1} = {r=3}{l,m},
  cell{12}{1} = {r=3}{l,m},
  cell{15}{1} = {r=3}{l,m},
  cell{18}{1} = {r=3}{l,m},
}
\toprule
Model & Metric & Dataset &  &  &  \\
\cmidrule[lr]{3-6}
 &  & iris & pima & thyroid & vehicle \\
\midrule
\midrule
GDGP & Accuracy & \textbf{96.00 (4.63)} & \textbf{73.27 (4.55)} & \textbf{98.22 (3.72)} & \textbf{80.14 (3.19)} \\
 & Logloss & \textbf{0.1153 (.0712)} & \textbf{0.5262 (.0560)} & \textbf{0.0717 (.0469)} & \textbf{0.4353 (.0554)} \\
 & Time & 26.25 (1.63) & 280.48 (9.38) & 53.86 (2.09) & 1182.20 (44.44) \\
\midrule
vGDGP & Accuracy & 95.00 (5.82) & 71.38 (4.53) & 95.00 (5.67) & 78.07 (3.65) \\
 & Logloss & 0.1489 (.1028) & 0.5438 (.0489) & 0.1268 (.0638) & 0.4830 (.0596) \\
 & Time & 67.53 (9.15) & 127.27 (2.75) & 92.97 (3.32) & 798.62 (19.51) \\
\midrule
bGPC & Accuracy & 94.50 (5.25) & 72.50 (4.30) & 94.89 (6.65) & 72.49 (4.39) \\
 & Logloss & 0.3844 (.0624) & 0.5299 (.0455) & 0.6176 (.0429) & 0.8389 (.0317) \\
 & Time & 0.02 (.02) & 1.80 (.36) & 0.12 (.02) & 1.93 (.10) \\
\midrule
SVGP & Accuracy & 94.42 (5.28) & 72.39 (4.58) & 93.67 (7.64) & 75.51 (5.72) \\
 & Logloss & 0.1529 (.0539) & 0.5429 (.0381) & 0.2334 (.1239) & 0.6232 (.0839) \\
 & Time & 0.56 (.06) & 0.63 (.08) & 0.67 (.38) & 0.87 (.11) \\
\midrule
\midrule
NNet & Accuracy & 95.42 (4.77) & 67.55 (6.17) & 95.89 (5.11) & 79.31 (3.32) \\
 & Logloss & 1.0027 (1.3361) & 2.0777 (1.2169) & 1.0545 (1.5007) & 5.2156 (1.1464) \\
 & Time & 0.01 (.00) & 0.05 (.01) & 0.00 (.00) & 0.31 (.01) \\
\midrule
RF & Accuracy & 94.17 (5.09) & 73.15 (4.37) & 96.33 (5.06) & 74.74 (4.32) \\
 & Logloss & 0.1516 (.1059) & 0.5394 (.0601) & 0.1184 (.0658) & 0.4971 (.0404) \\
 & Time & 0.01 (.01) & 0.22 (.04) & 0.03 (.00) & 0.37 (.04) \\
\midrule
\bottomrule
\end{booktabs}
\end{table}

\subsection{Count Likelihoods}
\label{sec:count}
In this section, we examine the performance of GDGP with count likelihoods for simulators whose outputs are non-negative integer-valued counts, and consider the Rosenzweig-MacArthur predator-prey model, which incorporates density-dependent prey growth and a nonlinear Type-II functional response for the predator~\citep{rosenzweig1963graphical,pineda2008gillespiessa}. In particular, the simulator is evaluated over four inputs: the predator death rate ($d_F\in[0.1,2.0]$), the prey death rate ($d_R\in[0.1,1.8]$), the predation efficiency ($\alpha\in[0.01,0.02]$), and the degree of predator saturation ($w\in[0,0.04]$), and returns the prey population count ($R$) at time $100$, given fixed initial populations of $R_0=50$ prey and $F_0=5$ predators. We generate training data by evaluating the simulator at $300$ unique input design points, with a random number of replicates (i.e., runs) between $1$ and
{$100$} assigned to each design point. The resulting training set is used to fit GDGP, SVGP, and a generalized linear model (GLM) with a Poisson likelihood. We evaluate the three models using the negative log-likelihood (NLL) on a test set consisting of simulator evaluations at $1{,}000$ unique Latin hypercube test sites, with $50$ replicates per site. Figure~\subref*{fig:figure_6a} compares the performance of the three models across $50$ training trials and shows that GDGP substantially outperforms both the GLM and SVGP in emulating the simulator, achieving noticeably lower NLL overall and smaller variation across trials {than SVGP}. 

\begin{figure}
\centering 
\subfloat[]{\label{fig:figure_6a}\includegraphics[width=0.48\linewidth]{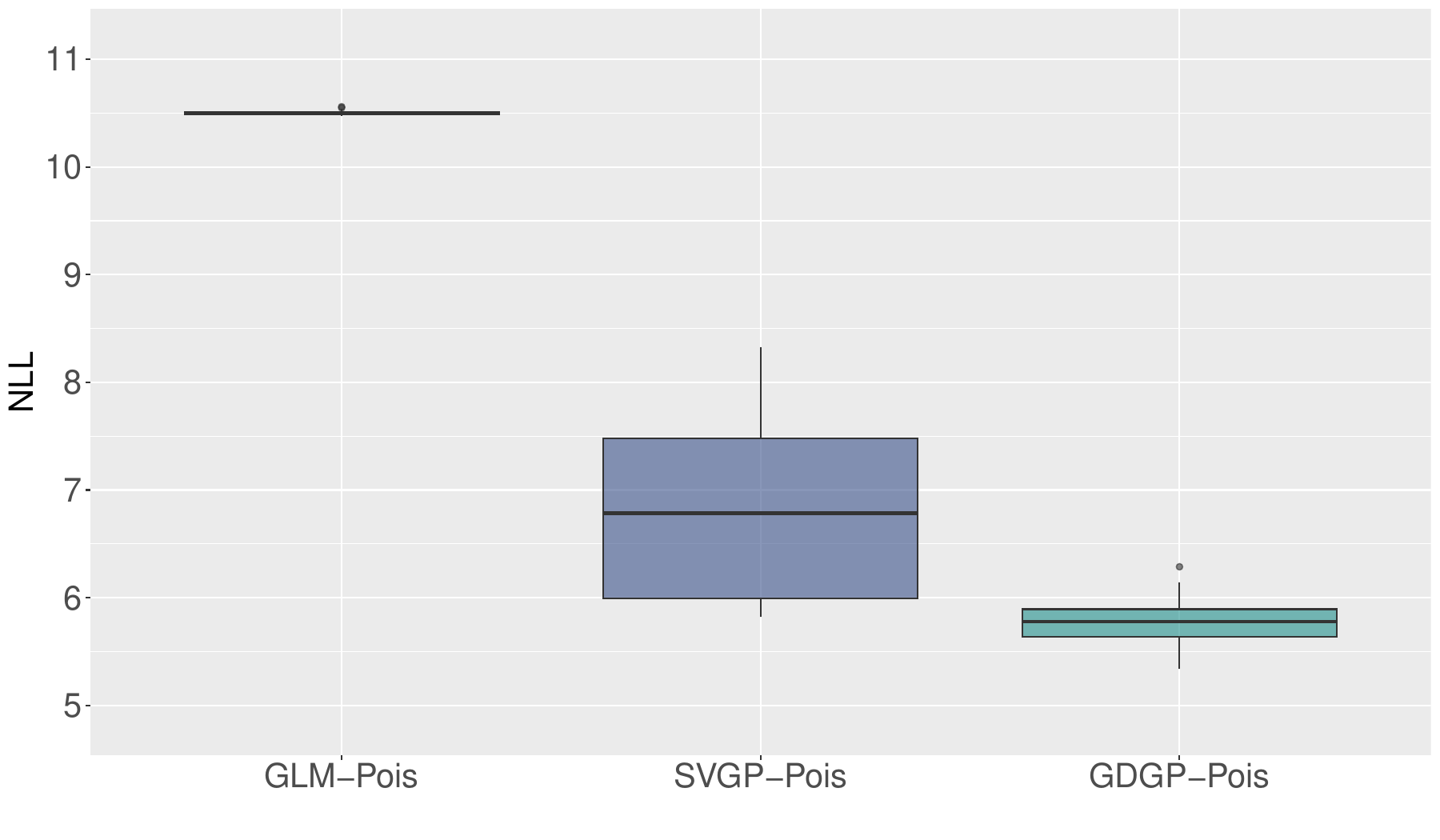}}
\hspace{0.5em}
\subfloat[]{\label{fig:figure_6b}\includegraphics[width=0.48\linewidth]{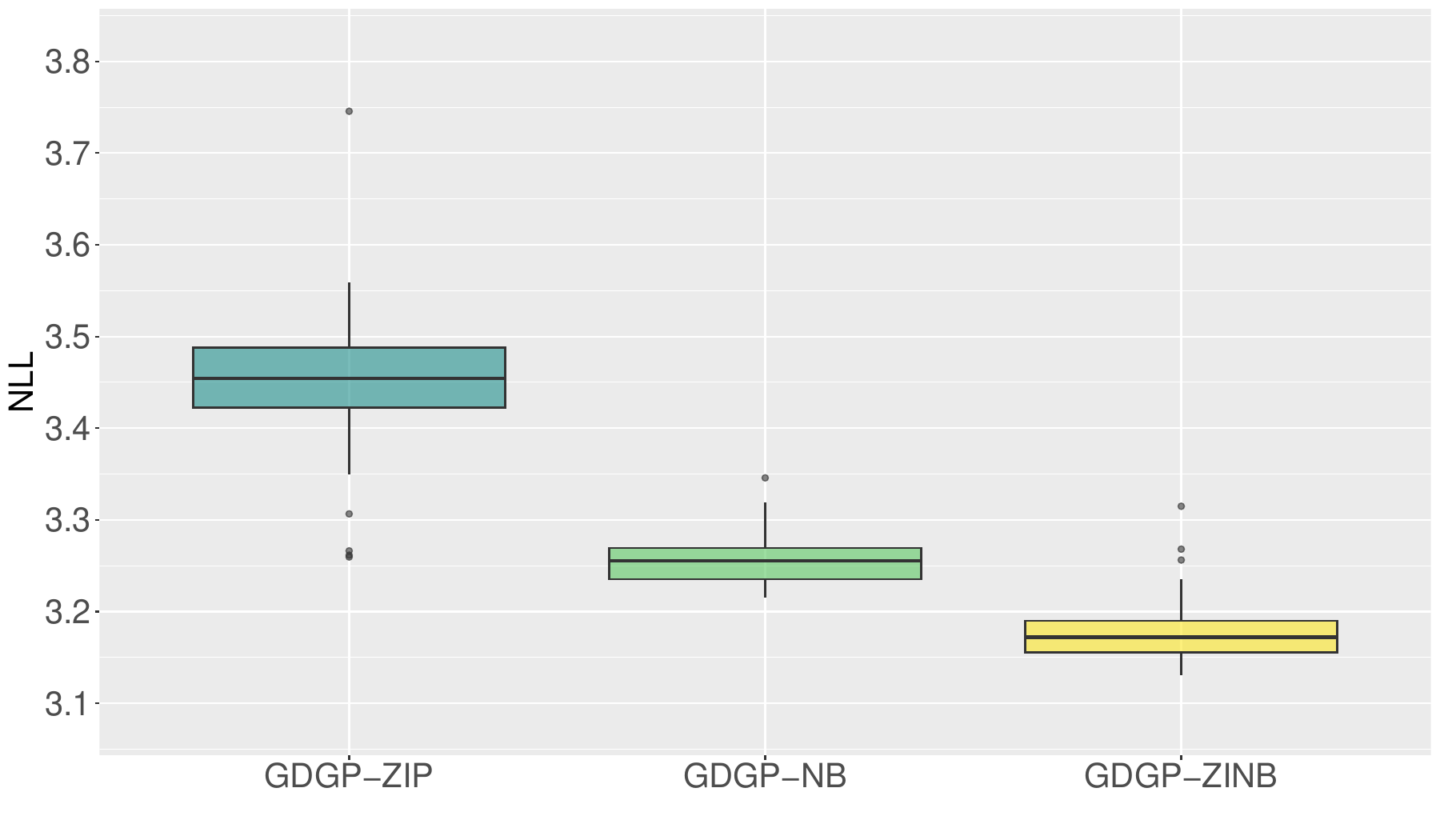}}
\caption{Negative log-likelihood (NLL; lower is better) for \emph{(a)} the generalized linear model (GLM), SVGP, and GDGP with a Poisson (Pois) likelihood, and \emph{(b)} GDGPs with {zero-inflated Poisson (ZIP), negative binomial (NB)}, and zero-inflated negative binomial (ZINB) likelihoods, trained on $300$ unique input locations with up to {$100$} replicates per location, for emulating the prey population count of the predator-prey simulator described in Section~\ref{sec:count}. NLLs are evaluated on a test set of size $N=50{,}000$, consisting of $1{,}000$ space-filling input locations with $50$ replicates per location, and summarized over $50$ independent training trials.}
\label{fig:figure_6}
\end{figure}

To demonstrate the flexibility of GDGP and assess sensitivity to the choice of count likelihood, we extend the experiment to include GDGP with {zero-inflated Poisson (ZIP), negative binomial (NB) and zero-inflated negative binomial (ZINB) likelihoods}. This is motivated by the fact that the discrete birth-death and predation reactions in the simulator can induce substantial replicate-to-replicate variability in the terminal counts (e.g., some trajectories may reach absorbing extinction states by the end of the simulation, yielding zero predator or prey counts), resulting in over-dispersion and occasional excess zeros in the simulated outputs. As shown in Figure~\subref*{fig:figure_6b}, GDGPs with {ZIP, NB, and ZINB} likelihoods consistently achieve {lower} NLL than GDGP with a Poisson likelihood. {In particular, the lower NLL achieved by GDGP with ZIP than with Poisson highlights the importance of accounting for excess zeros, while the lower average NLL achieved by GDGP with NB than with ZIP indicates that accounting for additional dispersion among positive counts offers greater benefits in this experiment than explicitly modeling excess zeros alone.} GDGP with ZINB achieves the lowest overall NLL, implying that it provides the best probabilistic fit by jointly modeling zero inflation and residual over-dispersion beyond the extra zeros present in the simulated counts.

{The exact mean NLLs and empirical computation times for all models considered in this experiment are summarized in Table~\ref{tab:tab_4}. As with the other experiments, SVGP, which has a shallow GP structure, and the even simpler GLM are by nature faster than GDGP, but the increased computational cost of GDGP yields improved accuracy in return. The results further show that using non-Poisson likelihoods delivers a better fit for GDGP, but at an additional computational cost, partly reflecting the greater numbers of GP nodes and latent outputs ($6$ each for NB and ZIP, and $7$ each for ZINB) required by these more flexible likelihood structures.}

\begin{table}
\caption{{Mean negative log-likelihoods (NLLs; lower is better) and computation times (in seconds), calculated over $50$ independent training trials, for the Poisson GLM (GLM-Pois), Poisson SVGP (SVGP-Pois), and GDGPs with Poisson (GDGP-Pois), {zero-inflated Poisson (GDGP-ZIP), negative binomial (GDGP-NB)}, and zero-inflated negative binomial (GDGP-ZINB) likelihoods in the predator-prey experiment described in Section~\ref{sec:count}.}}
\label{tab:tab_4}
\centering
\begin{booktabs}{
  colspec = {Q[l,m] Q[r,m] Q[r,m]},
  row{1} = {c,m,font=\bfseries}
}
\toprule
Model & NLL & Time (s) \\
\midrule
GLM-Pois & 10.50 & 0.03 \\
SVGP-Pois & 6.79 & 43.20 \\
GDGP-Pois & 5.78 & 85.46 \\
GDGP-ZIP & 3.45 & 150.58 \\
GDGP-NB & 3.25 & 171.11 \\
GDGP-ZINB & 3.18 & 235.86 \\
\bottomrule
\end{booktabs}
\end{table}

Beyond the aggregate comparisons above, Figure~\ref{fig:figure_7} directly explores how different models under a Poisson likelihood and different likelihood choices within GDGP affect the predictive distribution. For this visual comparison, all six models considered are trained using a common training dataset comprising $300$ unique locations with up to $100$ replicates per location. The one-dimensional slices of the resulting predictive PMFs, shown by the shading in Figure~\ref{fig:figure_7}, are then generated over $d_R\in[0.1,1.8]$ while fixing the remaining inputs at $d_F=1.050$, $\alpha=0.015$, and $w=0.004$. {The PMFs for GDGPs were estimated at each prediction input along $d_R$ from the relative frequencies of $5{,}000$ predictive count samples, comprising $500$ samples for each of the $K=10$ imputations, generated using Steps~\ref{step2}--\ref{step3} of Algorithm~\ref{alg:composition}. For GLM with a Poisson likelihood, its predictive PMF at each input location was evaluated directly using the fitted Poisson rate, while for SVGP its per-location predictive PMF was evaluated by numerically integrating the conditional Poisson PMF over the Gaussian variational predictive distribution of the latent GP output, as implemented in \texttt{gpflow}.} As an empirical reference for the ground-truth distribution, $150$ simulator runs are generated at each of $31$ evenly spaced values of $d_R$ and overlaid on the PMFs. 

Figure~\subref*{fig:figure_7a} shows that the GLM with a Poisson likelihood produces a predictive distribution that does not align with the simulator runs, with its probability mass centered below the main empirical concentration of the outputs and its dispersion remaining largely unchanged as $d_R$ varies. Both SVGP and GDGP with a Poisson likelihood, shown in Figures~\subref*{fig:figure_7b} and~\subref*{fig:figure_7c}, respectively, better capture the central pattern of the simulator runs than the GLM, but the posterior predictive distribution of SVGP is relatively diffuse, whereas that of GDGP is more tightly centered on the main concentration of simulated prey counts. Nevertheless, all three models struggle to represent the coexistence of zero and widely dispersed positive prey counts at larger values of $d_R$. {The GDGP with a ZIP likelihood in Figure~\subref*{fig:figure_7d} better represents the mixture of zero and positive prey counts by assigning distinct probability mass to zero while retaining support over positive counts, but it exhibits insufficient dispersion among the nonzero responses as $d_R$ increases. By comparison, the GDGP with an NB likelihood in Figure~\subref*{fig:figure_7e} accommodates the greater dispersion at larger values of $d_R$, consistent with its lower overall NLL than that of GDGP with ZIP shown in Figure~\subref*{fig:figure_6b} and Table~\ref{tab:tab_4}. However, because it represents zero and positive counts through a single count distribution, the increasingly frequent zero prey counts as $d_R$ increases shift more of its predictive mass towards zero. Consequently, it assigns excessive probability mass to small positive counts relative to the empirical distribution. The GDGP with a ZINB likelihood, shown in Figure~\subref*{fig:figure_7f}, addresses both features by assigning distinct probability mass to zero while accommodating the increasing dispersion among nonzero outcomes more adequately.}

\begin{figure}
\centering 
\subfloat[GLM-Pois]{\label{fig:figure_7a}\includegraphics[width=0.46\linewidth]{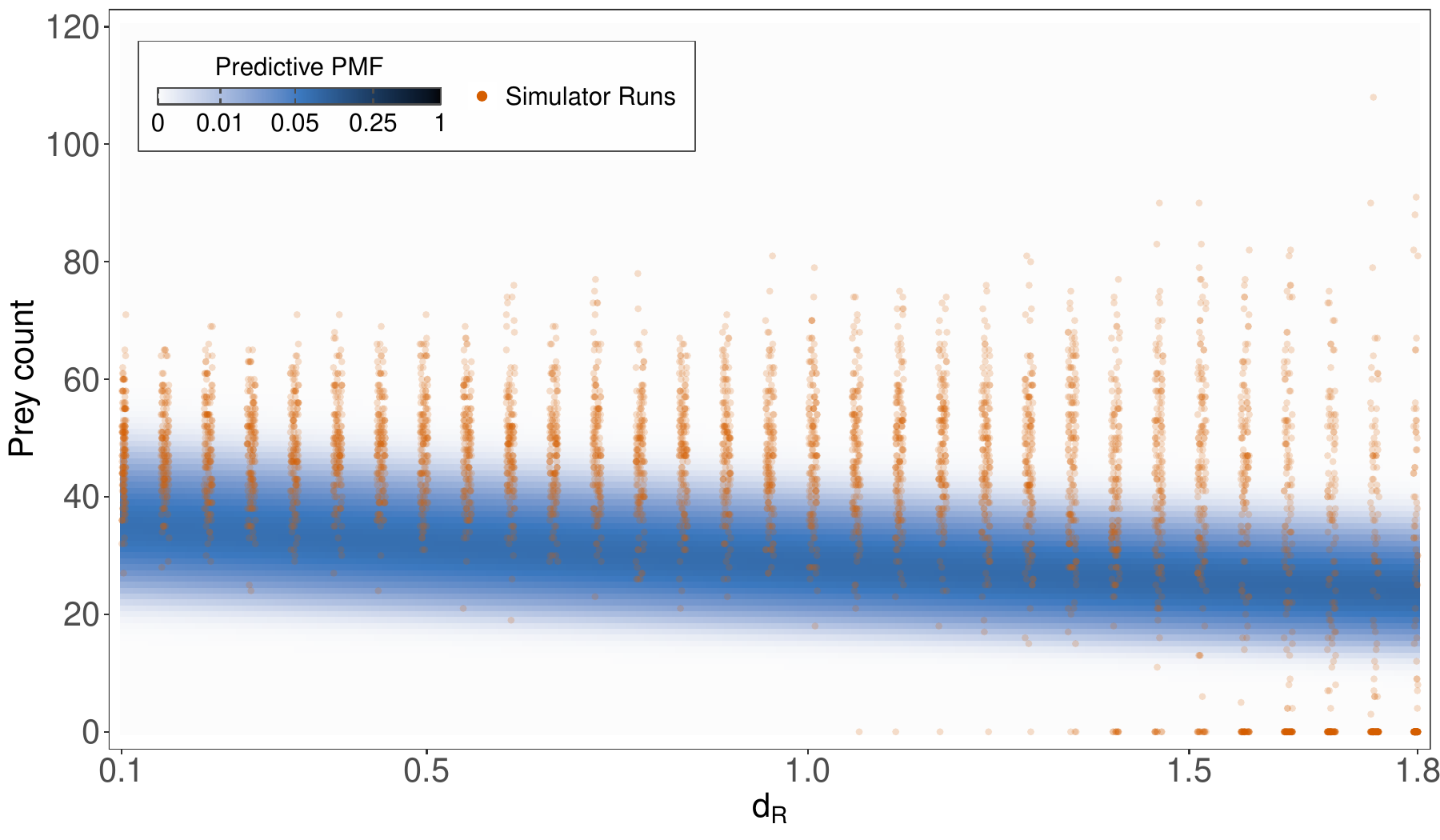}}
\hspace{0.7em}
\subfloat[SVGP-Pois]{\label{fig:figure_7b}\includegraphics[width=0.46\linewidth]{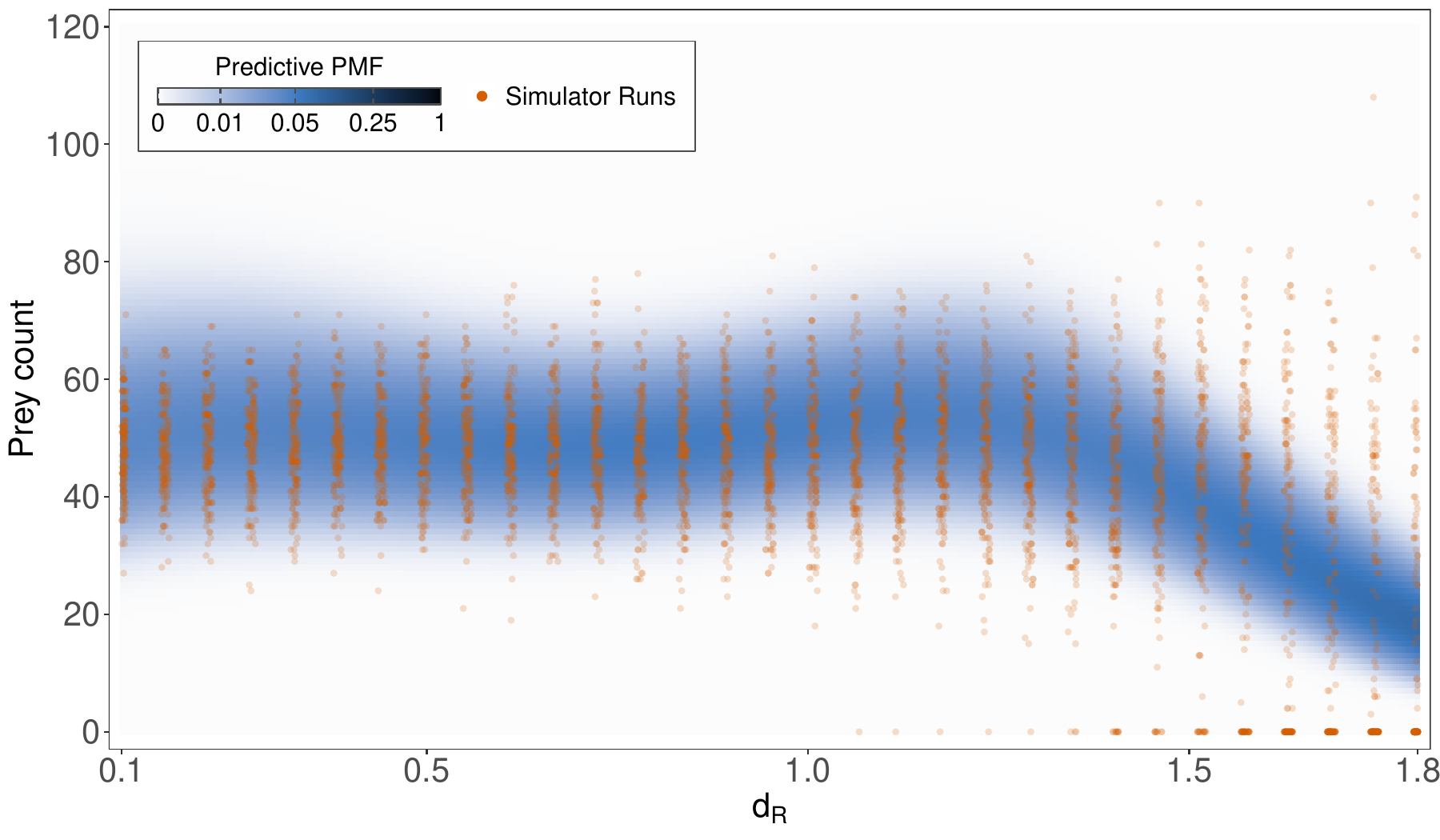}}\\
\subfloat[GDGP-Pois]{\label{fig:figure_7c}\includegraphics[width=0.46\linewidth]{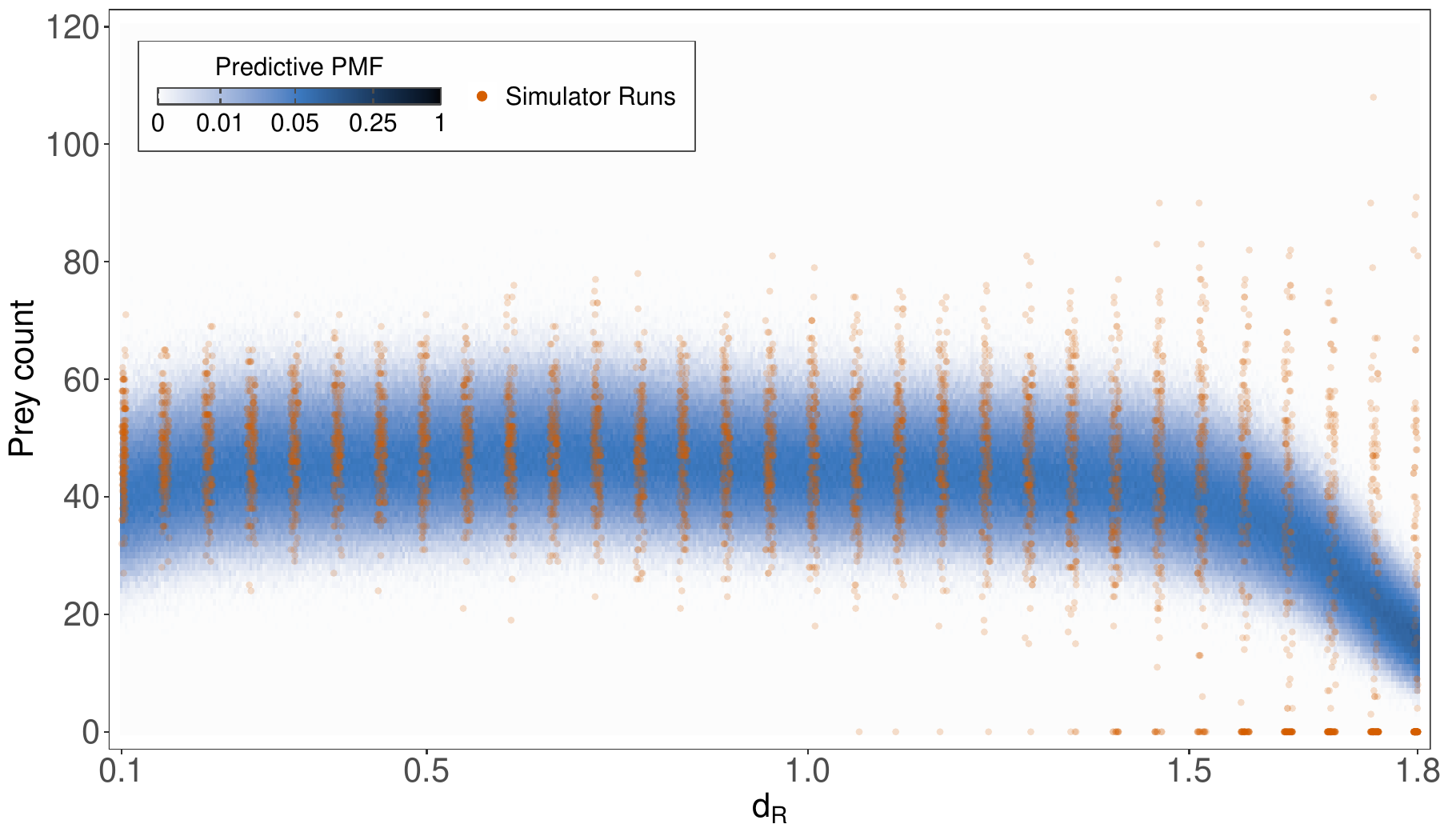}}
\hspace{0.7em}
\subfloat[GDGP-ZIP]{\label{fig:figure_7d}\includegraphics[width=0.46\linewidth]{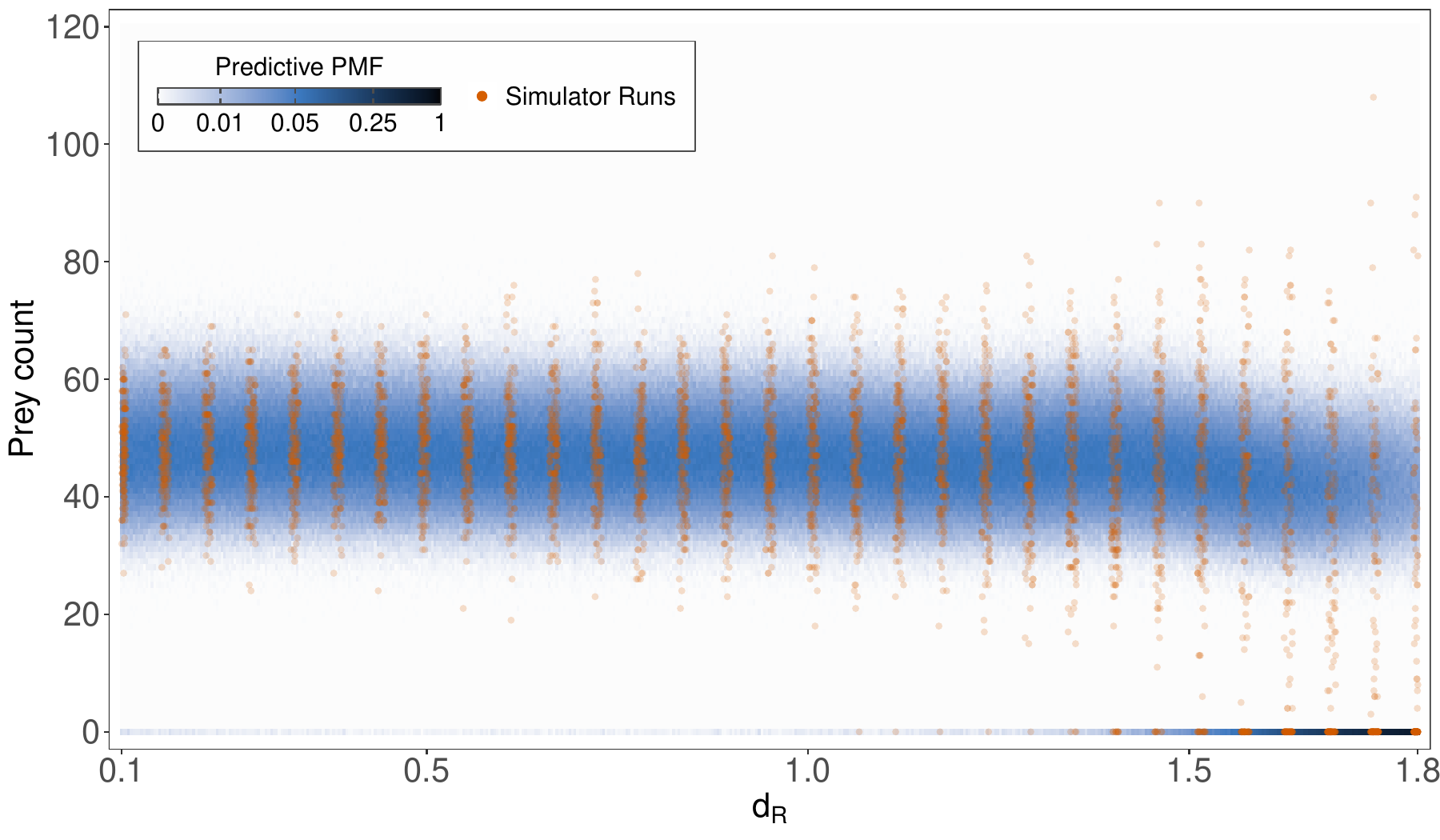}}\\
\subfloat[GDGP-NB]{\label{fig:figure_7e}\includegraphics[width=0.46\linewidth]{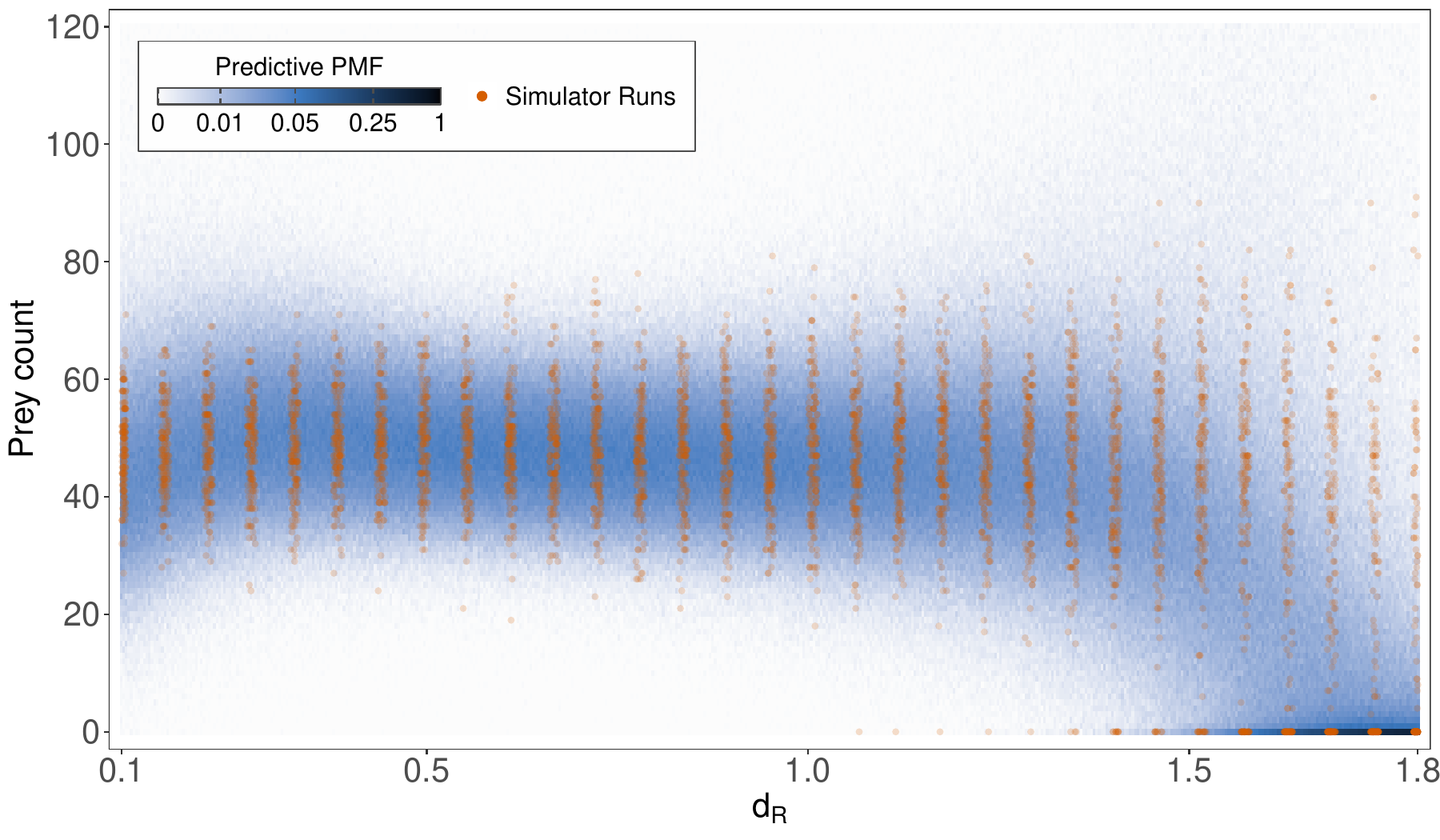}}
\hspace{0.7em}
\subfloat[GDGP-ZINB]{\label{fig:figure_7f}\includegraphics[width=0.46\linewidth]{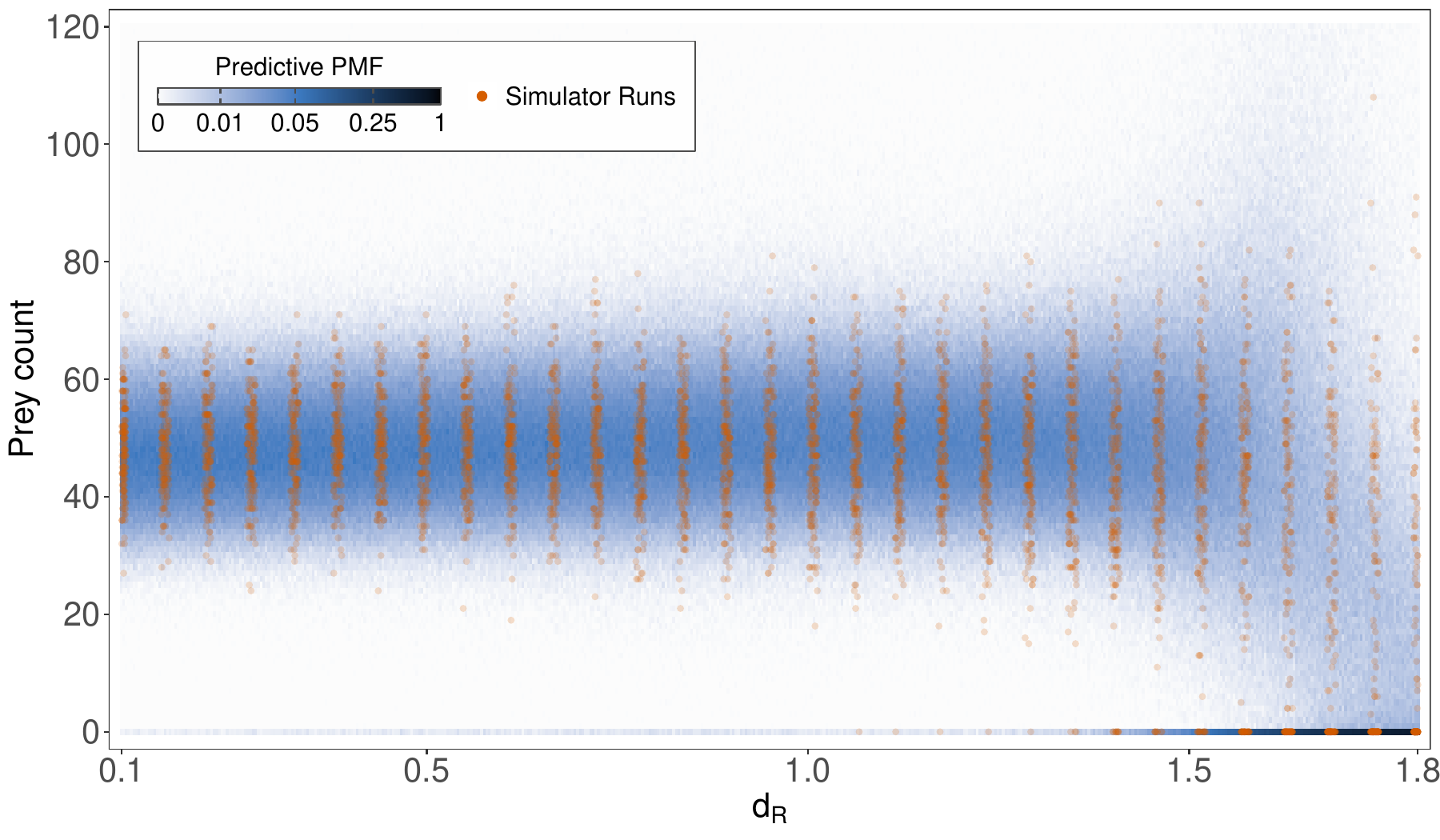}}
\caption{{One-dimensional slices of the predictive probability mass functions (PMFs) for the prey population count obtained from \emph{(a)} a GLM, \emph{(b)} an SVGP, and \emph{(c)} a GDGP  with a Poisson (Pois) likelihood, and from GDGPs with {\emph{(d)}} zero-inflated Poisson (ZIP), {\emph{(e)}} negative binomial (NB), and \emph{(f)} zero-inflated negative binomial (ZINB) likelihoods. All six models are fitted using a common training dataset comprising $300$ unique input locations with up to $100$ replicates per location. The slices vary $d_R$ over $[0.1,1.8]$ while fixing $d_F=1.050$, $\alpha=0.015$, and $w=0.004$. The shading represents the predictive probability mass, with darker shades indicating greater mass. The points show the empirical distributions of prey counts obtained from $150$ simulator runs at each of $31$ evenly spaced values of $d_R$ and are horizontally jittered for visibility.}}
\label{fig:figure_7}
\end{figure}

\section{Conclusion}
\label{sec:conclusion}
In this work, we proposed a Generalized Deep Gaussian Process (GDGP) framework for emulating non-stationary computer models with non-Gaussian outputs. By introducing a likelihood layer at the top of the DGP hierarchy, the proposed framework extends DGP emulation beyond deterministic and homogeneous Gaussian settings to accommodate a broad class of response distributions. This construction preserves the flexibility of DGPs for modeling non-stationary simulator behavior, while providing a unified approach to handling diverse simulator output types, including heteroskedastic, categorical, count, and zero-inflated outputs.

Inference for GDGP was developed within the Stochastic Imputation (SI) framework. To improve practicality in large-scale settings, we incorporated the Vecchia approximation to enable scalable inference for large numbers of input locations, and showed that SI can accommodate large numbers of replicates (within stochastic simulators) without increasing the computational complexity of inference. We also demonstrated that, in the heteroskedastic Gaussian case, additional analytical structure can be exploited to obtain closed-form conditional updates in the imputation step, yielding further computational gains. {Through a series of synthetic and empirical examples covering heteroskedastic Gaussian, categorical, and count responses, we showed that the developments introduced in this work make GDGP a computationally viable framework for emulation problems requiring substantially richer structure than is available from standard shallow GP and GLM formulations. Across the examples considered, GDGP achieved the strongest overall predictive performance, with particularly marked gains in probabilistic prediction and hence uncertainty quantification. These gains come at an additional computational cost, as expected from the deep hierarchy of latent GPs within GDGP, but the scalable inference developments introduced in this work keep GDGP computationally practical for the problems considered. We note that when a simpler model provides an adequate representation of the response, it should be preferred and may perform better than GDGP, since the additional structure in GDGP must itself be learned from the available data. However, when simpler models prove inadequate, GDGP can model more complex input-response relationships and expands the range of response structures that can be represented within a probabilistic emulator.

As noted in Section~\ref{sec:experiments}, all experiments in this work use GDGP emulators with a two-layer DGP structure. To explore sensitivity to this choice, we also report the corresponding results for a three-layer (i.e., $L=3$) DGP structure in the Supplementary Material. The results indicate that increasing the DGP depth from two to three layers does not provide a consistent improvement in emulation performance. The three-layer structure yields some gains when emulating the mean function in the heteroskedastic step function experiment but produces comparable or even less favorable results in all other cases. GDGPs with three-layer DGP structures also require noticeably longer computation times, particularly for higher dimensional inputs, due to the larger numbers of GP nodes and latent outputs introduced by the additional layer. For example, the average computation times of GDGPs using two- and three-layer DGP structures (with $22$ versus $40$ GP nodes and latent outputs each) on the \emph{vehicle} dataset in Section~\ref{sec:classification} are $1182.20$ and $5166.79$ seconds, respectively. Therefore, a two-layer DGP structure appears to be the more justifiable choice for GDGP emulation in practice, considering both emulation performance and the required computational cost.}

As part of this work, we also provide a careful implementation of the GDGP framework in the open-source \texttt{R} package \texttt{dgpsi}, which, to the best of our knowledge, is the first out-of-the-box software package for DGP-based scalable emulation of computer simulators with non-Gaussian responses.

Several directions for future work remain. On the methodological side, it would be of interest to extend the framework to more complex output structures, such as multivariate or mixed-type responses. From a computational perspective, further gains may be possible by identifying additional likelihoods that admit analytical conditional updates, in the same spirit as the heteroskedastic Gaussian case. Another promising direction is to develop hybrid inference schemes at the likelihood layer, in which only parameters exhibiting clear input dependence are modeled through DGP outputs, while others are treated as global parameters and estimated by maximum likelihood. Such a strategy could reduce the number of latent GPs and further improve computational efficiency for multi-parameter likelihoods. 

{It is also worth noting that GDGP with a Poisson likelihood and a log link is related to a log-Gaussian Cox process (LGCP;~\citealp{moller1998log,diggle2013spatial}), since they both construct a positive Poisson rate by exponentiating a latent function, modeled using a DGP and a GP, respectively. The settings in which the two approaches are typically used are, however, rather different. In this work, the GDGP is developed for computer model emulation, with count responses observed at a prescribed set of inputs in a potentially moderate- or high-dimensional simulator parameter space. LGCPs instead model spatial point patterns, with the latent intensity defined over a typically low-dimensional spatial domain. This distinction is important because the computational regimes are correspondingly different: spatial applications often involve very large numbers of observations in two or three dimensions, whereas computer experiments commonly involve substantially higher-dimensional inputs but comparatively sparse designs. Our current GDGP implementation therefore does not directly model raw spatial point-pattern data and is not developed or assessed as an LGCP methodology. Nevertheless, it could in principle be applied to provide a discretized approximation to the spatial intensity surface, in a similar fashion to~\citet{zilber2021vecchia}, who used a GP-based approach to approximate an LGCP by partitioning the spatial domain into grid cells and employed a Vecchia-Laplace approximation for scalable inference. Establishing the numerical and statistical performance of GDGP in such a setting would require separate investigation and represents a potential future extension of the present work.}

{Finally, when simulator evaluations are limited, allocating them efficiently between unique input locations and replicates becomes important for making the best use of the available simulation budget to build a high-quality stochastic emulator. In practice, one could use a space-filling design to improve coverage of the input space, especially in higher-dimensional problems, while allocating additional replicates to input regions expected to be particularly informative for improving the estimation of characteristics of the output distribution, such as variance or zero inflation. A more automated and adaptive alternative is sequential design, which updates the emulator as simulator evaluations are collected and determines whether new evaluations should be allocated to previously untried input locations or used to replicate existing ones. This exploration-replication trade-off has been studied for heteroskedastic GP emulation by~\citet{binois2019replication}. The \texttt{dgpsi} package already implements sequential design for GP and DGP emulators; extending these methods to the GDGP framework, including batch-sequential settings, is an important direction for future work.}

\section*{Data Availability}
The code and data used for the experiments reported in Section~\ref{sec:experiments} and in the Supplementary Material can be accessed via \url{https://doi.org/10.17632/7psdbn24yy.1}.

\section*{Acknowledgments}
The authors acknowledge support from the ADD-TREES project, funded by the Engineering and Physical Sciences Research Council (EPSRC) under grant EP/Y005597/1.

\bibliographystyle{agsm}  
\bibliography{references}  

\begin{appendices}
\setcounter{equation}{0}
\setcounter{section}{0}
\setcounter{figure}{0}
\setcounter{table}{0}
\makeatletter
\renewcommand{\theequation}{A\arabic{equation}} 
\renewcommand{\thefigure}{A.\arabic{figure}} 
\renewcommand{\thetable}{A.\arabic{table}} 

\section{Closed-form Predictive Means and Variances}
\label{app:mean_variance}

Let $\mu^{(q)}_{0,k}$ and ${{\sigma}^{2}}^{(q)}_{0,k}$ denote the mean and variance of the LGP approximation, $\widehat{p}(f^{(q)}_0|\mathbf{x}_0;\mathbf{f}_k,\{\mathbf{w}^{(p)}_{l}\}_k,\mathbf{x})$, at input location $\mathbf{x}_0$, for $q=1,\dots,Q$. Then, under a range of parametric choices for $p(y|\boldsymbol{\phi})$, the predictive mean $\tilde{\mu}^{Y}_{0,k}$ and variance $(\tilde{\sigma}^{Y}_{0,k})^2$ admit closed-form expressions. Table~\ref{tab:tab_a1} lists several representative distributions for which such expressions are available. The corresponding closed-form results are given below:

\subsection{Poisson}
\begin{align*}
    \tilde{\mu}^{Y}_{0,k} =& \exp\left\{\mu^{(1)}_{0,k}+\frac{{{\sigma}^{2}}^{(1)}_{0,k}}{2}\right\}\,,\\
    (\tilde{\sigma}^{Y}_{0,k})^2 =& \exp\left\{\mu^{(1)}_{0,k}+\frac{{{\sigma}^{2}}^{(1)}_{0,k}}{2}\right\}+\left(e^{{{\sigma}^{2}}^{(1)}_{0,k}}-1\right)\exp\left\{2\mu^{(1)}_{0,k}+{{\sigma}^{2}}^{(1)}_{0,k}\right\}\,.
\end{align*}

\subsection{Exponential}
\begin{align*}
    \tilde{\mu}^{Y}_{0,k} =& \exp\left\{\mu^{(1)}_{0,k}+\frac{{{\sigma}^{2}}^{(1)}_{0,k}}{2}\right\}\,,\\
    (\tilde{\sigma}^{Y}_{0,k})^2 =& \left(2e^{{{\sigma}^{2}}^{(1)}_{0,k}}-1\right)\exp\left\{2\mu^{(1)}_{0,k}+{{\sigma}^{2}}^{(1)}_{0,k}\right\}\,.
\end{align*}

\subsection{Gamma}
\begin{align*}
    \tilde{\mu}^{Y}_{0,k} =& \exp\left\{\mu^{(1)}_{0,k}+\frac{{{\sigma}^{2}}^{(1)}_{0,k}}{2}\right\}\,,\\
    (\tilde{\sigma}^{Y}_{0,k})^2 =& \left(e^{{{\sigma}^{2}}^{(1)}_{0,k}}+e^{{{\sigma}^{2}}^{(1)}_{0,k}+2\mu^{(2)}_{0,k}+2{{\sigma}^{2}}^{(2)}_{0,k}}-1\right)\exp\left\{2\mu^{(1)}_{0,k}+{{\sigma}^{2}}^{(1)}_{0,k}\right\}\,.
\end{align*}

\subsection{Heteroskedastic Gaussian}
\begin{align*}
    \tilde{\mu}^{Y}_{0,k} =& \mu^{(1)}_{0,k}\,,\\
    (\tilde{\sigma}^{Y}_{0,k})^2 =& \exp\left\{\mu^{(2)}_{0,k}+\frac{{{\sigma}^{2}}^{(2)}_{0,k}}{2}\right\}+{{\sigma}^{2}}^{(1)}_{0,k}\,.
\end{align*}

\subsection{Negative Binomial}
\begin{align*}
    \tilde{\mu}^{Y}_{0,k} =& \exp\left\{\mu^{(1)}_{0,k}+\frac{{{\sigma}^{2}}^{(1)}_{0,k}}{2}\right\}\,,\\
    (\tilde{\sigma}^{Y}_{0,k})^2 =& \left(e^{{{\sigma}^{2}}^{(1)}_{0,k}}-1\right)\exp\left\{2\mu^{(1)}_{0,k}+{{\sigma}^{2}}^{(1)}_{0,k}\right\} + \exp\left\{\mu^{(1)}_{0,k}+\frac{{{\sigma}^{2}}^{(1)}_{0,k}}{2}\right\} + \exp\left\{\mu^{(2)}_{0,k} +\frac{{{\sigma}^{2}}^{(2)}_{0,k}}{2} +2\mu^{(1)}_{0,k} +2{{\sigma}^{2}}^{(1)}_{0,k}\right\}\,.
\end{align*}

\subsection{Zero-Inflated Poisson}
\begin{align*}
    \tilde{\mu}^{Y}_{0,k} =& \left(1-\bar{\pi}_{0,k}\right)
    \exp\left\{\mu^{(1)}_{0,k}+\frac{{{\sigma}^{2}}^{(1)}_{0,k}}{2}\right\}\,,\\
    (\tilde{\sigma}^{Y}_{0,k})^2 =& \left(1-\bar{\pi}_{0,k}\right)
    \left[
    \exp\left\{\mu^{(1)}_{0,k}+\frac{{{\sigma}^{2}}^{(1)}_{0,k}}{2}\right\}
    +
    \left(e^{{{\sigma}^{2}}^{(1)}_{0,k}}-1\right)
    \exp\left\{2\mu^{(1)}_{0,k}+{{\sigma}^{2}}^{(1)}_{0,k}\right\}
    \right]\\
    &\,
    +\bar{\pi}_{0,k}
    \left(1-\bar{\pi}_{0,k}\right)
    \exp\left\{2\mu^{(1)}_{0,k}+{{\sigma}^{2}}^{(1)}_{0,k}\right\}\,,
\end{align*}
where $\bar{\pi}_{0,k} =
\operatorname{logit}^{-1}\left(
\frac{\mu^{(2)}_{0,k}}
{\sqrt{1+\frac{\pi}{8}{{\sigma}^{2}}^{(2)}_{0,k}}}
\right)$, obtained using the analytical approximation, suggested by \citet{mackay1992evidence}, to the expectation of the logistic function of a Gaussian random variable.

\subsection{Zero-Inflated Negative Binomial}
\begin{align*}
    \tilde{\mu}^{Y}_{0,k} =& \left(1-\bar{\pi}_{0,k}\right)
    \exp\left\{\mu^{(1)}_{0,k}+\frac{{{\sigma}^{2}}^{(1)}_{0,k}}{2}\right\}\,,\\
    (\tilde{\sigma}^{Y}_{0,k})^2 =& \left(1-\bar{\pi}_{0,k}\right)
    \Bigg[
    \exp\left\{\mu^{(1)}_{0,k}+\frac{{{\sigma}^{2}}^{(1)}_{0,k}}{2}\right\}
    +
    \exp\left\{
    2\mu^{(1)}_{0,k}
    +2{{\sigma}^{2}}^{(1)}_{0,k}
    +\mu^{(2)}_{0,k}
    +\frac{{{\sigma}^{2}}^{(2)}_{0,k}}{2}
    \right\}\\
    & + \left(e^{{{\sigma}^{2}}^{(1)}_{0,k}}-1\right)
    \exp\left\{2\mu^{(1)}_{0,k}+{{\sigma}^{2}}^{(1)}_{0,k}\right\}
    \Bigg]+\bar{\pi}_{0,k}
    \left(1-\bar{\pi}_{0,k}\right)
    \exp\left\{2\mu^{(1)}_{0,k}+{{\sigma}^{2}}^{(1)}_{0,k}\right\}\,,
\end{align*}
where $\bar{\pi}_{0,k} =
\operatorname{logit}^{-1}\left(
\frac{\mu^{(3)}_{0,k}}
{\sqrt{1+\frac{\pi}{8}{{\sigma}^{2}}^{(3)}_{0,k}}}
\right)$.

\begin{table}
\footnotesize{
\caption{Different choices of $p(y|\boldsymbol{\phi})$ for which predictive mean $\tilde{\mu}^{Y}_{0,k}$ and variance $(\tilde{\sigma}^{Y}_{0,k})^2$ admit closed-form expressions.}
\label{tab:tab_a1}	
\begin{center}
\begin{tabular}{lccc}
\toprule
& Parameter ($\boldsymbol{\phi}$) & Link Function ($\mathbf{g}$) & Probability Function ($p(y|\boldsymbol{\phi})$)  \\
\midrule
\textbf{\makecell[l]{Poisson}} & \makecell[l]{$\phi_1=\lambda\in(0,\infty)$} & log & $\frac{\lambda^y}{y!}e^{-\lambda}$\\
\addlinespace[0.3cm]
\textbf{\makecell[l]{Exponential}} & \makecell[l]{$\phi_1=\mu\in(0,\infty)$} & log & $\frac{1}{\mu}e^{-\frac{y}{\mu}}$\\
\addlinespace[0.3cm]
\textbf{\makecell[l]{Gamma}} & \makecell[l]{$\phi_1=\mu\in(0,\infty)$\\$\phi_2=\sigma\in(0,\infty)$} & \makecell[c]{log\\log} & $\frac{y^{1/\sigma^2-1}e^{-y/(\mu\sigma^2)}}{(\mu\sigma^2)^{1/\sigma^2}\Gamma(1/\sigma^2)}$\\
\addlinespace[0.3cm]
\textbf{\makecell[l]{Heteroskedastic\\Gaussian}} & \makecell[l]{$\phi_1=\mu\in\mathbb{R}$\\$\phi_2=\sigma^2\in(0,\infty)$} & \makecell[c]{identity\\log} &  $\frac{1}{\sqrt{2\pi\sigma^2}}e^{-\frac{(y-\mu)^2}{2\sigma^2}}$\\
\addlinespace[0.3cm]
\textbf{\makecell[l]{Negative\\Binomial}} 
& \makecell[l]{$\phi_1=\mu\in(0,\infty)$\\$\phi_2=\sigma\in(0,\infty)$}
& \makecell[c]{log\\log} 
& \makecell[c]{
$\frac{\Gamma(y+\frac{1}{\sigma})}{\Gamma(1/{\sigma})\Gamma(y+1)}\left(\frac{\sigma\mu}{1+\sigma\mu}\right)^y\left(\frac{1}{1+\sigma\mu}\right)^{1/{\sigma}}$}\\
\addlinespace[0.3cm]
\textbf{\makecell[l]{Zero-Inflated\\Poisson}} 
& \makecell[l]{$\phi_1=\lambda\in(0,\infty)$\\$\phi_2=\pi\in(0,1)$} 
& \makecell[c]{log\\logit} 
& \makecell[c]{
$\begin{cases}
\pi+(1-\pi)e^{-\lambda}, & y=0,\\
(1-\pi)\frac{\lambda^y}{y!}e^{-\lambda}, & y>0
\end{cases}$}\\
\addlinespace[0.3cm]
\textbf{\makecell[l]{Zero-Inflated\\Negative\\Binomial}} 
& \makecell[l]{$\phi_1=\mu\in(0,\infty)$\\$\phi_2=\sigma\in(0,\infty)$\\$\phi_3=\pi\in(0,1)$} 
& \makecell[c]{log\\log\\logit} 
& \makecell[c]{
$\begin{cases}
\pi+(1-\pi)\left(\frac{1}{1+\sigma\mu}\right)^{1/\sigma}, & y=0,\\
(1-\pi)\frac{\Gamma(y+\frac{1}{\sigma})}{\Gamma(1/\sigma)\Gamma(y+1)}
\left(\frac{\sigma\mu}{1+\sigma\mu}\right)^y
\left(\frac{1}{1+\sigma\mu}\right)^{1/\sigma}, & y>0
\end{cases}$}\\
\bottomrule
\end{tabular}
\end{center}
}
\end{table}

\section{Proof of Proposition~\ref{prop:vecchia_lgp}}
\label{app:vecchia_lgp}
Since, when $l=1$, the LGP reduces to a GP and, under the Vecchia approximation, the predictive mean and variance of a GP at a single input location $\mathbf{x}_0$ coincide with those of the nearest-neighbor GP~\citep{katzfuss2020vecchia}, it follows that, under the Vecchia approximation, the LGP mean and variance for a GP node in the first layer ($l=1$) are given by
\begin{align*}
{\mu}^{(q)}_{1\rightarrow 1,k}(\mathbf{x}_0)&=\mathbf{r}^{(q)}_{\mathcal{C}}(\mathbf{x}_0)^\top\left(\mathbf{R}^{(q)}_{1,\mathcal{C}}\right)^{-1}\mathbf{w}^{(q)}_{1,k,\mathcal{C}}\\
{{\sigma}^{2}}^{(q)}_{1\rightarrow 1,k}(\mathbf{x}_0)&=\left(\sigma^{(q)}_1\right)^2\left(1+\eta^{(q)}_1-\mathbf{r}^{(q)}_{\mathcal{C}}(\mathbf{x}_0)^\top\left(\mathbf{R}^{(q)}_{1,\mathcal{C}}\right)^{-1}\mathbf{r}^{(q)}_{\mathcal{C}}(\mathbf{x}_0)\right)\,,
\end{align*}
where $\mathcal{C}\subseteq\{1,2,\dots,N\}$ is a conditioning index set of size $|\mathcal{C}|$. This establishes Equations~\eqref{eq:vecchia_gp_mean} and~\eqref{eq:vecchia_gp_var}.

Analogously, for a particular GP node $\mathcal{GP}^{(q)}_l$ with $l\geq 2$, its predictive mean $\mu^{(q)}_{l,k}$ and variance ${\sigma^2}^{(q)}_{l,k}$ under the Vecchia approximation for imputation $k$ are given by
\begin{align*}
\mu^{(q)}_{l,k}=&\mathbf{r}^{(q)}_{l,k,\mathcal{C}}(\mathbf{W}_{0,l-1,k})^\top\left(\mathbf{R}^{(q)}_{l,k,\mathcal{C}}\right)^{-1}\mathbf{w}^{(q)}_{l,k,\mathcal{C}} \\
{\sigma^2}^{(q)}_{l,k}=&\left(\sigma^{(q)}_l\right)^2\left(1+\eta^{(q)}_l-\mathbf{r}^{(q)}_{l,k,\mathcal{C}}(\mathbf{W}_{0,l-1,k})^\top\left(\mathbf{R}^{(q)}_{l,k,\mathcal{C}}\right)^{-1}\mathbf{r}^{(q)}_{l,k,\mathcal{C}}(\mathbf{W}_{0,l-1,k})\right)\,, 
\end{align*}
where $\mathbf{W}_{0,l-1,k}=(W^{(1)}_{0,l-1,k},\dots,W^{(P_{l-1})}_{0,l-1,k})$ with $W^{(q)}_{0,l-1,k}\sim\mathcal{N}\left(\mu^{(q)}_{1\rightarrow (l-1),k}(\mathbf{x}_0),{\sigma^2}^{(q)}_{1\rightarrow (l-1),k}(\mathbf{x}_0)\right)$.

Applying the laws of total expectation and variance to $W^{(q)}_{0,l,k}$ then gives
\begin{align*}
\mu^{(q)}_{1\rightarrow l,k}(\mathbf{x}_0)=&\mathbb{E}\left(\mathbf{r}^{(q)}_{l,k,\mathcal{C}}(\mathbf{W}_{0,l-1,k})\right)^\top\left(\mathbf{R}^{(q)}_{l,k,\mathcal{C}}\right)^{-1}\mathbf{w}^{(q)}_{l,k,\mathcal{C}}, \\
{\sigma^2}^{(q)}_{1\rightarrow l,k}(\mathbf{x}_0)=&\left(\mathbf{w}^{(q)}_{l,k,\mathcal{C}}\right)^\top\left(\mathbf{R}^{(q)}_{l,k,\mathcal{C}}\right)^{-1}\mathbb{E}\left(\mathbf{r}^{(q)}_{l,k,\mathcal{C}}(\mathbf{W}_{0,l-1,k})\mathbf{r}^{(q)}_{l,k,\mathcal{C}}(\mathbf{W}_{0,l-1,k})^\top\right)\left(\mathbf{R}^{(q)}_{l,k,\mathcal{C}}\right)^{-1}\mathbf{w}^{(q)}_{l,k,\mathcal{C}}\\
&\quad-\left(\mathbb{E}\left(\mathbf{r}^{(q)}_{l,k,\mathcal{C}}(\mathbf{W}_{0,l-1,k})\right)^\top\left(\mathbf{R}^{(q)}_{l,k,\mathcal{C}}\right)^{-1}\mathbf{w}^{(q)}_{l,k,\mathcal{C}}\right)^2\nonumber\\
&\quad+\left(\sigma^{(q)}_{l}\right)^2\left(1+\eta^{(q)}_{l}-\mathrm{tr}\left\{\left(\mathbf{R}^{(q)}_{l,k,\mathcal{C}}\right)^{-1}\mathbb{E}\left(\mathbf{r}^{(q)}_{l,k,\mathcal{C}}(\mathbf{W}_{0,l-1,k})\mathbf{r}^{(q)}_{l,k,\mathcal{C}}(\mathbf{W}_{0,l-1,k})^\top\right)\right\}\right)\,,
\end{align*}
where the expectations are taken with respect to $\mathbf{W}_{0,l-1,k}$. Defining $\mathbf{I}_{l,k,\mathcal{C}}^{(q)}(\mathbf{x}_0)=\mathbb{E}\left(\mathbf{r}^{(q)}_{l,k,\mathcal{C}}(\mathbf{W}_{0,l-1,k})\right)$ and $\mathbf{J}^{(q)}_{l,k,\mathcal{C}}(\mathbf{x}_0)=\mathbb{E}\left(\mathbf{r}^{(q)}_{l,k,\mathcal{C}}(\mathbf{W}_{0,l-1,k})\mathbf{r}^{(q)}_{l,k,\mathcal{C}}(\mathbf{W}_{0,l-1,k})^\top\right)$, we then establish Equations~\eqref{eq:vecchia_linkgp_mean} and~\eqref{eq:vecchia_linkgp_var}, with the $i$-th element of $\mathbf{I}_{l,k,\mathcal{C}}^{(q)}(\mathbf{x}_0)$ given by
\begin{equation*}
\prod_{d=1}^{P_{l-1}}\mathbb{E}\left(k^{(q)}_{l,d}\left(W^{(d)}_{0,l-1,k},\,(\mathbf{w}^{(d)}_{l-1,k})_{\mathcal{C}_i}\right)\right)
\end{equation*}
and the $ij$-th element of $\mathbf{J}^{(q)}_{l,k,\mathcal{C}}(\mathbf{x}_0)$ given by
\begin{equation*}
\prod_{d=1}^{P_{l-1}}\mathbb{E}\left(k^{(q)}_{l,d}\left(W^{(d)}_{0,l-1,k},\,(\mathbf{w}^{(d)}_{l-1,k})_{\mathcal{C}_i}\right)k^{(q)}_{l,d}\left(W^{(d)}_{0,l-1,k},\,(\mathbf{w}^{(d)}_{l-1,k})_{\mathcal{C}_j}\right)\right)\,,
\end{equation*}
which can be expressed as
\begin{equation*}
\prod_{d=1}^{P_{l-1}}\xi^{(q)}_{l}\left({\mu}^{(d)}_{1\rightarrow(l-1),k}(\mathbf{x}_0),\,{{\sigma}^{2}}^{(d)}_{1\rightarrow(l-1),k}(\mathbf{x}_0),\,(\mathbf{w}^{(d)}_{l-1,k})_{\mathcal{C}_i}\right)
\end{equation*}
and
\begin{equation*}
\prod_{d=1}^{P_{l-1}}\zeta^{(q)}_{l}\left({\mu}^{(d)}_{1\rightarrow(l-1),k}(\mathbf{x}_0),\,{{\sigma}^{2}}^{(d)}_{1\rightarrow(l-1),k}(\mathbf{x}_0),\,(\mathbf{w}^{(d)}_{l-1,k})_{\mathcal{C}_i},\,(\mathbf{w}^{(d)}_{l-1,k})_{\mathcal{C}_j}\right)
\end{equation*}
respectively. These quantities can consequently be computed analytically, since $\xi^{(q)}_{l}(\cdot,\cdot,\cdot)$ and $\zeta^{(q)}_{l}(\cdot,\cdot,\cdot,\cdot)$ admit closed-form expressions~\citep[Appendix A]{ming2021linked} when the kernel functions $k^{(q)}_{l,d}(\cdot,\cdot)$ are squared exponential or Matérn.

\section{Proof of Proposition~\ref{prop:hetero_case1}}
\label{app:hetero_case1}
Given $\boldsymbol{\mu}$, we have that
$\mathbf{Y}|\boldsymbol{\mu},\log\boldsymbol{\sigma}^2\sim\mathcal{N}(\mathbf{M}\boldsymbol{\mu}, \boldsymbol{\Lambda})$, where {$\boldsymbol{\Lambda}={\diag}(\mathbf{M}\boldsymbol{\sigma}^2)$} with $\boldsymbol{\sigma}^2=(\sigma^2_1,\dots,\sigma^2_N)^\top$. Since $\boldsymbol{\mu}=\mathbf{F}^{(1)}$ and $p\left(\mathbf{F}^{(1)}|\{\mathbf{W}^{(p)}_l\},\mathbf{x}\right) = p\left(\mathbf{F}^{(1)}|\mathbf{W}_{L-1}\right)\sim\mathcal{N}\left(\mathbf{0},\boldsymbol{\Sigma}^{(1)}_L(\mathbf{w}_{L-1})\right)$, we have 
\begin{align*}
p\left(\boldsymbol{\mu}|\log\boldsymbol{\sigma}^2, \{\mathbf{w}^{(p)}_l\},\mathbf{y},\mathbf{x}\right) &\propto  p(\mathbf{y}|\boldsymbol{\mu}, \boldsymbol{\sigma}^2)\,p\left(\boldsymbol{\mu}|\{\mathbf{w}^{(p)}_l\},\mathbf{x}\right)\\
&\propto \exp\left\{-\frac{1}{2}(\mathbf{y}-\mathbf{M}\boldsymbol{\mu})^\top\boldsymbol{\Lambda}^{-1}(\mathbf{y}-\mathbf{M}\boldsymbol{\mu})\right\} \exp\left\{-\frac{1}{2}\boldsymbol{\mu}^\top{\boldsymbol{\Sigma}^{(1)}_L(\mathbf{w}_{L-1})}^{-1}\boldsymbol{\mu}\right\}\\
&\propto \exp\left\{ -\frac{1}{2}\boldsymbol{\mu}^\top\left(\mathbf{M}^\top\boldsymbol{\Lambda}^{-1}\mathbf{M} + {\boldsymbol{\Sigma}^{(1)}_L(\mathbf{w}_{L-1})}^{-1}\right)\boldsymbol{\mu} + \boldsymbol{\mu}^\top\mathbf{M}^\top\boldsymbol{\Lambda}^{-1}\mathbf{y} \right\}\,.
\end{align*}

By completing the square, we can express $p\left(\boldsymbol{\mu}|\log\boldsymbol{\sigma}^2, \{\mathbf{w}^{(p)}_l\},\mathbf{y},\mathbf{x}\right)$ in the standard form for a multivariate normal, giving
\begin{equation*}
    \boldsymbol{\mu}|\log\boldsymbol{\sigma}^2, \{\mathbf{W}^{(p)}_l\},\mathbf{Y},\mathbf{x} \sim \mathcal{N}\left(\left(\mathbf{M}^\top\boldsymbol{\Lambda}^{-1}\mathbf{M} + {\boldsymbol{\Sigma}^{(1)}_L(\mathbf{w}_{L-1})}^{-1}\right)^{-1}\mathbf{M}^\top\boldsymbol{\Lambda}^{-1}\mathbf{y},\,\left(\mathbf{M}^\top\boldsymbol{\Lambda}^{-1}\mathbf{M} + {\boldsymbol{\Sigma}^{(1)}_L(\mathbf{w}_{L-1})}^{-1}\right)^{-1}\right)\,,
\end{equation*}
which is equivalent to multivariate normal distribution in
{Equation~\eqref{eq:hetero_case1}} by factoring out ${\boldsymbol{\Sigma}^{(1)}_L(\mathbf{w}_{L-1})}^{-1}$ from the expression within the brackets.

\section{Proof of Proposition~\ref{prop:hetero_case2}}
\label{app:hetero_case2}
Let $\mathbf{D}=(\mathbf{Y}|\log\boldsymbol{\sigma}^2, \{\mathbf{W}^{(p)}_l\},\mathbf{x})\;-\;(\boldsymbol{\mu}|\{\mathbf{W}^{(p)}_l\},\mathbf{x})$. Since 
\begin{equation*}
p\left(\mathbf{Y}|\log\boldsymbol{\sigma}^2,\{\mathbf{W}^{(p)}_l\},\mathbf{x},\boldsymbol{\mu}\right) = p\left(\mathbf{Y}|\log\boldsymbol{\sigma}^2,\boldsymbol{\mu}\right) \sim\mathcal{N}(\boldsymbol{\mu},\boldsymbol{\Gamma})\,,
\end{equation*}
where $\boldsymbol{\Gamma}={\diag}(\sigma^2_1,\dots,\sigma^2_N)$, and 
\begin{equation*}
    p\left(\boldsymbol{\mu}|\{\mathbf{W}^{(p)}_l\},\mathbf{x}\right) = p\left(\mathbf{F}^{(1)}|\mathbf{W}_{L-1}\right)\sim\mathcal{N}\left(\mathbf{0},\boldsymbol{\Sigma}^{(1)}_L(\mathbf{w}_{L-1})\right)\,,
\end{equation*}
we have that $\mathbf{D}|\boldsymbol{\mu} \sim \mathcal{N}(\mathbf{0},\boldsymbol{\Gamma})$. {Since the conditional distribution of $\mathbf{D}$ given $\boldsymbol{\mu}$ is $\mathcal{N}(\mathbf{0},\boldsymbol{\Gamma})$ for every value of $\boldsymbol{\mu}$, we have that $\mathbf{D}$ and $\boldsymbol{\mu}$ are independent.} Consequently, $\mathbf{D} \sim \mathcal{N}(\mathbf{0},\boldsymbol{\Gamma})$. Let $\mathbf{V}=\begin{pmatrix}\mathbf{F}^{(1)}|\mathbf{W}_{L-1}\\ \mathbf{D}\end{pmatrix}$. Since $p(\boldsymbol{\mu}|\{\mathbf{W}^{(p)}_l\},\mathbf{x})=p(\mathbf{F}^{(1)}|\mathbf{W}_{L-1})$, we have 
\begin{equation*}
\mathbf{V}\sim\mathcal{N}\left(\mathbf{0},
\begin{pmatrix}
\boldsymbol{\Sigma}^{(1)}_L(\mathbf{w}_{L-1}) & \mathbf{0}\\
\mathbf{0} & \boldsymbol{\Gamma}
\end{pmatrix}\right)\,.
\end{equation*}
Define 
{\begin{equation}
\label{eq:z}
\mathbf{Z}=\begin{pmatrix}\mathbf{Y}|\mathbf{W}_{L-1},\log\boldsymbol{\sigma}^2\\ \mathbf{F}^{(1)}|\mathbf{W}_{L-1}\end{pmatrix}\,.
\end{equation}}
Since $\mathbf{Y}|\mathbf{W}_{L-1} = \mathbf{F}^{(1)}|\mathbf{W}_{L-1}\;+\;\mathbf{D}$, we have that $\mathbf{Z}$ is an affine transformation of $\mathbf{V}$:
\begin{equation*}
    \mathbf{Z} = \begin{pmatrix}
    \mathbf{I}_N,\mathbf{I}_N\\
    \mathbf{I}_N, \mathbf{0}
    \end{pmatrix}
    \mathbf{V}\,,
\end{equation*}
where $\mathbf{I}_N$ is the identity matrix. Thus, $\mathbf{Z}\sim\mathcal{N}\left(\mathbf{0},
\boldsymbol{\Sigma}_{\mathbf{Z}}\right)$, where
\begin{equation*}
    \boldsymbol{\Sigma}_{\mathbf{Z}} = \begin{pmatrix}
\boldsymbol{\Sigma}^{(1)}_L(\mathbf{w}_{L-1}) + \boldsymbol{\Gamma} & \boldsymbol{\Sigma}^{(1)}_L(\mathbf{w}_{L-1})\\
\boldsymbol{\Sigma}^{(1)}_L(\mathbf{w}_{L-1}) & \boldsymbol{\Sigma}^{(1)}_L(\mathbf{w}_{L-1})
\end{pmatrix}\in\mathbb{R}^{2N\times2N}\,.
\end{equation*}
The covariance matrix $\boldsymbol{\Sigma}_{\mathbf{Z}}$ of $\mathbf{Z}$ can then be constructed by $(\sigma^{(1)}_L)^2\mathbf{R}_{\mathbf{Z}}(\mathbf{w}^{\text{stack}}_{L-1})+\boldsymbol{\Gamma}_0$, where $\mathbf{w}^{\text{stack}}_{L-1}=\begin{pmatrix}\mathbf{w}_{L-1}\\\mathbf{w}_{L-1}\end{pmatrix}$, $\boldsymbol{\Gamma}_0={\diag}(\sigma^2_1,\dots,\sigma^2_N, 0,\dots,0)$, and the $ij$-th element of $\mathbf{R}_{\mathbf{Z}}(\mathbf{w}^{\text{stack}}_{L-1})$ is specified by $
k_L^{(1)}(\mathbf{w}^{\text{stack}}_{L-1,i*},\,\mathbf{w}^{\text{stack}}_{L-1,j*})+\eta\mathds{1}_{\{i=j\}}$ for $i,j=1,\dots,2N$. Since $\mathbf{Z}$ is multivariate normal, under the Vecchia approximation $\mathbf{Z}\overset{\text{Vec}}{\sim}\mathcal{N}(\mathbf{0},\mathbf{P}^{-1})$, where $\mathbf{P}$ is the precision matrix that admits the upper-lower Cholesky decomposition $\mathbf{P}=\mathbf{U}\mathbf{U}^\top$, with $\mathbf{U}$ being a sparse upper-triangular matrix. Conformably with the partition of $\mathbf{Z}$ given in Equation~\eqref{eq:z}, $\mathbf{U}$ can be written as
\begin{equation*}
  \mathbf{U}=\begin{bmatrix}
  \mathbf{U}_{\mathbf{Y}\mathbf{Y}} & \mathbf{U}_{\mathbf{Y}\mathbf{F}^{(1)}} \\
  \mathbf{0} & \mathbf{U}_{\mathbf{F}^{(1)}\mathbf{F}^{(1)}}
  \end{bmatrix}\,,
\end{equation*}
where $\mathbf{U}_{\mathbf{Y}\mathbf{Y}}\in\mathbb{R}^{N\times N}$,
$\mathbf{U}_{\mathbf{Y}\mathbf{F}^{(1)}}\in\mathbb{R}^{N\times N}$, and
$\mathbf{U}_{\mathbf{F}^{(1)}\mathbf{F}^{(1)}}\in\mathbb{R}^{N\times N}$ denotes sub-matrices of $\mathbf{U}$ under this partition. Thus, we have
\begin{equation*}
    \mathbf{P} =
\begin{bmatrix}
\mathbf{U}_{\mathbf{Y}\mathbf{Y}}\mathbf{U}_{\mathbf{Y}\mathbf{Y}}^\top
+
\mathbf{U}_{\mathbf{Y}\mathbf{F}^{(1)}}\mathbf{U}_{\mathbf{Y}\mathbf{F}^{(1)}}^\top
&
\mathbf{U}_{\mathbf{Y}\mathbf{F}^{(1)}}\mathbf{U}_{\mathbf{F}^{(1)}\mathbf{F}^{(1)}}^\top
\\[4pt]
\mathbf{U}_{\mathbf{F}^{(1)}\mathbf{F}^{(1)}}\mathbf{U}_{\mathbf{Y}\mathbf{F}^{(1)}}^\top
&
\mathbf{U}_{\mathbf{F}^{(1)}\mathbf{F}^{(1)}}\mathbf{U}_{\mathbf{F}^{(1)}\mathbf{F}^{(1)}}^\top
\end{bmatrix}\,.
\end{equation*}
According to~\citet[Theorem~12.2]{gelfand2010handbook}, we then have that, under the Vecchia approximation, the distribution of $\mathbf{F}^{(1)}$ conditional on $\mathbf{W}_{L-1}$, $\log\boldsymbol{\sigma}^2$, $\mathbf{Y}$ is given by the multivariate normal distribution in {Equation~\eqref{eq:hetero_case2}}. Since $p\left(\boldsymbol{\mu} \mid \log \boldsymbol{\sigma}^2, \{\mathbf{W}^{(p)}_l\}, \mathbf{Y}, \mathbf{x}\right)=p\left(\mathbf{F}^{(1)} \mid \log\boldsymbol{\sigma}^2, \mathbf{W}_{L-1}, \mathbf{Y}\right)$, the proposition follows.

\section{Proof of Proposition~\ref{prop:hetero_case3}}
\label{app:hetero_case3}
Given $\boldsymbol{\mu}$, we have that
\begin{equation}
\label{eq:ygivenmu}
\mathbf{Y}|\boldsymbol{\mu},\log\boldsymbol{\sigma}^2\sim\mathcal{N}(\mathbf{M}\boldsymbol{\mu}, \boldsymbol{\Lambda})\,,
\end{equation}
where {$\boldsymbol{\Lambda}={\diag}(\mathbf{M}\boldsymbol{\sigma}^2)$} with $\boldsymbol{\sigma}^2=(\sigma^2_1,\dots,\sigma^2_N)^\top$. The likelihood is then given by:
\begin{equation*}
    p(\mathbf{y}|\boldsymbol{\mu},\log\boldsymbol{\sigma}^2)=(2\pi)^{-\sum^{N}_{i=1}S_i/2}|\boldsymbol{\Lambda}|^{-1/2}\exp\{-\frac{1}{2}(\mathbf{y}-\mathbf{M}\boldsymbol{\mu})^\top\boldsymbol{\Lambda}^{-1}(\mathbf{y}-\mathbf{M}\boldsymbol{\mu})\}\,.
\end{equation*}
Expand the quadratic form, we have:
\begin{equation*}
(\mathbf{y}-\mathbf{M}\boldsymbol{\mu})^\top\boldsymbol{\Lambda}^{-1}(\mathbf{y}-\mathbf{M}\boldsymbol{\mu}) = \mathbf{y}^\top\boldsymbol{\Lambda}^{-1}\mathbf{y}-2\boldsymbol{\mu}^\top\mathbf{M}^\top\boldsymbol{\Lambda}^{-1}\mathbf{y}+\boldsymbol{\mu}^\top\mathbf{M}^\top\boldsymbol{\Lambda}^{-1}\mathbf{M}\boldsymbol{\mu}\,.
\end{equation*}
Define $s(\mathbf{y})\overset{\text{def}}{=}\mathbf{M}^\top\boldsymbol{\Lambda}^{-1}\mathbf{y}\in\mathbb{R}^N$ and ${\boldsymbol{\Delta}}\overset{\text{def}}{=}\mathbf{M}^\top\boldsymbol{\Lambda}^{-1}\mathbf{M}\in\mathbb{R}^{N\times N}$, then
\begin{equation*}
p(\mathbf{y}|\boldsymbol{\mu},\log\boldsymbol{\sigma}^2)=(2\pi)^{-\sum^{N}_{i=1}S_i/2}|\boldsymbol{\Lambda}|^{-1/2}\exp\{-\frac{1}{2}\mathbf{y}^\top\boldsymbol{\Lambda}^{-1}\mathbf{y}\}\exp\{\boldsymbol{\mu}^\top s(\mathbf{y})-\frac{1}{2}\boldsymbol{\mu}^\top{\boldsymbol{\Delta}}\boldsymbol{\mu}\}\,.
\end{equation*}
Let $h(\mathbf{y})=(2\pi)^{-\sum^{N}_{i=1}S_i/2}|\boldsymbol{\Lambda}|^{-1/2}\exp\{-\frac{1}{2}\mathbf{y}^\top\boldsymbol{\Lambda}^{-1}\mathbf{y}\}$ and $g(s(\mathbf{y}),\boldsymbol{\mu})=\exp\{\boldsymbol{\mu}^\top s(\mathbf{y})-\frac{1}{2}\boldsymbol{\mu}^\top{\boldsymbol{\Delta}}\boldsymbol{\mu}\}$, we have 
\begin{equation*}
p(\mathbf{y}|\boldsymbol{\mu},\log\boldsymbol{\sigma}^2)=h(\mathbf{y})g(s(\mathbf{y}),\boldsymbol{\mu})\,.
\end{equation*}
Thus,
\begin{align*}
    p(\boldsymbol{\mu}|\log\boldsymbol{\sigma}^2,\mathbf{w}_{L-1},\mathbf{y})&\propto p(\mathbf{y}|\boldsymbol{\mu},\log\boldsymbol{\sigma}^2)p(\boldsymbol{\mu}|\mathbf{w}_{L-1})\\
&=h(\mathbf{y})g(s(\mathbf{y}),\boldsymbol{\mu})p(\boldsymbol{\mu}|\mathbf{w}_{L-1})\\
&\propto g(s(\mathbf{y}),\boldsymbol{\mu})p(\boldsymbol{\mu}|\mathbf{w}_{L-1})\,,
\end{align*}
which depends on $\mathbf{y}$ only through $s(\mathbf{y})$. Therefore, $p(\boldsymbol{\mu}|\log\boldsymbol{\sigma}^2,\mathbf{w}_{L-1},\mathbf{y}) =  p(\boldsymbol{\mu}|\log\boldsymbol{\sigma}^2,\mathbf{w}_{L-1},s(\mathbf{y}))$.
Define 
\begin{equation}
\label{eq:ytilde}
\tilde{\mathbf{Y}}=(\mathbf{M}^T\boldsymbol{\Lambda}^{-1}\mathbf{M})^{-1}\mathbf{M}^T\boldsymbol{\Lambda}^{-1}\mathbf{Y}\in\mathbb{R}^{N}\,.
\end{equation}
Since $\tilde{\mathbf{Y}}={\boldsymbol{\Delta}}^{-1}s(\mathbf{Y})$ is a deterministic transformation of $s(\mathbf{Y})$, conditioning on $s(\mathbf{Y})$ is equivalent to conditioning on $\tilde{\mathbf{Y}}$:
\begin{equation*}
p(\boldsymbol{\mu}|\log\boldsymbol{\sigma}^2,\mathbf{w}_{L-1},\mathbf{y}) =  p(\boldsymbol{\mu}|\log\boldsymbol{\sigma}^2,\mathbf{w}_{L-1},s(\mathbf{y})) =  p(\boldsymbol{\mu}|\log\boldsymbol{\sigma}^2,\mathbf{w}_{L-1},\tilde{\mathbf{y}})\,.
\end{equation*}
As a result, finding $p(\boldsymbol{\mu}|\log\boldsymbol{\sigma}^2,\mathbf{w}_{L-1},\mathbf{y})$ is equivalent to finding $p(\boldsymbol{\mu}|\log\boldsymbol{\sigma}^2,\mathbf{w}_{L-1},\tilde{\mathbf{y}})$. Using Equation~\eqref{eq:ygivenmu} and Definition~\eqref{eq:ytilde}, we have that $\tilde{\mathbf{Y}}|\boldsymbol{\mu},\log\boldsymbol{\sigma}^2$ is a multivariate normal: $\tilde{\mathbf{Y}}|\boldsymbol{\mu},\log\boldsymbol{\sigma}^2\sim\mathcal{N}(\boldsymbol{\mu}, {\boldsymbol{\Delta}}^{-1})$. Define 
\begin{equation*}
\mathbf{Z}=\begin{pmatrix}\tilde{\mathbf{Y}}|\mathbf{W}_{L-1},\log\boldsymbol{\sigma}^2\\ \mathbf{F}^{(1)}|\mathbf{W}_{L-1}\end{pmatrix}\,,
\end{equation*}
then following the same arguments in Appendix~\ref{app:hetero_case2}, we can obtain that $\mathbf{Z}\sim\mathcal{N}\left(\mathbf{0},
\boldsymbol{\Sigma}_{\mathbf{Z}}\right)$, where
\begin{equation*}
    \boldsymbol{\Sigma}_{\mathbf{Z}} = \begin{pmatrix}
\boldsymbol{\Sigma}^{(1)}_L(\mathbf{w}_{L-1}) + {\boldsymbol{\Delta}}^{-1} & \boldsymbol{\Sigma}^{(1)}_L(\mathbf{w}_{L-1})\\
\boldsymbol{\Sigma}^{(1)}_L(\mathbf{w}_{L-1}) & \boldsymbol{\Sigma}^{(1)}_L(\mathbf{w}_{L-1})
\end{pmatrix}\in\mathbb{R}^{2N\times2N}\,.
\end{equation*}
Using the same arguments in Appendix~\ref{app:hetero_case2} again, we obtain that, under the Vecchia approximation, the
distribution of $\mathbf{F}^{(1)}$ conditional on $\mathbf{W}_{L-1}$, $\log\boldsymbol{\sigma}^2$, and $\tilde{\mathbf{Y}}$ is given by the multivariate normal distribution in {Equation~\eqref{eq:hetero_case3}}. Since $p(\boldsymbol{\mu} \mid \log \boldsymbol{\sigma}^2, \{\mathbf{w}^{(p)}_l\}, \mathbf{y}, \mathbf{x})=p(\boldsymbol{\mu} \mid \log\boldsymbol{\sigma}^2, \mathbf{w}_{L-1}, \mathbf{y})$ and $p(\boldsymbol{\mu} \mid \log\boldsymbol{\sigma}^2, \mathbf{w}_{L-1}, \mathbf{y}) = p(\boldsymbol{\mu} \mid \log\boldsymbol{\sigma}^2, \mathbf{w}_{L-1}, \tilde{\mathbf{y}}) = p(\mathbf{f}^{(1)} \mid \log\boldsymbol{\sigma}^2, \mathbf{w}_{L-1}, \tilde{\mathbf{y}})$, the proposition follows.
\end{appendices}

\vfill

\pagebreak
\begin{center}
\label{supp}
\textbf{\Large Supplementary Materials}
\end{center}
\setcounter{equation}{0}
\setcounter{section}{0}
\setcounter{figure}{0}
\setcounter{table}{0}
\makeatletter
\renewcommand{\thetable}{S.\arabic{table}}
\renewcommand{\thesection}{S.\arabic{section}}
\renewcommand{\theequation}{S\arabic{equation}} 
\renewcommand{\thefigure}{S.\arabic{figure}} 

\section{Introduction}

This supplementary material reports additional results for GDGP using a three-layer (i.e., $L=3$) DGP structure in all the experiments presented in the main manuscript. In this structure, each of the first two layers contains a number of GP nodes equal to the input dimension, while the third layer contains a number of GP nodes equal to the number of parameters of the associated likelihood. Except for the number of DGP layers used in the GDGP hierarchy, all experimental designs, GP and likelihood specifications, training and test data, inference settings, and evaluation metrics remain the same as in the corresponding experiments in the main manuscript.

\section{Heteroskedastic Gaussian Likelihood}

\subsection{Heteroskedastic Step Function}

\begin{figure}[H]
\centering 
\subfloat[2 layers]{\label{fig:figure_s1a}\includegraphics[width=0.37\linewidth]{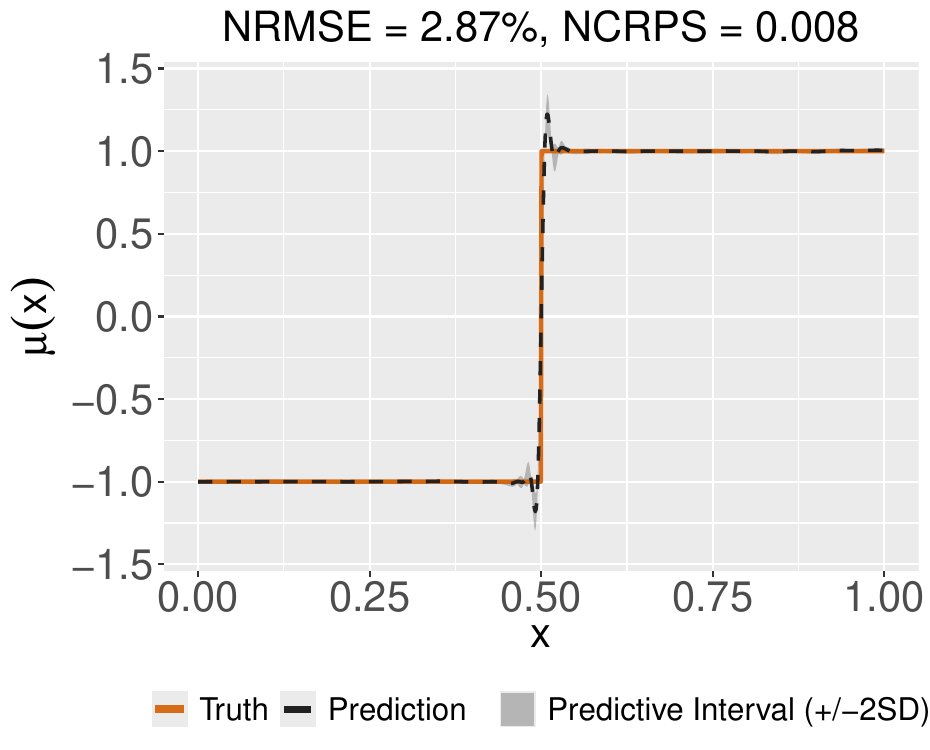}}
\hspace{3em}
\subfloat[3 layers]{\label{fig:figure_s1b}\includegraphics[width=0.37\linewidth]{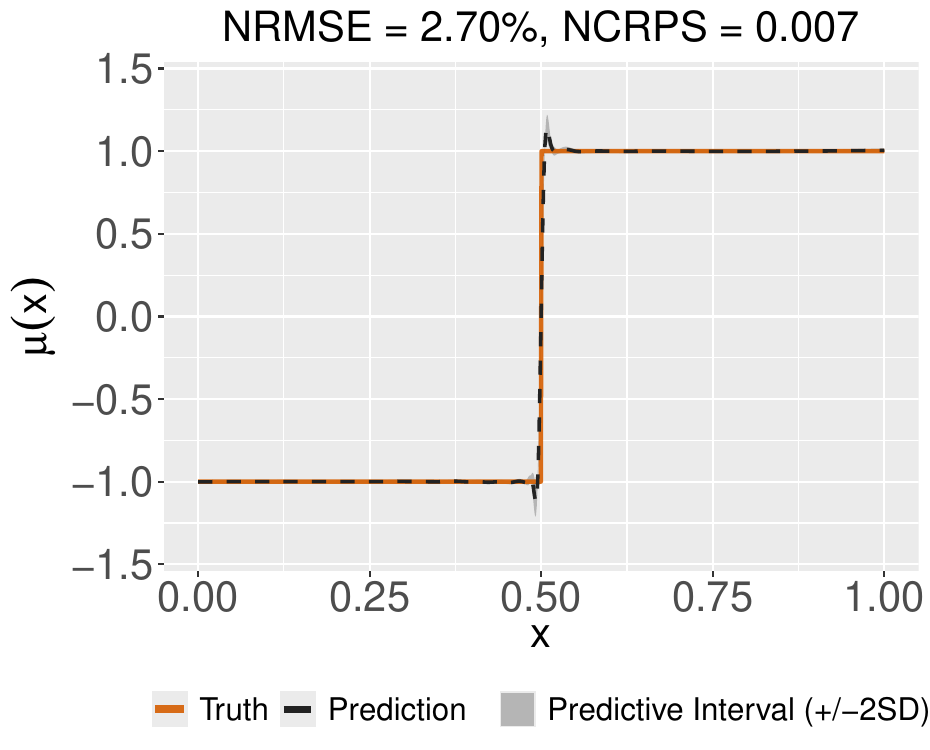}}\\
\subfloat[2 layers]{\label{fig:figure_s1c}\includegraphics[width=0.37\linewidth]{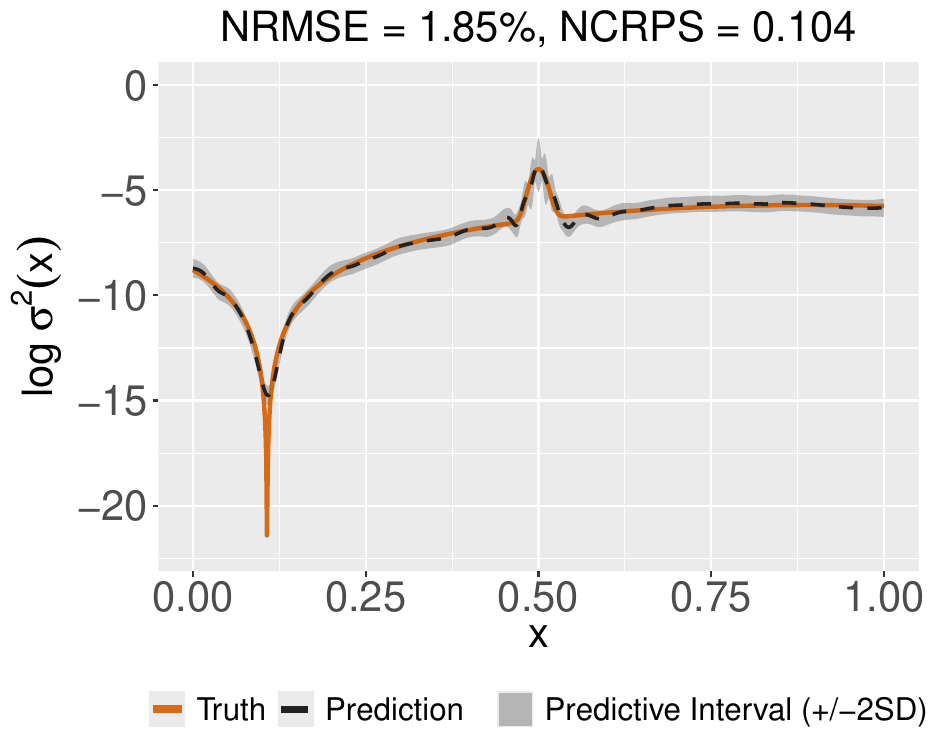}}
\hspace{3em}
\subfloat[3 layers]{\label{fig:figure_s1d}\includegraphics[width=0.37\linewidth]{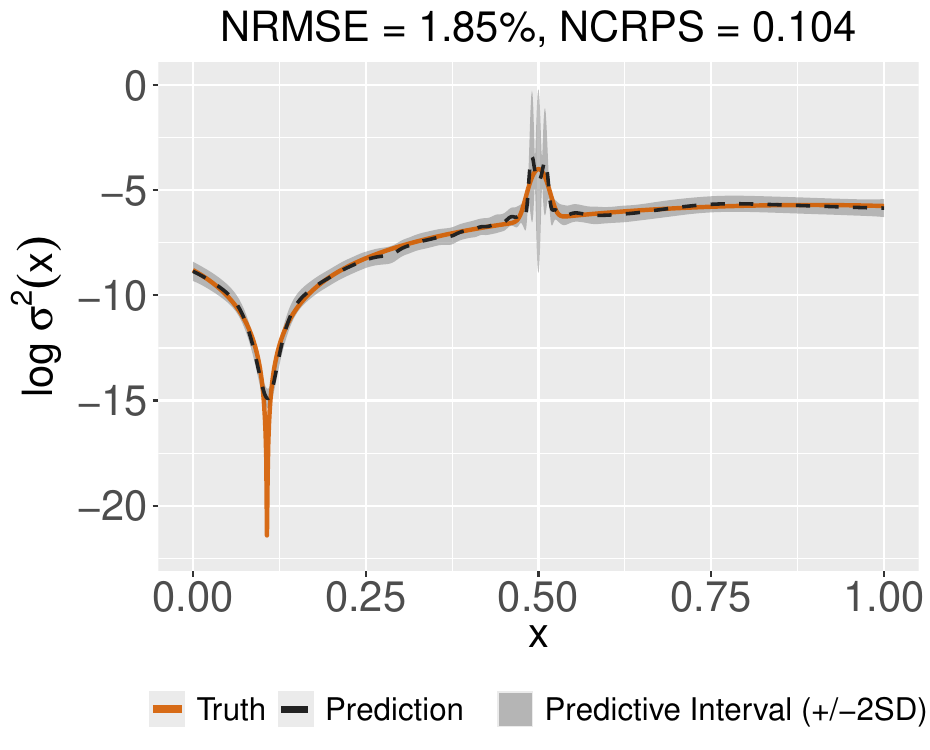}}
\caption{Emulations of the mean (top row) and log-variance functions (bottom row) of the synthetic simulator described in Section~\ref{sec:1d-het} of the main manuscript, obtained using GDGPs with two- and three-layer DGPs. NRMSE denotes the normalized root mean squared error, and NCRPS denotes the negative continuous ranked probability score; lower values indicate better performance for both metrics. Both metrics are evaluated using $1{,}000$ validation points over $[0,1]$.}
\label{fig:figure_s1}
\end{figure}

\begin{figure}[H]
\centering 
\subfloat[NRMSE of $\mu(x)$]{\label{fig:figure_s2a}\includegraphics[width=0.5\linewidth]{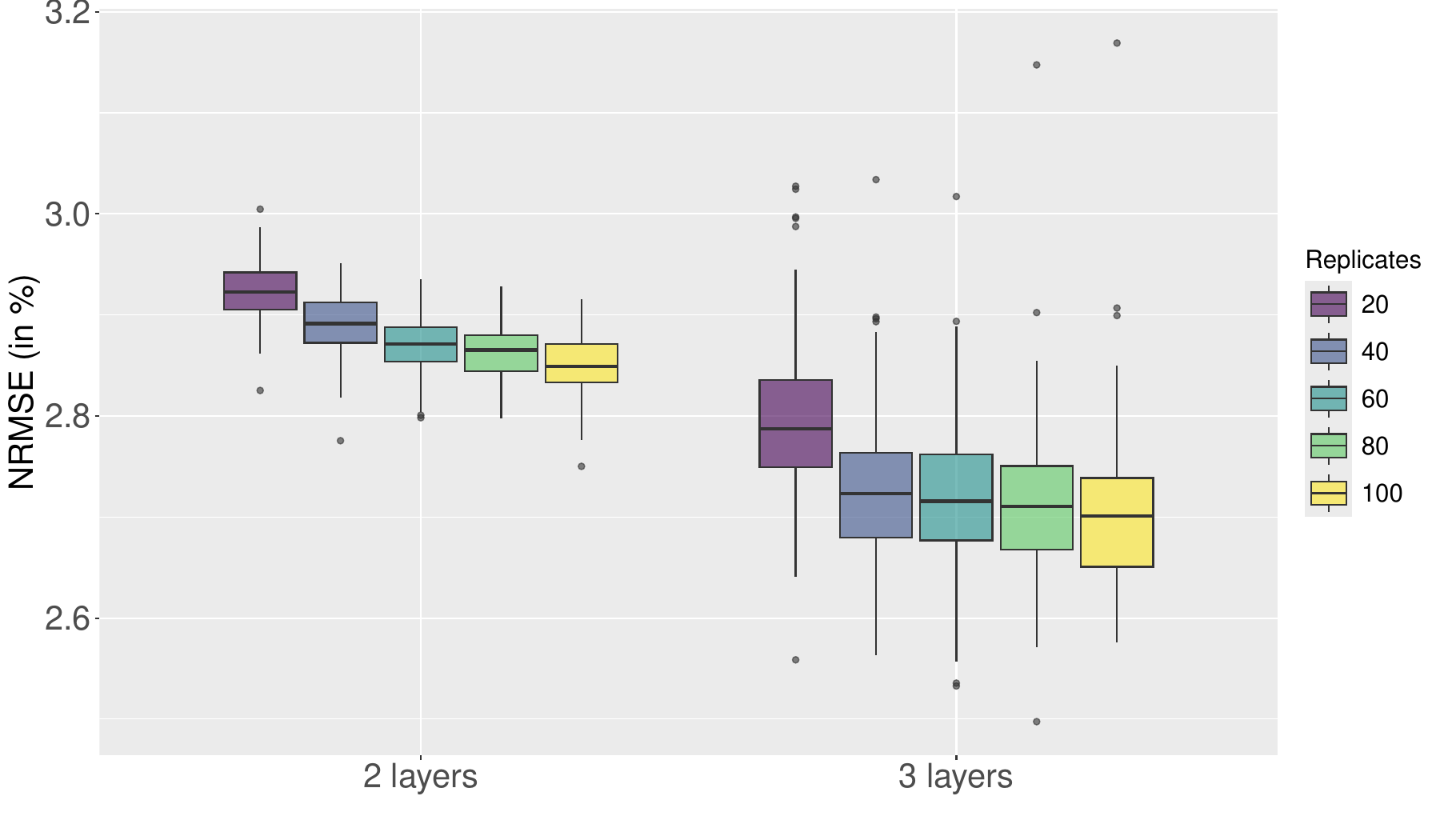}}
\subfloat[NCRPS of $\mu(x)$]{\label{fig:figure_s2b}\includegraphics[width=0.5\linewidth]{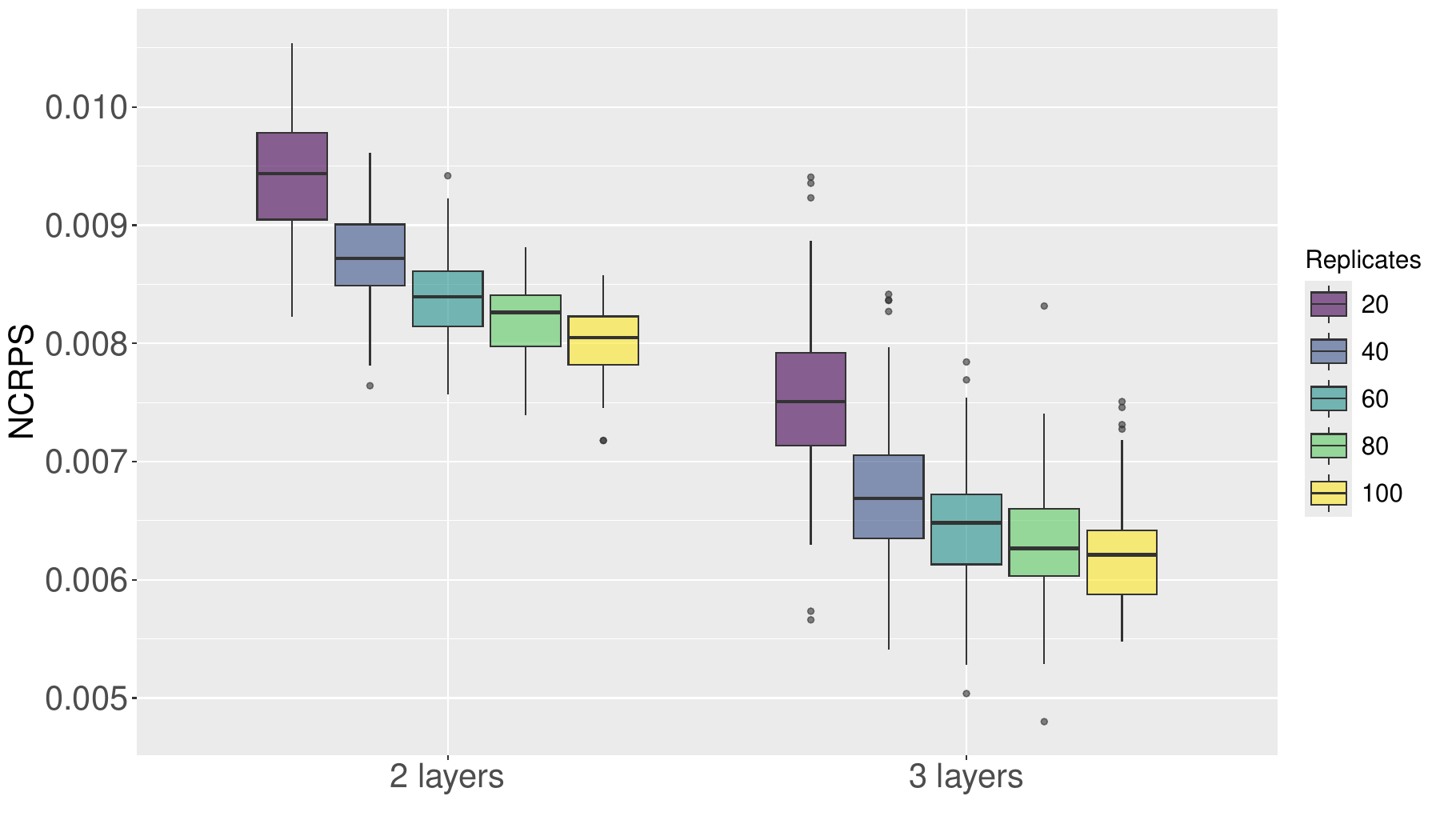}}\\
\subfloat[NRMSE of $\log\sigma^2(x)$]{\label{fig:figure_s2c}\includegraphics[width=0.5\linewidth]{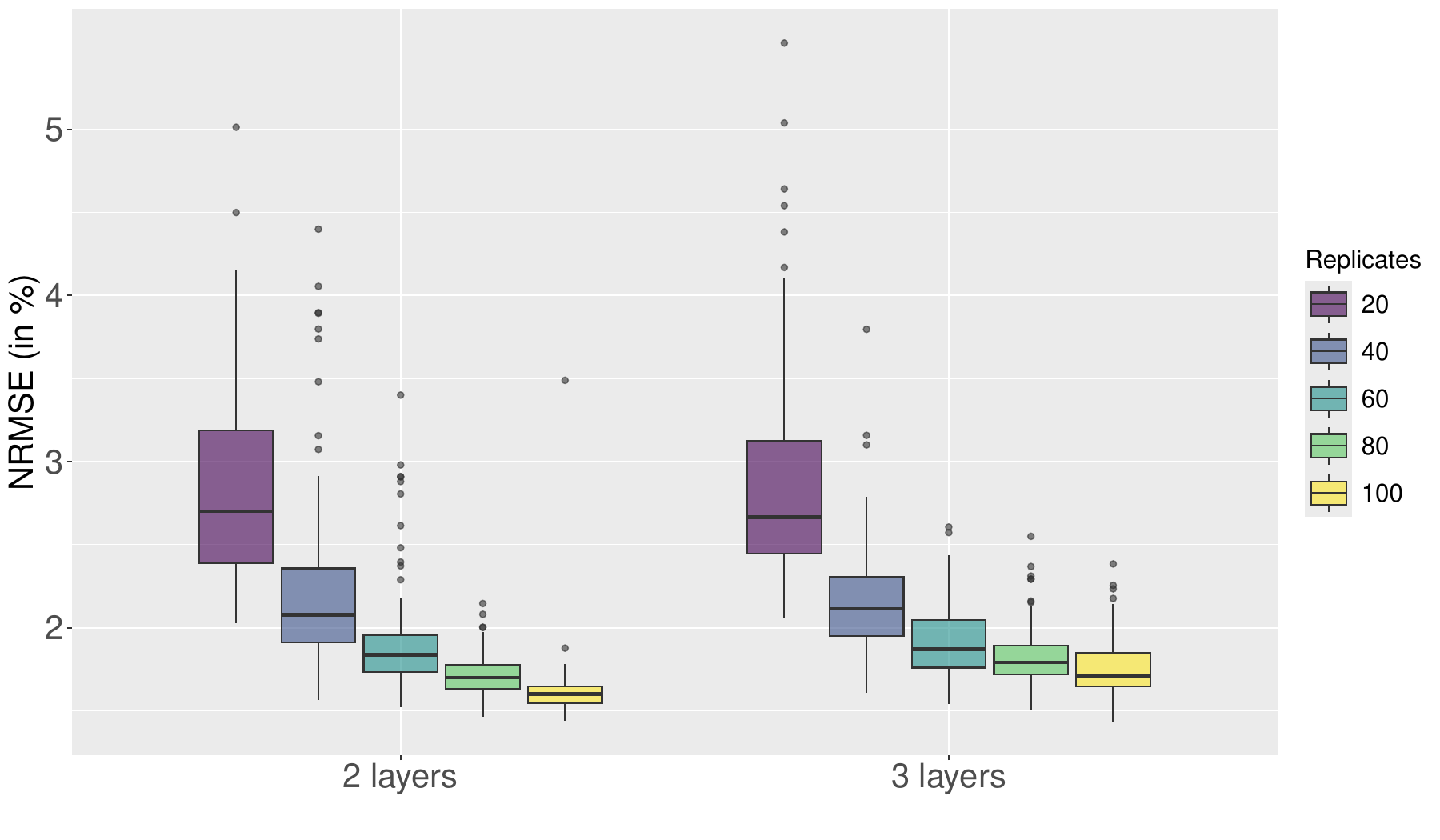}}
\subfloat[NCRPS of $\log\sigma^2(x)$]{\label{fig:figure_s2d}\includegraphics[width=0.5\linewidth]{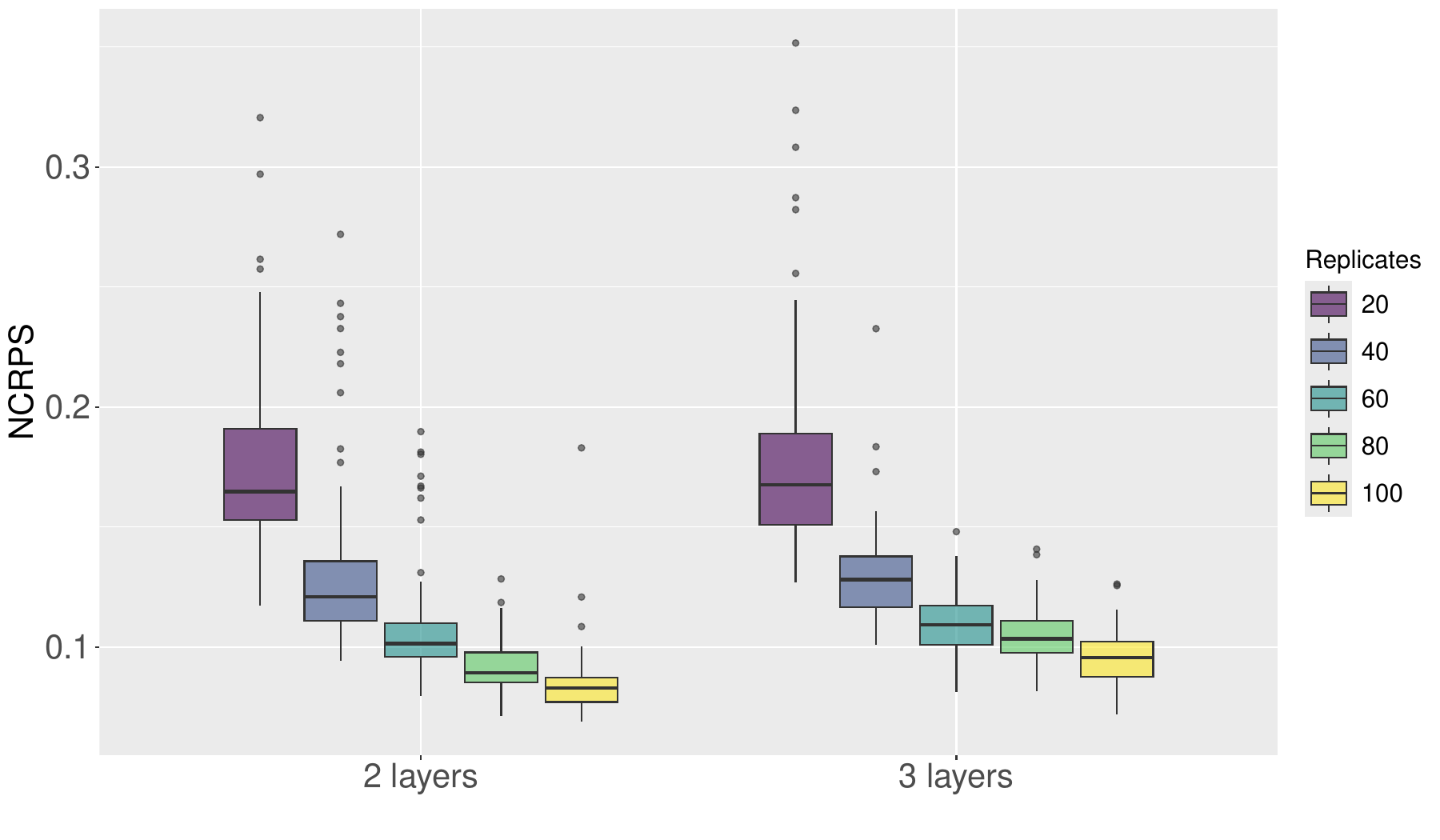}}
\caption{NRMSEs (lower is better) and NCRPSs (lower is better) for GDGPs using two- and three-layer DGP structures, trained on $100$ unique input locations with $R\in\{20,40,60,80,100\}$ replicates per location, for emulating the mean function $\mu(x)$ and log-variance function $\log \sigma^2(x)$ of the synthetic simulator described in Section~\ref{sec:1d-het} of the main manuscript. Results are evaluated using $1{,}000$ validation points and summarized over $100$ independent training trials.}
\label{fig:figure_s2}
\end{figure}

\begin{table}[H]
\caption{Mean computation times (in seconds), averaged over $100$ independent training trials, for GDGPs using two- and three-layer DGP structures, trained on $100$ unique input locations with $R\in\{20,40,60,80,100\}$ replicates per location from the synthetic simulator described in Section~\ref{sec:1d-het} of the main manuscript.}
\centering
\begin{booktabs}{
  colspec = {Q[c,m] Q[r,m] Q[r,m] Q[r,m] Q[r,m] Q[r,m]},
  cell{1}{1} = {r=2}{c,m,font=\bfseries},
  cell{1}{2} = {c=5}{c,m,font=\bfseries},
  cell{2}{2-6} = {c,m}
}
\toprule
\shortstack{DGP\\Layers} & Number of Replicates & & & & \\
\cmidrule[lr]{2-6}
 & 20 & 40 & 60 & 80 & 100 \\
\midrule
2 & 12.78 & 13.43 & 13.81 & 13.56 & 14.96 \\
3 & 14.80 & 14.92 & 15.70 & 16.81 & 22.21 \\
\bottomrule
\end{booktabs}
\end{table}

\subsection{Susceptible, Infective, and Recovered (SIR) Model}

\begin{figure}[H]
\centering 
\subfloat[]{\label{fig:figure_s3a}\includegraphics[width=0.5\linewidth]{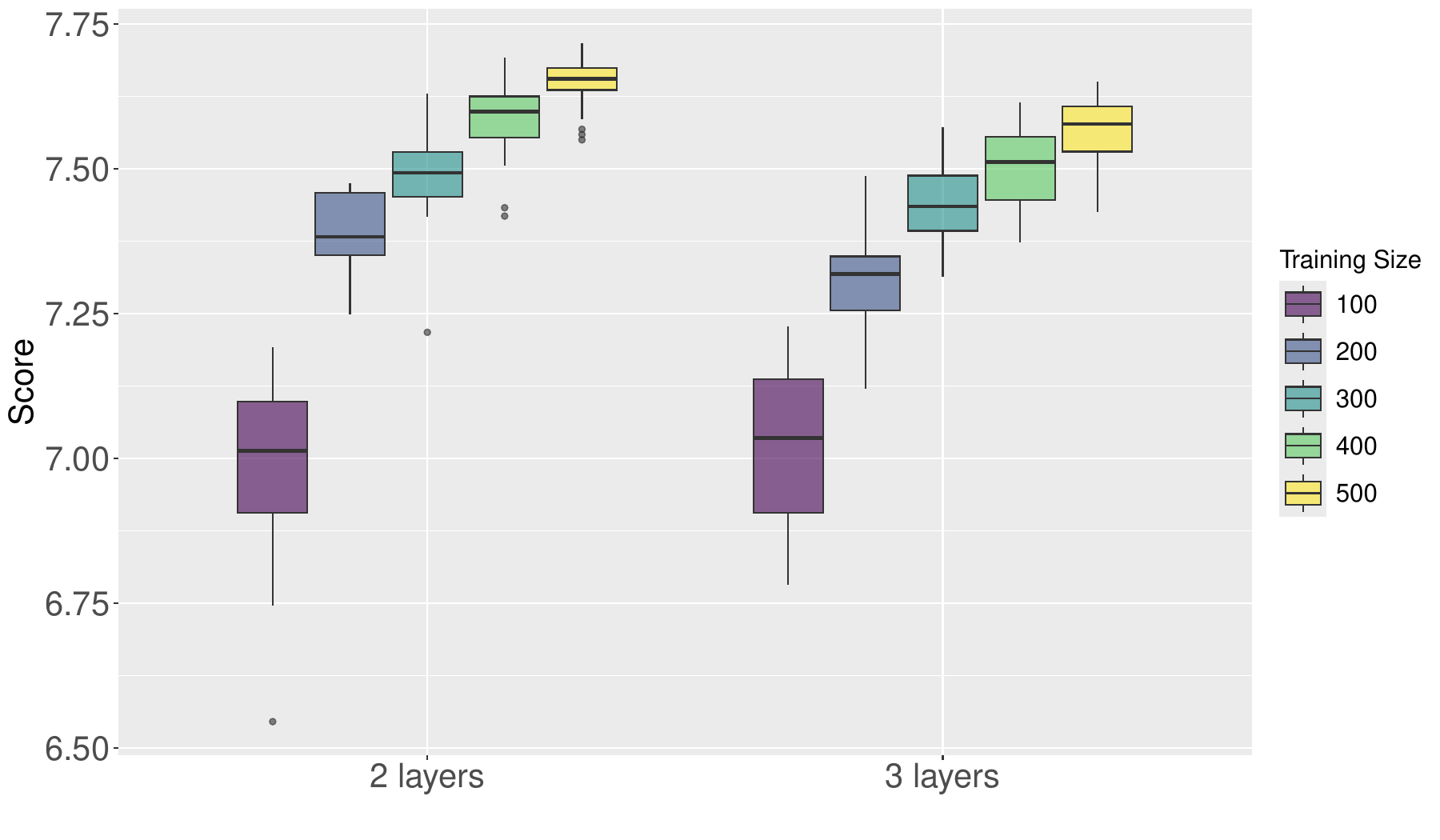}}
\subfloat[]{\label{fig:figure_s3b}\includegraphics[width=0.5\linewidth]{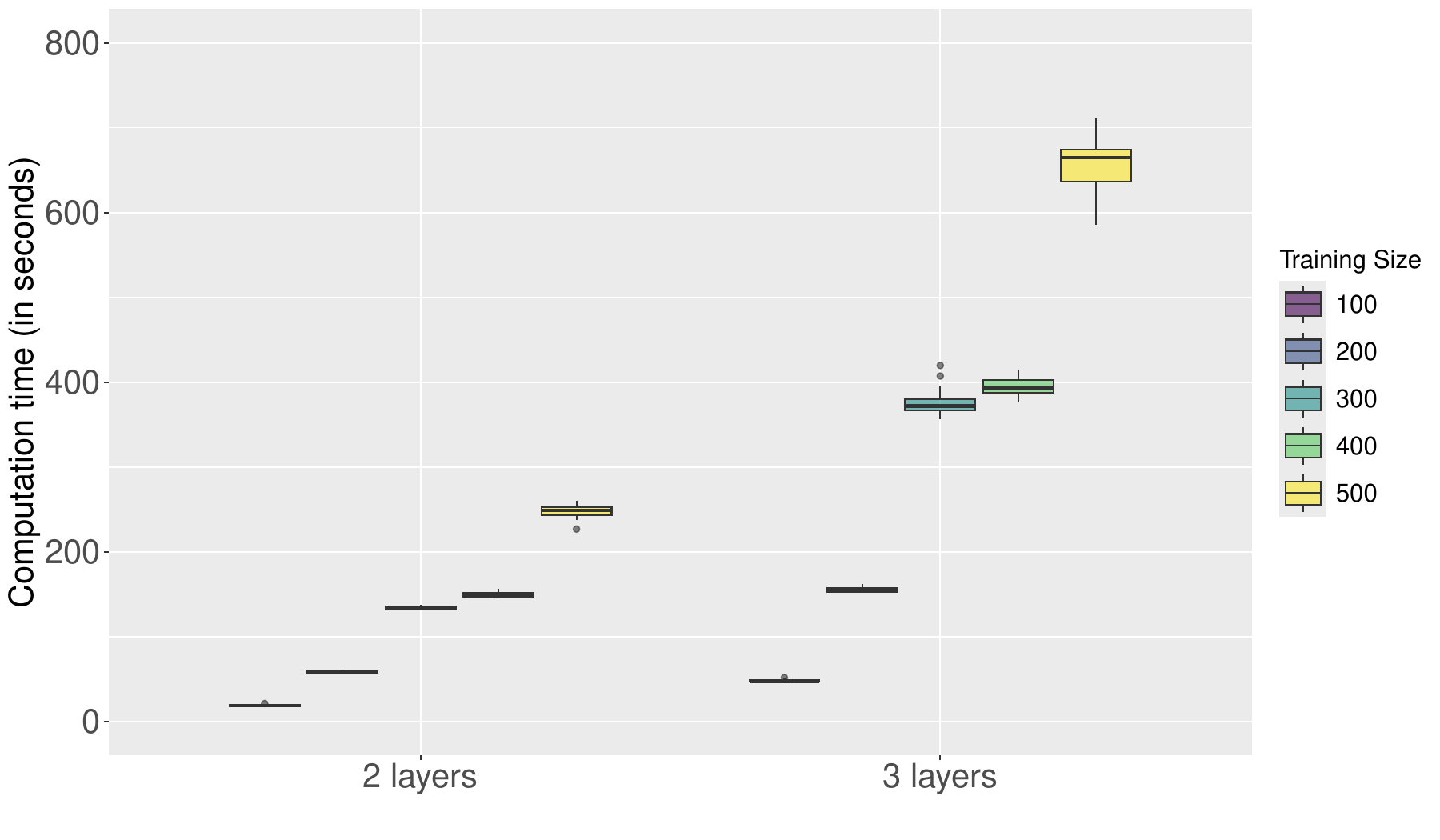}}
\caption{\emph{(a)} Scores (higher is better) and \emph{(b)} computation times (in seconds) for GDGPs using two- and three-layer DGP structures, trained on $n\in\{100,200,300,400,500\}$ unique input locations, with a random number of replicates between $1$ and $100$ at each location, for emulating the cumulative attack proportion of the SIR simulator in Section~\ref{sec:sir} of the main manuscript at a time horizon of $100$. Scores are evaluated on a test set of size $N=60{,}000$, consisting of $2{,}000$ space-filling input locations with $30$ replicates at each location. Both scores and computation times are summarized over $20$ independent training trials for each $n$.}
\label{fig:figure_s3}
\end{figure}

\section{Categorical Likelihood}

\begin{table}[H]
\caption{Means and standard deviations (in parentheses) of accuracy (higher is better; see Equation~\eqref{eq:accuracy} of the main manuscript), logloss (lower is better; see Equation~\eqref{eq:logloss} of the main manuscript), and computation time (in seconds), computed across the $20$ random partitions of the datasets listed in Table~\ref{tab:tab_2} of the main manuscript, for GDGP and Vecchia-approximated GDGP (vGDGP) using two- and three-layer DGP structures. The highest mean accuracy and lowest mean logloss for each dataset are highlighted in bold.}
\centering
\begin{booktabs}{
  colspec = {Q[l,m] Q[c,m] Q[l,m] Q[c,m] Q[c,m] Q[c,m] Q[c,m]},
  cell{1}{1} = {r=2}{c,m,font=\bfseries},
  cell{1}{2} = {r=2}{c,m,font=\bfseries},
  cell{1}{3} = {r=2}{c,m,font=\bfseries},
  cell{1}{4} = {c=4}{c,m,font=\bfseries},
  cell{2}{4-7} = {c,m},
  cell{3}{1} = {r=6}{l,m},
  cell{9}{1} = {r=6}{l,m},
  cell{3}{2} = {r=3}{c,m},
  cell{6}{2} = {r=3}{c,m},
  cell{9}{2} = {r=3}{c,m},
  cell{12}{2} = {r=3}{c,m},
}
\toprule
Model & \shortstack{DGP\\Layers} & Metric & Dataset &  &  &  \\
\cmidrule[lr]{4-7}
 &  &  & iris & pima & thyroid & vehicle \\
\midrule
\midrule
GDGP & 2 & Accuracy & \textbf{96.00 (4.63)} & \textbf{73.27 (4.55)} & \textbf{98.22 (3.72)} & \textbf{80.14 (3.19)} \\
 &  & Logloss & \textbf{0.1153 (.0712)} & \textbf{0.5262 (.0560)} & \textbf{0.0717 (.0469)} & \textbf{0.4353 (.0554)} \\
 &  & Time & 26.25 (1.63) & 280.48 (9.38) & 53.86 (2.09) & 1182.20 (44.44) \\
\cmidrule[lr]{2-7}
 & 3 & Accuracy & 94.33 (4.79) & 72.65 (5.25) & 96.44 (5.17) & 77.70 (4.30) \\
 &  & Logloss & 0.1240 (.0647) & 0.5315 (.0560) & 0.1208 (.0970) & 0.4644 (.0526) \\
 &  & Time & 77.95 (5.93) & 1540.99 (78.70) & 178.56 (35.33) & 5166.79 (393.30) \\
\midrule
\midrule
vGDGP & 2 & Accuracy & 95.00 (5.82) & 71.38 (4.53) & 95.00 (5.67) & 78.07 (3.65) \\
 &  & Logloss & 0.1489 (.1028) & 0.5438 (.0489) & 0.1268 (.0638) & 0.4830 (.0596) \\
 &  & Time & 67.53 (9.15) & 127.27 (2.75) & 92.97 (3.32) & 798.62 (19.51) \\
\cmidrule[lr]{2-7}
 & 3 & Accuracy & 94.08 (6.01) & 69.82 (3.54) & 95.00 (6.91) & 77.47 (4.16) \\
 &  & Logloss & 0.1810 (.1434) & 0.5778 (.0576) & 0.1355 (.1236) & 0.4861 (.0558) \\
 &  & Time & 208.55 (10.66) & 1096.91 (46.50) & 329.74 (54.06) & 4348.44 (79.31) \\
\bottomrule
\end{booktabs}
\end{table}

\section{Count Likelihoods}

\begin{figure}[H]
\centering 
\subfloat[Poisson (2 layers)]{\label{fig:figure_s4a}\includegraphics[width=0.39\linewidth]{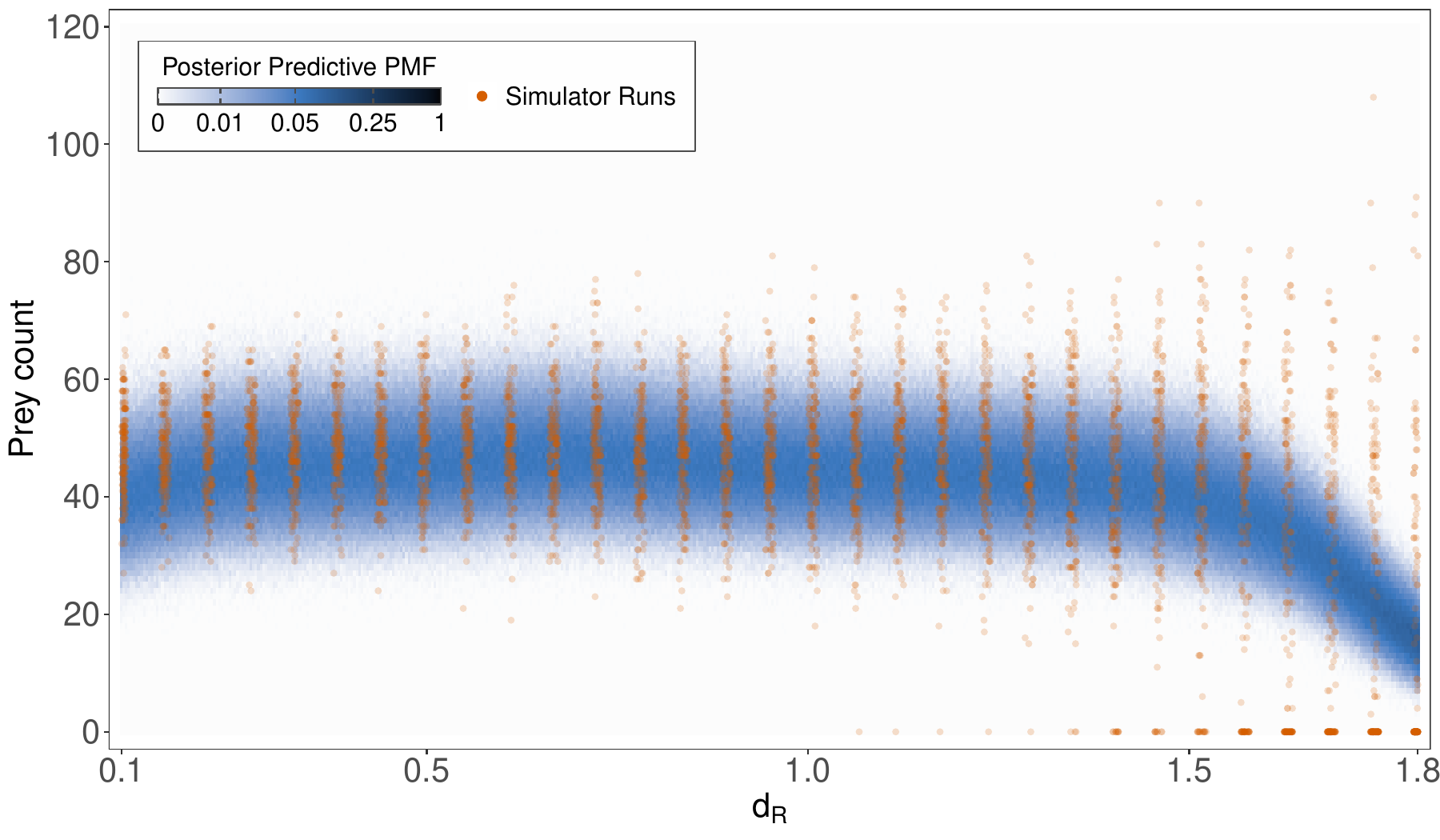}}
\hspace{3em}
\subfloat[Poisson (3 layers)]{\label{fig:figure_s4b}\includegraphics[width=0.39\linewidth]{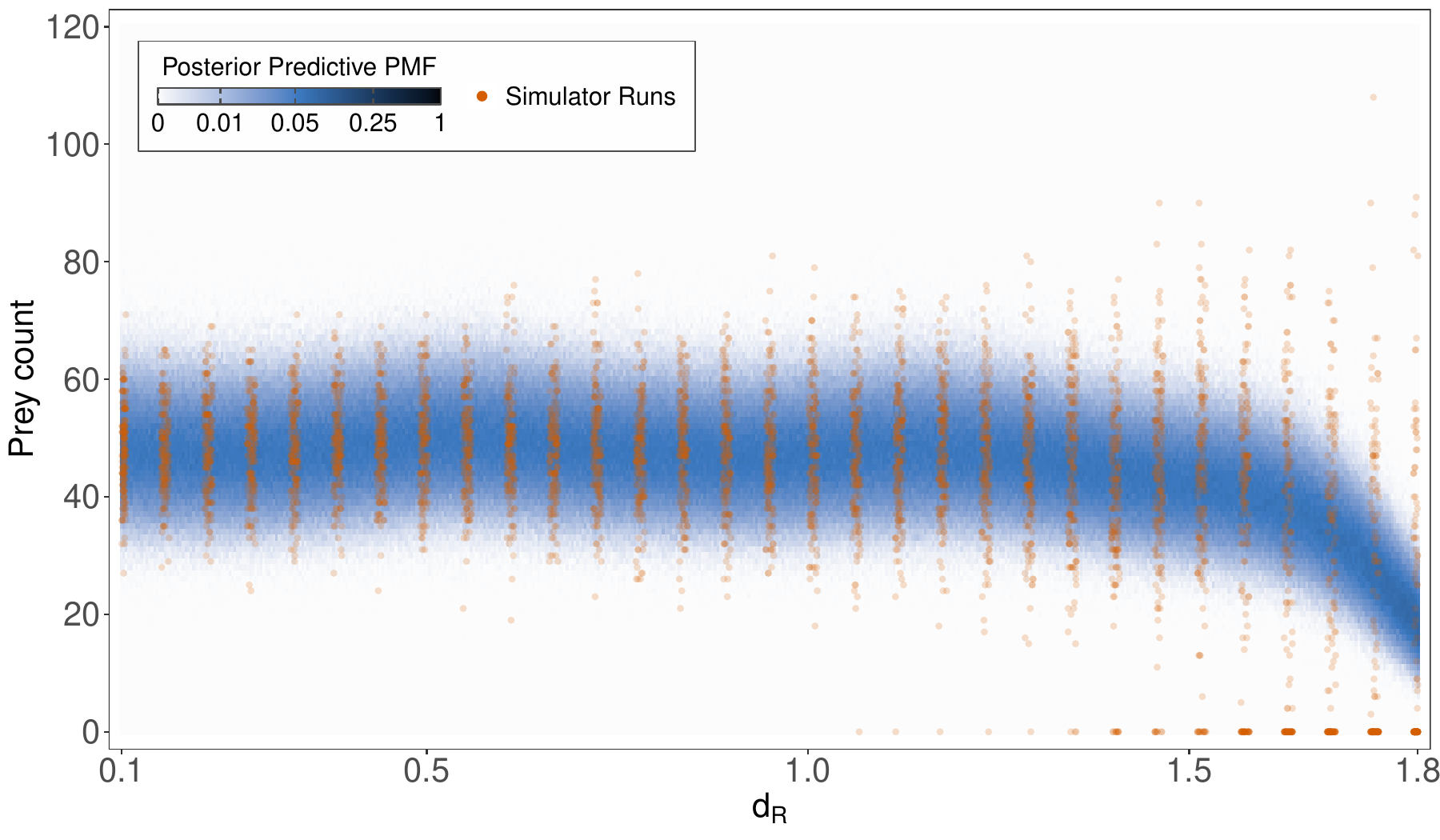}}\\
\subfloat[ZIP (2 layers)]{\label{fig:figure_s4c}\includegraphics[width=0.39\linewidth]{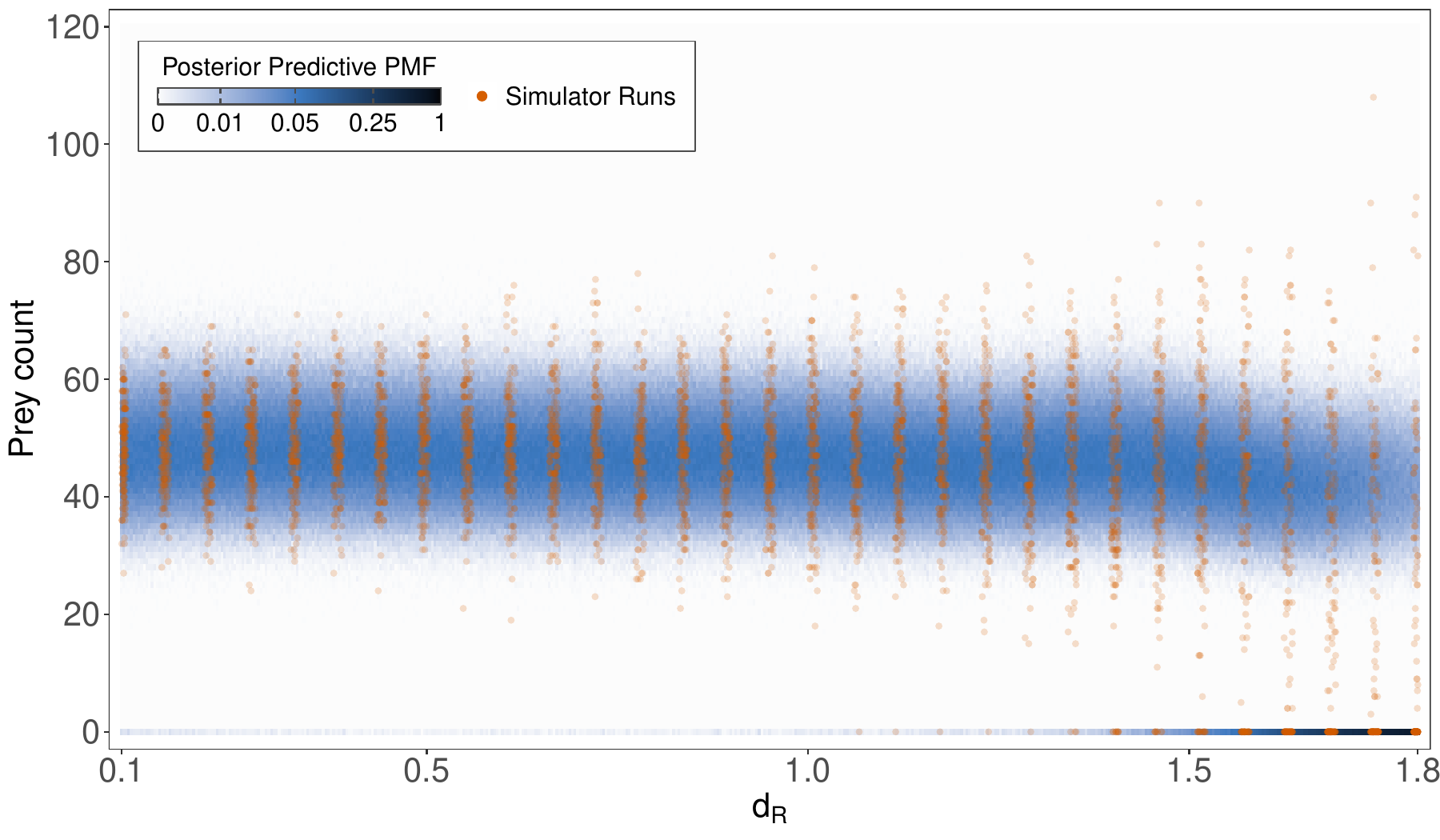}}
\hspace{3em}
\subfloat[ZIP (3 layers)]{\label{fig:figure_s4d}\includegraphics[width=0.39\linewidth]{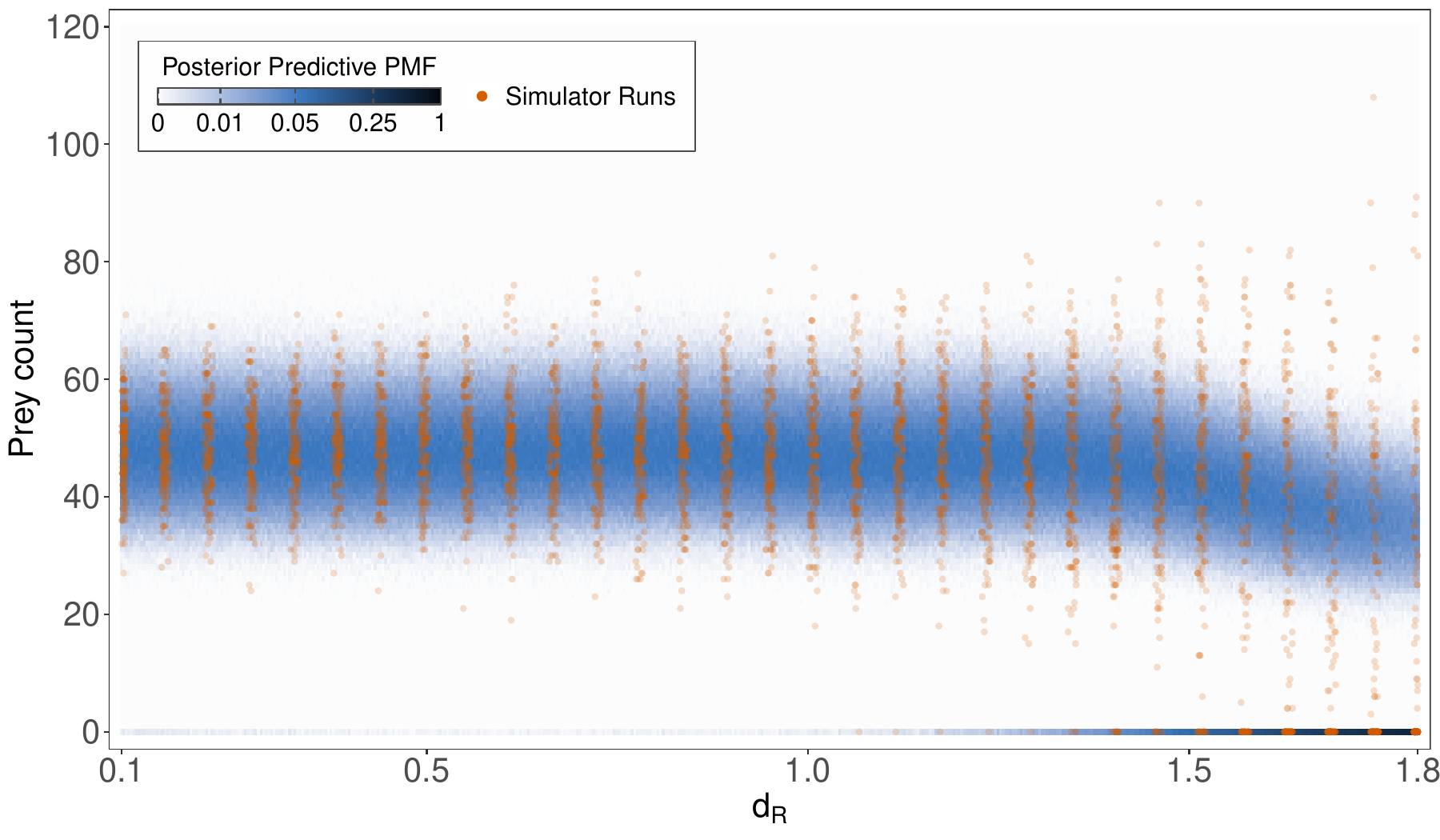}}\\
\subfloat[NB (2 layers)]{\label{fig:figure_s4e}\includegraphics[width=0.39\linewidth]{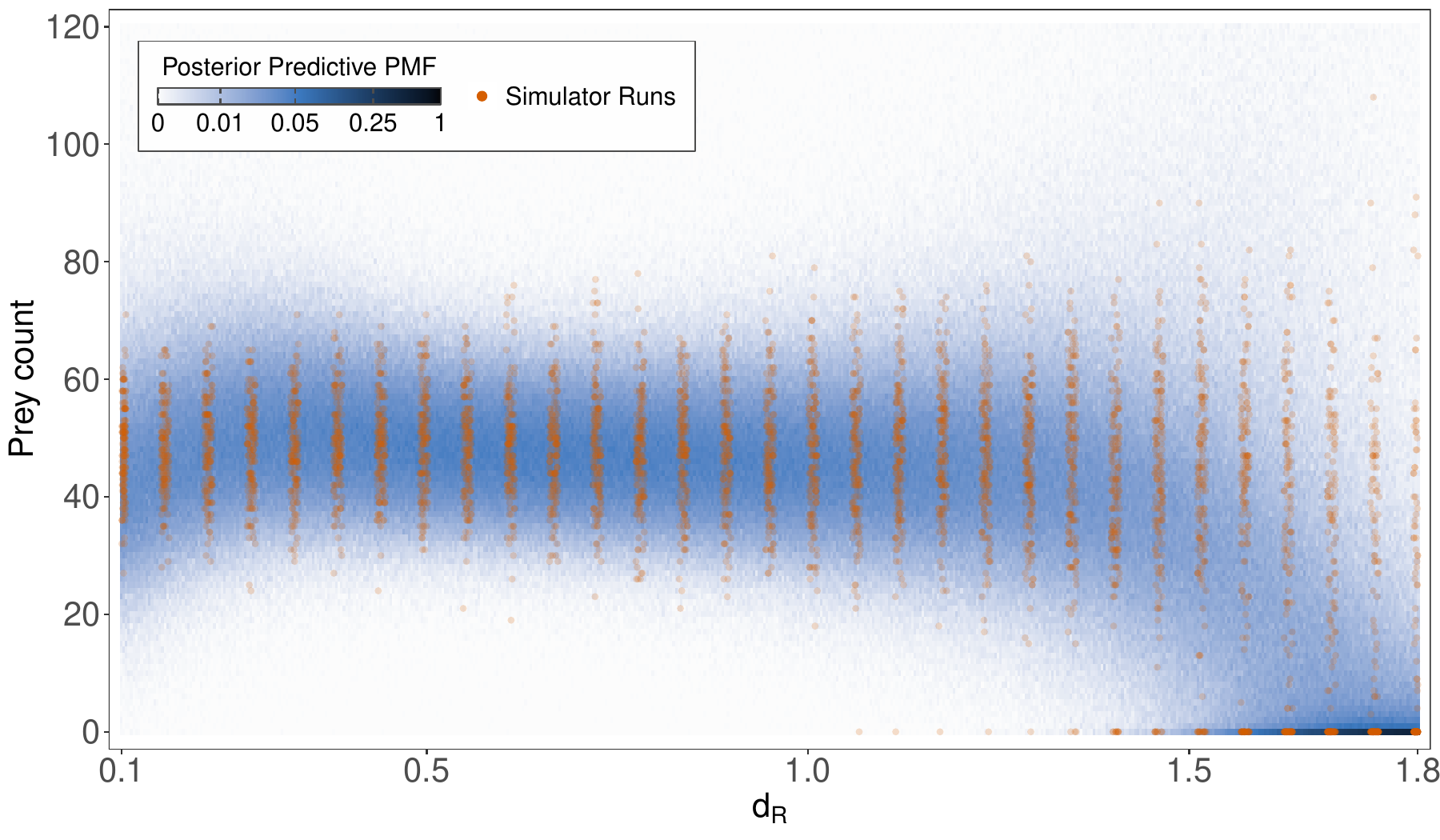}}
\hspace{3em}
\subfloat[NB (3 layers)]{\label{fig:figure_s4f}\includegraphics[width=0.39\linewidth]{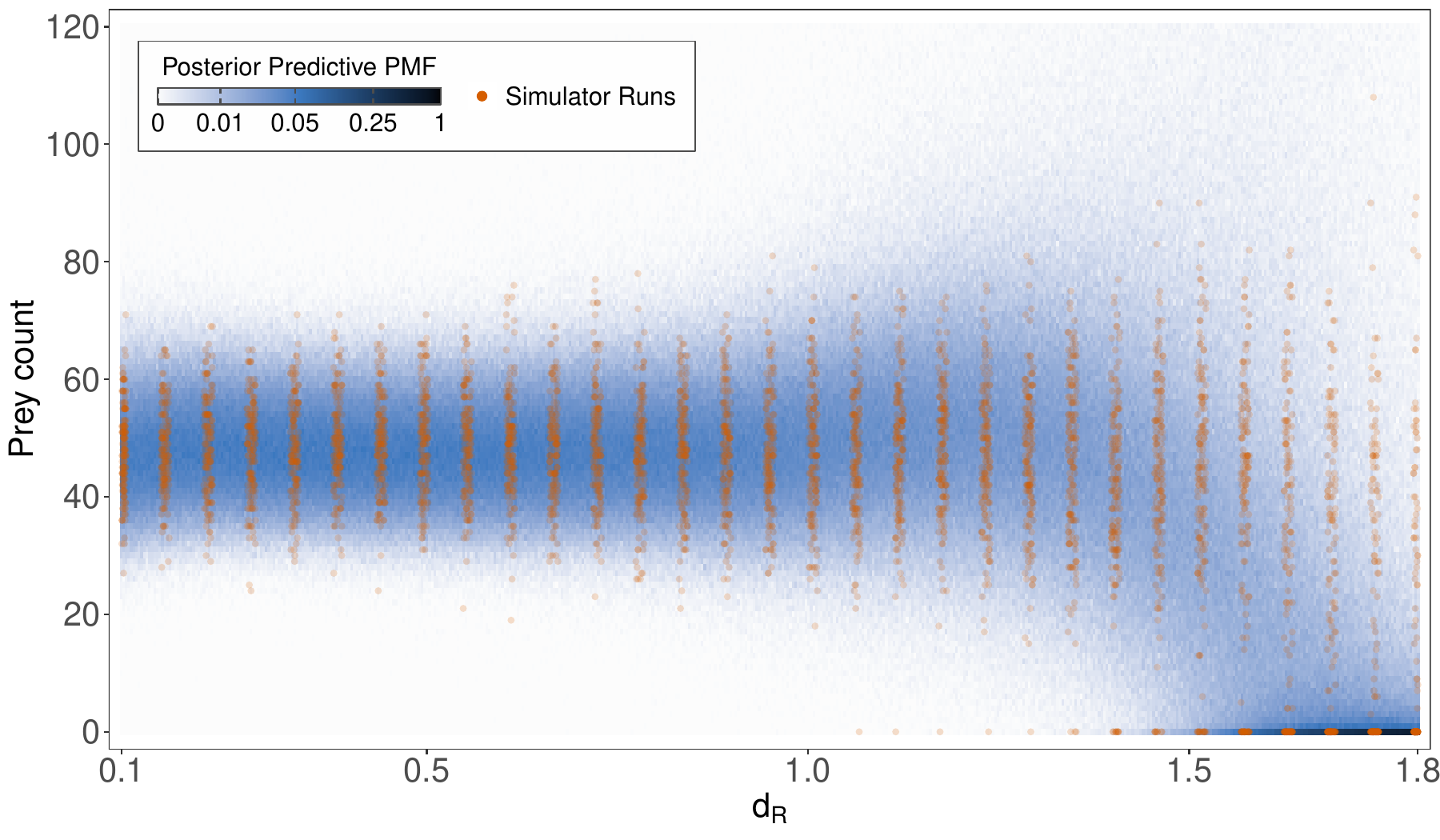}}\\
\subfloat[ZINB (2 layers)]{\label{fig:figure_s4g}\includegraphics[width=0.39\linewidth]{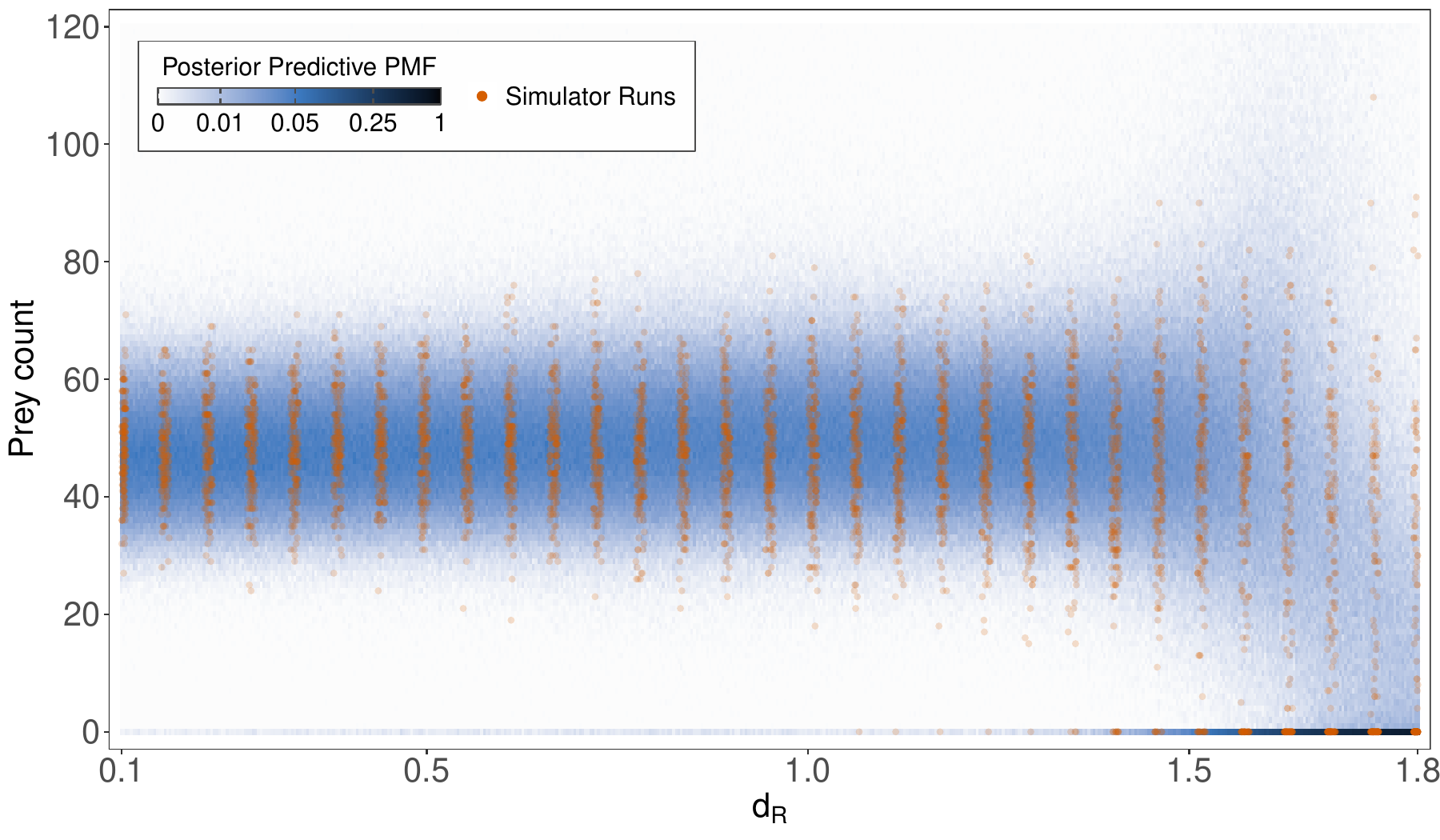}}
\hspace{3em}
\subfloat[ZINB (3 layers)]{\label{fig:figure_s4h}\includegraphics[width=0.39\linewidth]{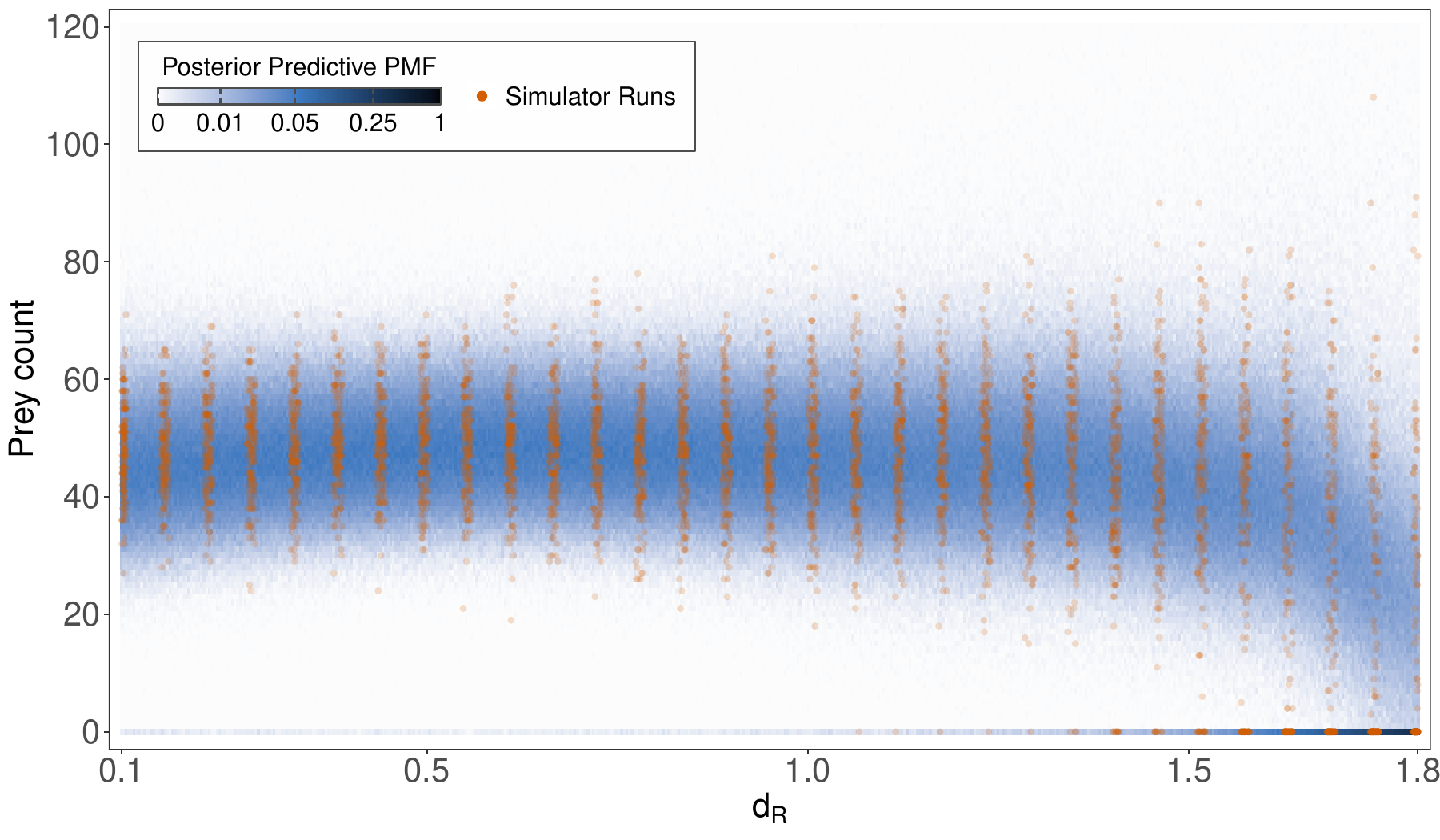}}
\caption{One-dimensional slices of the posterior predictive probability mass functions (PMFs) for the prey population count from GDGPs using two- and three-layer DGP structures with Poisson, zero-inflated Poisson (ZIP), negative binomial (NB), and zero-inflated negative binomial (ZINB) likelihoods. The slices vary $d_R$ over $[0.1,1.8]$ while fixing $d_F=1.050$, $\alpha=0.015$, and $w=0.004$. The shading represents the posterior predictive probability mass, with darker shades indicating greater mass. The points show the empirical distributions of prey counts obtained from $150$ simulator runs at each of $31$ evenly spaced values of $d_R$ and are horizontally jittered for visibility.}
\label{fig:figure_s4}
\end{figure}

\begin{figure}[H]
\centering 
\subfloat[]{\label{fig:figure_s5a}\includegraphics[width=0.48\linewidth]{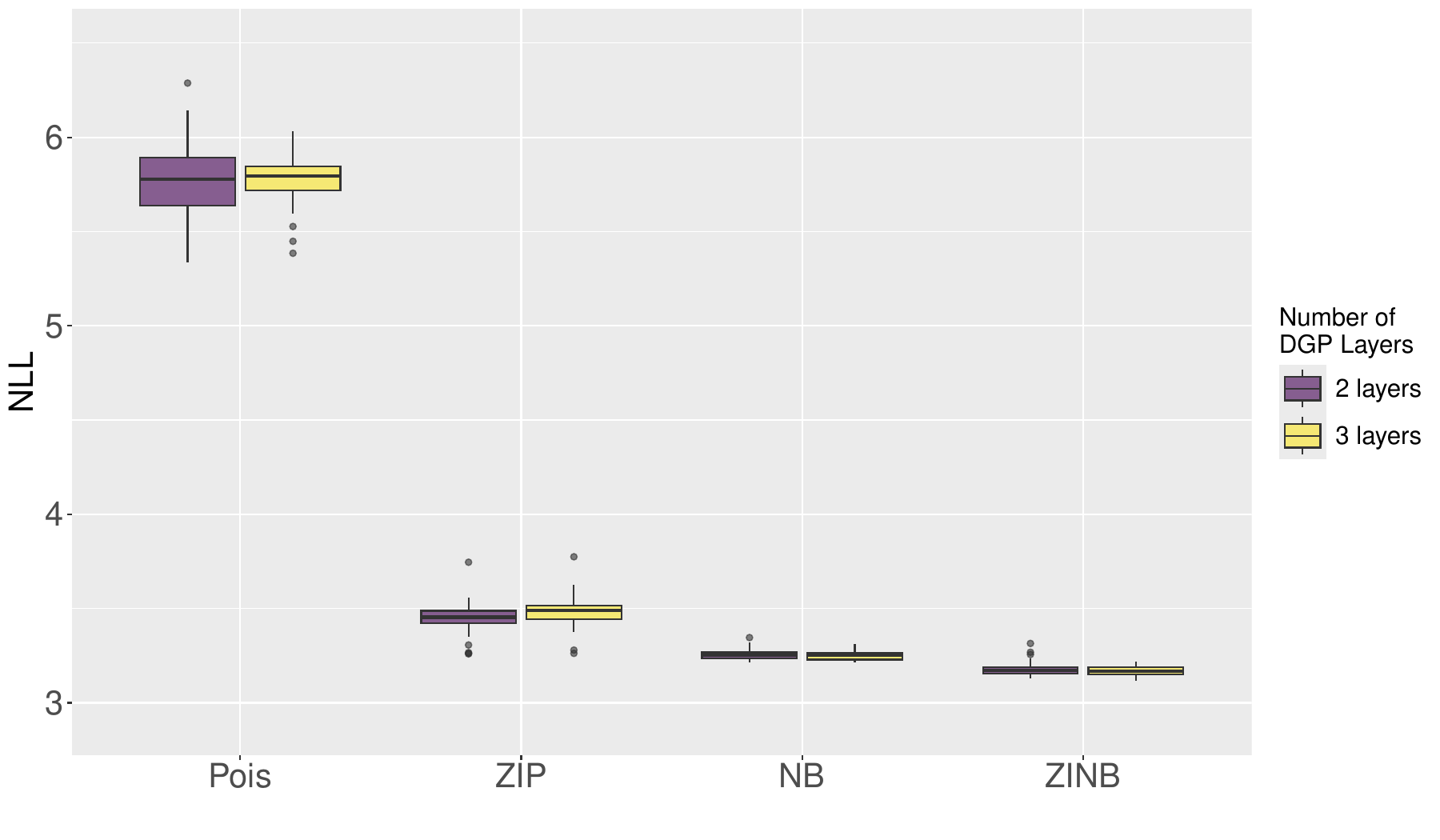}}
\hspace{0.5em}
\subfloat[]{\label{fig:figure_s5b}\includegraphics[width=0.48\linewidth]{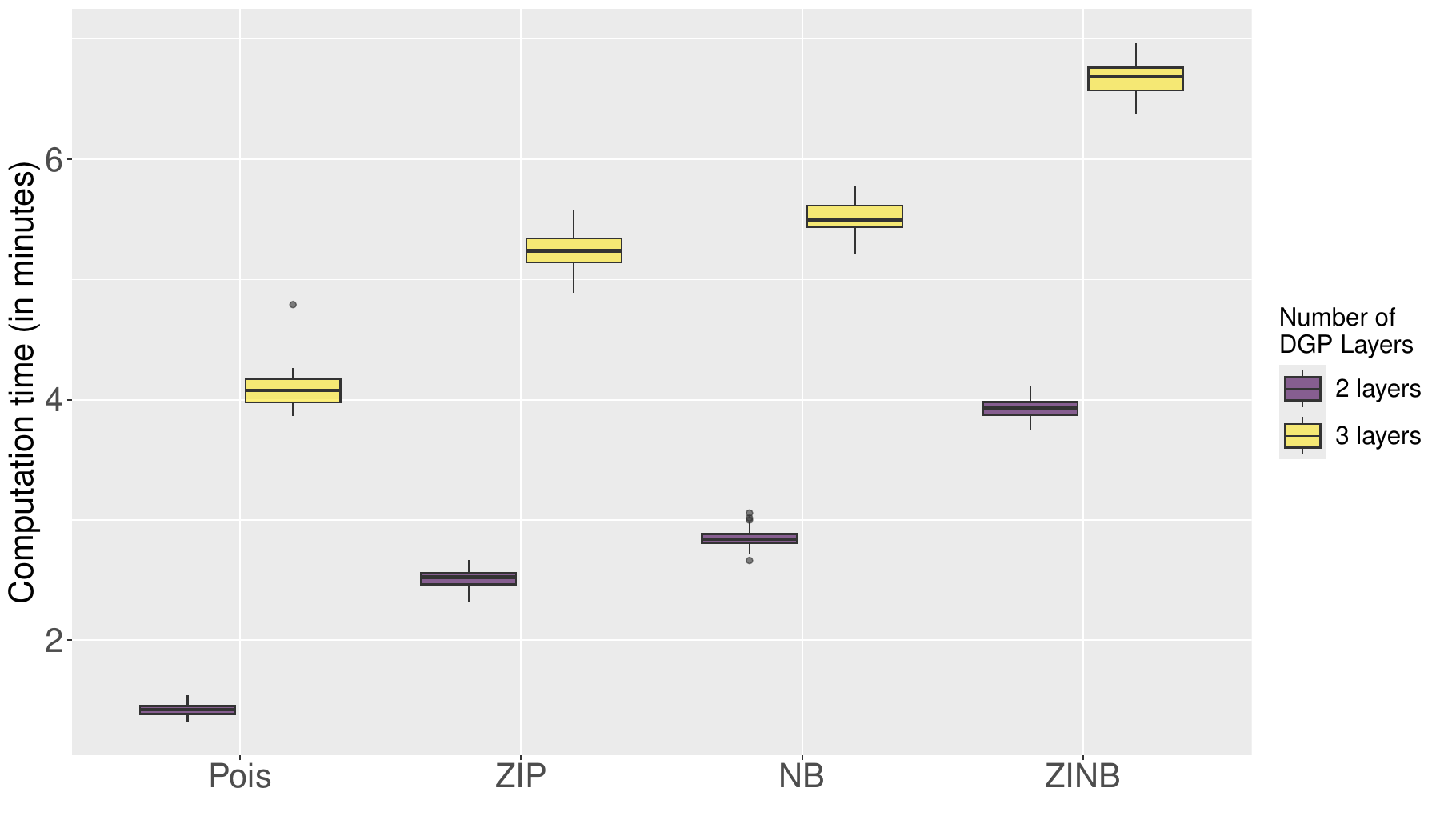}}
\caption{\emph{(a)} Negative log-likelihood (NLL; lower is better) and \emph{(b)} computation time (in minutes) for  GDGPs using two- and three-layer DGP structures with Poisson (Pois), zero-inflated Poisson (ZIP), negative binomial (NB), and zero-inflated negative binomial (ZINB) likelihoods, trained on $300$ unique input locations with up to $100$ replicates per location, for emulating the prey population count of the predator-prey simulator described in Section~\ref{sec:count} of the main manuscript. NLLs are evaluated on a test set of size $N=50{,}000$, consisting of $1{,}000$ space-filling input locations with $50$ replicates per location, and summarized over $50$ independent training trials.}
\label{fig:figure_s5}
\end{figure}

\begin{table}[H]
\caption{Mean negative log-likelihoods (NLLs; lower is better) and computation times (in seconds), calculated over $50$ independent training trials, for GDGPs using two- and three-layer DGP structures with Poisson (Pois), zero-inflated Poisson (ZIP), negative binomial (NB), and zero-inflated negative binomial (ZINB) likelihoods in the predator-prey experiment described in Section~\ref{sec:count} of the main manuscript.}
\centering
\begin{booktabs}{
  colspec = {Q[l,m] Q[c,m] Q[r,m] Q[r,m]},
  cell{1}{1} = {r=2}{c,m,font=\bfseries},
  cell{1}{3} = {r=2}{c,m,font=\bfseries},
  cell{1}{4} = {r=2}{c,m,font=\bfseries},
  cell{1}{2} = {c,m,font=\bfseries},
  cell{2}{2} = {c,m,font=\bfseries},
  cell{3}{1} = {r=2}{c,m},
  cell{5}{1} = {r=2}{c,m},
  cell{7}{1} = {r=2}{c,m},
  cell{9}{1} = {r=2}{c,m},
}
\toprule
Likelihood & DGP & NLL & Time(s) \\
           & Layers &     &         \\
\midrule
Pois & 2 & 5.78 & 85.46 \\
 & 3 & 5.78 & 245.07 \\
\midrule
ZIP & 2 & 3.45 & 150.58 \\
 & 3 & 3.48 & 314.81 \\
\midrule
NB & 2 & 3.25 & 171.11 \\
 & 3 & 3.25 & 331.45 \\
\midrule
ZINB & 2 & 3.18 & 235.86 \\
 & 3 & 3.17 & 400.70 \\
\bottomrule
\end{booktabs}
\end{table}
\end{document}